\documentclass[aps,prl,reprint,nobalancelastpage,nofootinbib]{revtex4-2}

\usepackage[utf8]{inputenc}

\usepackage{amsmath}
\usepackage{amssymb}
\usepackage[all]{xy}
\usepackage{tikz}
\usetikzlibrary{calc}

\usepackage[toc,page]{appendix}
\usepackage{enumerate}
\usepackage{tensor}
\usepackage{booktabs}
\usepackage{bbm}
\usepackage{youngtab}
\usepackage{slashed}
\usepackage{multirow}
\usepackage{makecell}
\usepackage{float}
\usepackage{comment}
\usepackage{verbatim}
\usepackage{xcolor}
\usepackage[framemethod=tikz]{mdframed}
\usepackage{soul}
\usepackage{mathtools}
\usepackage[inline]{enumitem}

\usepackage{amsthm}
\newtheorem{theorem}{Theorem}
\newtheorem{lemma}{Lemma}
\newtheorem{corollary}{Corollary}[lemma]
\newtheorem{remark}{Remark}[lemma]

\newcommand{\ab}{|}
\newcommand{\der}{\partial}
\newcommand{\de}{\mathrm{d}}

\newcommand{\vevc}{\chi}
\newcommand{\sdil}{\phi}

\newcommand{\tsdil}{\tilde{\phi}}
\newcommand{\tomega}{\tilde{\omega}}
\newcommand{\tdelta}{\tilde{\delta}}
\newcommand{\tsigma}{\tilde{\sigma}}
\newcommand{\tvarphi}{\tilde{\varphi}}
\newcommand{\ttau}{\tilde{\tau}}

\newcommand{\e}{\mathrm{e}}
\newcommand{\I}{\mathrm{i}}

\newcommand{\p}{\mathrm{P}}

\newcommand{\x}{\mathrm{x}}
\newcommand{\y}{\mathrm{y}}
\newcommand{\bx}{\overline{\x}}
\newcommand{\by}{\overline{\y}}
\newcommand{\bc}{\overline{\mathrm{c}}}
\newcommand{\hx}{\hat{\x}}
\newcommand{\hy}{\hat{\y}}

\newcommand{\tvec}[2]{\left(\!\!\begin{array}{c}
    #1 \\
    #2
\end{array}\!\!\right)}

\newcommand{\matr}[4]{\left(\!\!\!\begin{array}{cc}
    #1 & #2 \\
    #3 & #4
\end{array}\!\!\!\right)}

\usepackage{hyperref}
\hypersetup{
    colorlinks=true,
    linkcolor=cyan!60!blue,
    citecolor=cyan!60!blue,
    filecolor=cyan!60!blue,      
    urlcolor=cyan!60!blue,
}

\begin{document}
\numberwithin{equation}{section}

\title{Late-time attractors and cosmic acceleration}
\author{Gary Shiu}
\email{shiu@physics.wisc.edu}
\author{Flavio Tonioni}
\email{tonioni@wisc.edu}
\affiliation{Department of Physics, University of Wisconsin-Madison, 1150 University Avenue, Madison, WI 53706, USA}
\author{Hung V. Tran}
\email{hung@math.wisc.edu}
\affiliation{Department of Mathematics, University of Wisconsin-Madison, 480
Lincoln Drive, Madison, WI 53706, USA}

\begin{abstract}
We prove the conditions under which scaling cosmologies are inevitable late-time attractors of multi-field multi-exponential potentials, independently of initial conditions. The advantage of such scaling cosmologies is that the time dependence of the fields and of the scale factor is known analytically, thus allowing late-time observables to be determined exactly. Expanding the earlier results of ref. \href{https://arxiv.org/abs/2303.03418}{arXiv:hep-th/2303.03418}, here we continue the program of analytically characterizing the late-time behavior of cosmological solutions. Our results are general in that they are derived without relying on any approximation nor are they based on any assumption on the sources of the potential, such as their higher-dimensional or string-theoretic origin. We point out a number of model-independent features that follow from our analytic results, including a convex-hull criterion for cosmic acceleration. When applied to string theory, our analytic knowledge of late-time cosmological solutions enables us to single out potentials that can describe an accelerating universe from those which cannot and to quantitatively test several conjectured Swampland criteria.
\end{abstract}

\maketitle

\section{Introduction}

Among the possible mechanisms for cosmic acceleration, the simplest possibility is for it to be sourced by a positive cosmological constant. Realizing this possibility in string theory however is a notoriously difficult task. This is because the source of cosmic expansion should be derived from an ultraviolet-complete theory. It is a highly non-trivial balancing act that the microphysics that stabilize the moduli in string theory would lead to a metastable de Sitter vacuum. The Dine-Seiberg problem \cite{Dine:1985he} further highlights the tension between stabilizing moduli and computational control. This tension is not insurmountable, as a number of sophisticated scenarios for realizing de Sitter vacua in string theory have been developed over the years (for a recent review, see ref. \cite{Flauger:2022hie}). However, the lack of fully explicit models has led some to contemplate the possibility that metastable stable de Sitter vacua may not exist in string theory \cite{Danielsson:2018ztv}. In fact, efforts to find de Sitter vacua reveal more: it is also rare to find potentials derived in string compactifications that are positive and sufficiently flat. These and related considerations are the backdrop for various de Sitter conjectures \cite{Obied:2018sgi, Ooguri:2018wrx, Bedroya:2019snp, Rudelius:2021azq} (see e.g. also refs. \cite{Andriot:2018wzk, Garg:2018reu}) which bound in one way or another the gradient and/or the curvature of the potential. While some of these conjectures, if proven, would forbid de Sitter vacua, inflationary and quintessence scenarios are not a priori ruled out unless the order-1 numbers in the bounds can be decisively determined. This makes the current cosmic acceleration driven by a rolling potential an interesting possibility, especially in view of the Dine-Seiberg problem mentioned above. As the scalar fields roll asymptotically, we reach the boundary of field space, which in string theory is characterized by weak couplings and restoration of symmetries. It has long been speculated that the observed hierarchies and unnaturally small couplings in nature may be attributed to the universe approaching the asymptotic regime, dating back even to Dirac \cite{dirac1974cosmological}. The Dark Dimension scenario  \cite{Montero:2022prj} is another realization of this idea. Whether cosmic acceleration can take place at the boundary of the moduli space in string theory has recently been explored in a variety of approaches \cite{Conlon:2022pnx, Rudelius:2022gbz, Calderon-Infante:2022nxb, Marconnet:2022fmx, Bedroya:2022tbh, Shiu:2023nph}.

In this article, we continue the program -- initiated in ref. \cite{Shiu:2023nph} -- aimed at characterizing multi-field multi-exponential scalar-field cosmologies at late times. In particular, we prove and discuss universal convergence criteria to cosmologies with power-law scale factor. These so-called scaling solutions correspond to the critical points of the autonomous system of differential equations describing the time evolution of scalar fields coupled to gravity in an FLRW-spacetime. This is a very useful result since we have complete analytic knowledge of scaling cosmologies \cite{Collinucci:2004iw}. Earlier work discussed the stability of these solutions at perturbative linear order and numerically \cite{Halliwell:1986ja, Copeland:1997et, Malik:1998gy, Coley:1999mj, Guo:2003eu, Guo:2003rs, Bergshoeff:2003vb, Collinucci:2004iw, Kim:2005ne, Hartong:2006rt}. Moreover, we also present a new general bound with a simple geometric interpretation. Jointly with the universal bound for late-time cosmic acceleration presented in ref. \cite{Shiu:2023nph}, which is also simple to interpret geometrically, we are now endowed with powerful analytic tools to characterize late-time multi-field multi-exponential scalar-field cosmologies.

Throughout our study, we do not work under the assumptions of slow-roll nor do we look specifically for de Sitter vacua. Accelerated expansion can take place beyond the slow-roll approximation as long as the time variation of the Hubble parameter is sufficiently slow, i.e $\epsilon \equiv -\dot{H}/H^2< 1$. As is well known, this $\epsilon$-parameter, rather than the norm of the potential gradient, provides the proper diagnostic for cosmic acceleration regardless of whether the potential or the kinetic term dominates (see e.g. ref. \cite{Bahamonde:2017ize} for a review). In fact, this $\epsilon$-parameter is related to the deceleration parameter  $q \equiv - \ddot{a} a / \dot{a}^2$ commonly used in the cosmology literature by $\epsilon=1+q$. Given this distinction between cosmic acceleration in general and the special case of slow-roll, we discuss the exact relationships between the scalar-potential directional derivative, the norm of the scalar-potential gradient and the fate of cosmic acceleration -- including its time dependence -- in a model-independent way. Scaling cosmologies are special in that they saturate many inequalities and allow for exact checks of more general bounds. This work therefore naturally connects with the de Sitter conjecture of the Swampland Program in that it provides analytic results for the late-time potential gradient norm of multi-field multi-exponential potentials commonly found in the asymptotic regions of the moduli space in string theory.

Importantly, our mathematical results are completely general, regardless of any higher-dimensional and/or string-theoretic assumption. Nonetheless, the structure of the potentials we characterize is ubiquitous in the asymptotic regions of the moduli space of string compactifications. This motivates us to discuss implications of our late-time convergence results for the Swampland Program. To help keep track of general results from string theoretical discussions, all instances in which we make a string-theoretic assumption will be pointed out explicitly. Concerning late-time cosmologies, we formulate and give an analytic proof of a convex-hull criterion for cosmic acceleration. With respect to the universal bound on late-time acceleration of ref. \cite{Shiu:2023nph}, we show analytically that the distance of the coupling convex hull from the origin gives the lower bound for the $\epsilon$-parameter, and we further show that scaling cosmologies saturate this bound. Moreover, we introduce a method to compute the lower bound for the norm of the scalar-potential gradient, as it is a defining quantity of the de Sitter conjecture. We also point out that there exist scaling solutions where a subset of the scalars are stabilized, while the others are rolling. The fact that scaling solutions are late-time attractors means that at sufficiently late time, the dynamics of the rolling fields still keeps the stabilized moduli intact. As our results do not rely on string-theoretic assumptions, what we prove is a mathematically-rigorous diagnostic of cosmic acceleration based on the convex hull of the exponential couplings. Nonetheless, this convex-hull criterion can be applied to string theoretical models to check whether the associated couplings allow for cosmic acceleration. The same holds for the characterization of the constant appearing in the de Sitter conjecture. For instance, it has already been discussed by ref. \cite{Shiu:2023nph} that a rolling $d$-dimensional dilaton poses a strong obstacle to acceleration since it couples universally to all scalar-potential terms and it does so with a very steep potential profile. Our formalism makes it clear that the specific dilaton couplings that appear in string theory rule out cosmic acceleration in a vast class of models. Recently, cosmologies with single-field exponential potentials have received considerable attention in the context of string theory \cite{Conlon:2022pnx, Rudelius:2022gbz, Bedroya:2022tbh, Apers:2022cyl}. Our general results presented here would enable one to substantially extend such studies.

This paper is organized as follows. In section \ref{sec: late-time cosmologies}, we extensively review the bound presented in ref. \cite{Shiu:2023nph} including its physical interpretation, and also present a new bound. In section \ref{sec: late-time scaling cosmologies}, we discuss the convergence criteria to scaling cosmologies and comment on their mathematical and physical properties. In section \ref{sec: late-time cosmologies and the swampland}, we discuss implications of our bounds and convergence results for the Swampland Program. In section \ref{sec: examples}, we discuss a few illustrative toy models. For ease of presentation, all analytic proofs are presented, in a detailed mathematical fashion, in appendix \ref{app: late-time cosmological attractors}; in the main text, we focus primarily on the physical interpretation of the results. Our conventions on reference frames and the terminologies we use for various dilatons and radions are summarized in appendix \ref{app: string dimensional reductions}. For completeness, in appendix \ref{app: string-compactification scalar potentials} we provide useful formulae for the string-theoretic scalar potentials generated by internal curvature, fluxes and localized sources as well as Casimir energies.

\section{Late-time cosmologies} \label{sec: late-time cosmologies}
In this paper, we consider low-energy effective theories in which a number of canonically-normalized scalar fields $\phi^a$, for $a=1,\dots,n$, are subject to a scalar potential of the form
\begin{equation} \label{generic exponential potential}
    V = \sum_{i = 1}^m \Lambda_i \, \e^{- \kappa_d \gamma_{i a} \phi^a}.
\end{equation}
We are agnostic about the origin of the potential, whether it descends from a string compactification, or it simply describes a phenomenological model without a higher-dimensional structure. Here, $\Lambda_i$ and $\gamma_{ia}$ are constants and they depend on the microscopic origin of the scalar-potential, if there is one, while $\kappa_d$ is the $d$-dimensional gravitational coupling. In the string-theory context, the set of scalars $\phi^a$ includes minimally the $d$-dimensional dilaton $\smash{\tdelta}$ and a radion $\smash{\tsigma}$ that controls the string-frame volume, unless these fields are stabilized at high energy scales. This general class of potentials also subsumes e.g. generalized assisted inflation \cite{Liddle:1998jc, Copeland:1999cs}. For $d>2$, let the non-compact $d$-dimensional spacetime be described by the FLRW-metric
\begin{equation} \label{FLRW-metric}
    d \tilde{s}_{1,d-1}^2 = - \de t^2 + a^2(t) \, \de l_{\mathbb{E}^{d-1}}^2
\end{equation}
with an Euclidean $(d-1)$-dimensional space; the Hubble parameter is $\smash{H = \dot{a}/a}$, where $a$ is the scale factor. Then, it can be shown that the scalar-field and Friedmann equations reduce to
\begin{subequations}
\begin{align}
    & \ddot{\phi}^a + (d-1) H \dot{\phi}^a + \dfrac{\der V}{\der \phi_a} = 0, \label{FRW-KG eq.} \\[1.0ex]
    & \dfrac{(d-1) (d-2)}{2} \, H^2 - \kappa_d^2 \biggl[ \dfrac{1}{2} \, \dot{\phi}_a \dot{\phi}^a + V \biggr] = 0, \label{Friedmann eq. 1} \\
    & \dot{H} = - \dfrac{\kappa_d^2}{d-2} \, \biggl[ \dfrac{1}{2} \, \dot{\phi}_a \dot{\phi}^a - V \biggr] - \dfrac{d-1}{2} H^2, \label{Friedmann eq. 2}
\end{align}
\end{subequations}
where for simplicity it has been assumed that the scalars only depend on the FLRW-metric time parameter. A combination of eq. (\ref{Friedmann eq. 1}) with eq. (\ref{Friedmann eq. 2}) gives
\begin{equation} \label{Friedmann eq.}
    \dot{H} = - \dfrac{\kappa_d^2}{d-2} \, \dot{\phi}_a \dot{\phi}^a.
\end{equation}
For multi-field multi-exponential scalar potentials of the form in eq. (\ref{generic exponential potential}), one can reformulate the scalar-field and Friedmann equations in terms of an autonomous system of ordinary differential equations \cite{Copeland:1997et, Coley:1999mj, Guo:2003eu}. On the one hand, this allows one to find and classify specific exact solutions to the cosmological equations in an easier way than by solving the original equations \cite{Collinucci:2004iw}. Perturbative and numerical analyses of the stability of such solutions are also made simpler in these terms \cite{Halliwell:1986ja, Copeland:1997et, Malik:1998gy, Coley:1999mj, Guo:2003eu, Guo:2003rs, Bergshoeff:2003vb, Collinucci:2004iw, Kim:2005ne, Hartong:2006rt}. On the other hand, as shown in ref. \cite{Shiu:2023nph}, this also allows one to single out universal late-time behaviors for all solutions to eqs. (\ref{FRW-KG eq.}, \ref{Friedmann eq. 1}, \ref{Friedmann eq. 2}). In this paper, we continue the program of characterizing late-time cosmologies by analyzing the asymptotic behavior of the cosmological autonomous system.

\subsubsection*{A note on slow-roll cosmologies}

In slow-roll scenarios, one assumes that the scalar fields evolve over time by slowly rolling down through the scalar potential. Mathematically, we define the slow-roll conditions to be the the conditions under which the solutions $\smash{\phi_{\mathrm{sr}}^a}$ and $\smash{H_{\mathrm{sr}}}$ to the field equations are such that the following approximations are consistent \cite{Lyth:1998xn}.

\begin{enumerate}[label=(\alph*)]
    \item The second-order scalar-field derivative is negligible compared to Hubble-friction term, i.e.
    \begin{align*}
        \frac{\ab \ddot{\phi}^a_{\mathrm{sr}} \ab}{H_{\mathrm{sr}} \ab \dot{\phi}^a_{\mathrm{sr}} \ab} \ll 1.
    \end{align*}
    In eq. (\ref{FRW-KG eq.}), one can thus neglect the second-derivative term, writing
    \begin{equation} \label{slow-roll FRW-KG eq.}
        (d-1) H_{\mathrm{sr}} \dot{\phi}_{\mathrm{sr}}^a = - \dfrac{\der V_{\mathrm{sr}}}{\der \phi_{\mathrm{sr} a}},
    \end{equation}
    where $\smash{V_{\mathrm{sr}} = V(\phi_{\mathrm{sr}})}$. By the latter expression, one can also approximate eq. (\ref{Friedmann eq.}) as
    \begin{equation} \label{slow-roll Friedmann eq. 2}
        \dot{H}_{\mathrm{sr}} = - \dfrac{\kappa_d^2}{(d-2)(d-1)^2} \dfrac{1}{H_{\mathrm{sr}}^2} \dfrac{\der V_{\mathrm{sr}}}{\der \phi_{\mathrm{sr} a}} \dfrac{\der V_{\mathrm{sr}}}{\der \phi_{\mathrm{sr}}^a}.
    \end{equation}
    \item The scalar-field kinetic energy is negligible compared to the scalar-field potential energy, i.e.
    \begin{align*}
        \dfrac{1}{2} \, \dot{\phi}_{\mathrm{sr} a} \dot{\phi}^a_{\mathrm{sr}} \ll V_{\mathrm{sr}}.
    \end{align*}
    In eq. (\ref{Friedmann eq. 1}), one neglects the kinetic energy compared to the potential energy, thus getting
    \begin{equation} \label{slow-roll Friedmann eq. 1}
        H_{\mathrm{sr}}^2 = \dfrac{2 \kappa_d^2 \, V_{\mathrm{sr}}}{(d-1) (d-2)}.
    \end{equation}
\end{enumerate}
Notice that slow roll is not equivalent to cosmic acceleration: while slow roll can imply cosmic acceleration, cosmic acceleration does not imply slow roll.\footnote{See ref. \cite{Achucarro:2018vey} for a discussion of this in the context of multi-field inflation, in relation to the Swampland Program; for an example, see e.g. ref. \cite{Aragam:2019omo}.} In the following, we will consider the conditions for accelerated cosmic expansion, but we will not assume slow roll: the equations we base our analysis on are eqs. (\ref{FRW-KG eq.}, \ref{Friedmann eq. 1}, \ref{Friedmann eq. 2}) and eq. (\ref{Friedmann eq.}), not eqs. (\ref{slow-roll FRW-KG eq.}, \ref{slow-roll Friedmann eq. 2}, \ref{slow-roll Friedmann eq. 1}). In fact, we will actually review how the bound in ref. \cite{Shiu:2023nph} generally prevents it from being a reasonable approximation for multi-field multi-exponential potentials and then we will discuss in detail the relationship of scaling cosmologies with the slow-roll approximation.

\subsection{Bounds on late-time acceleration} \label{subsec: bounds on late-time cosmic-acceleration}

In ref. \cite{Shiu:2023nph} it is shown that for multi-field multi-exponential potentials there is a universal bound on late-time cosmic acceleration that only depends on the $\gamma_{ia}$-coefficients. Here we review this bound and further elaborate on its physical meaning.

An accelerated cosmological expansion can only be achieved if the total scalar potential is positive. Therefore, we focus on scenarios in which, at least asymptotically, we have $V>0$. Let $\smash{\Lambda_{i_+} > 0}$ and $\smash{\Lambda_{i_-} < 0}$ denote the positive- and negative-definite scalar-potential coefficients, respectively, distinguishing by the indices $\smash{i = i_+, i_-}$. For each field $\phi^a$, we define $\smash{\gamma_\pm^a = \min_{i_\pm} {\gamma_{i_\pm}}^a}$ and $\smash{\Gamma_\pm^a = \max_{i_\pm} {\gamma_{i_\pm}}^a}$: each field such that $\smash{\gamma_+^a \geq \Gamma_-^a}$ or $\smash{\gamma_-^a \geq \Gamma_+^a}$ contributes towards a non-trivial lower bound for the $\epsilon$-parameter at sufficiently late times.

For all fields $\smash{\phi^a}$, for simplicity let us assume that $\smash{\gamma_+^a \geq \Gamma_-^a}$: if $\smash{\gamma_+^a > 0}$, we define $\smash{\gamma_\infty^a = \gamma_+^a}$; otherwise, we define $\smash{\gamma_\infty^a = 0}$. If $\smash{(\gamma_\infty)^2 \leq \Gamma_d^2 = 4 \, (d-1) / (d-2)}$, then, given a sufficiently large time $t_\infty$, the $\epsilon$-parameter is bounded from below at all times $t > t_\infty$ as
\begin{equation} \label{epsilon bound}
    \epsilon \geq \dfrac{d-2}{4} \, (\gamma_\infty)^2.
\end{equation}
Of course, the $\epsilon$-parameter is also bounded from above as $\smash{\epsilon \leq d - 1}$. If $\smash{\gamma_-^a \geq \Gamma_+^a}$ for a field, then one can redefine this field as $\smash{\phi'{}^a = -\phi^a}$ and find the same bound in terms of the flipped $\gamma$-coefficients. If for a field none of these orderings is in place, then we must set $\smash{\gamma_\infty^a = 0}$. If $\smash{(\gamma_\infty)^2 > \Gamma_d^2}$, irrespective of the ordering of the $\smash{\gamma_\pm^a}$- and $\smash{\Gamma_\pm^a}$-coefficients, then the $\epsilon$-parameter asymptotically approaches the value $\epsilon = d-1$.\footnote{As shown rigorously in ref. \cite{Shiu:2023nph}, the upper bound $\epsilon = d-1$ is universal for all multi-field positive-definite multi-exponential potential. Nonetheless, locally in field space the potential gradient norm $\gamma$ may be larger than $\smash{\gamma = 2 \sqrt{(d-1)/(d-2)}}$. As we detail in subsec. \ref{subsec: scalar-potential derivatives and acceleration}, the value of $\gamma$
is generally unrelated to cosmic acceleration. Also, notice that the upper bound $\epsilon=d-1$ is unrelated to slow roll.} This extreme scenario corresponds to a late-time realization of kination \cite{Apers:2022cyl}.

For a positive-definite scalar potential, the bound in eq. (\ref{epsilon bound}) is easily interpreted. If all couplings are positive, each field $\phi^a$ is pushed towards $\smash{\phi^a \sim + \infty}$ at sufficiently late times. Asymptotically, the dominating term is the one with the smallest coupling $\smash{\gamma_a = \min_i \gamma_{ia} >0}$. Such term determines the minimum steepness of the potential in that field direction, which determines the minimum kinetic energy from the field $\phi^a$, and thus the minimum contribution to cosmic deceleration, as apparent from eq. (\ref{Friedmann eq.}). If the potential is steep enough, there cannot be acceleration: the kinetic terms are large and provide a large $\epsilon$-parameter. There is however a limitation on how large such kinetic terms can be relative to the Hubble parameter due to energy-momentum conservation, as apparent from eq. (\ref{Friedmann eq. 1}). This is the reason why the $\epsilon$-parameter is bounded from above as well. If the minimum $\gamma_{ia}$-coupling is negative, and it cannot be made positive by a field sign flip, then the scalar potential has a valley and the field is not necessarily being driven towards the moduli-space boundary. In this case, we are not generally receiving a non-trivial contribution to cosmic deceleration due to a kinetic term, for this specific field direction. A couple of representations of the bound for positive-definite potentials are in figs. \ref{fig.: acceleration bound 1} and \ref{fig.: acceleration bound 2}.

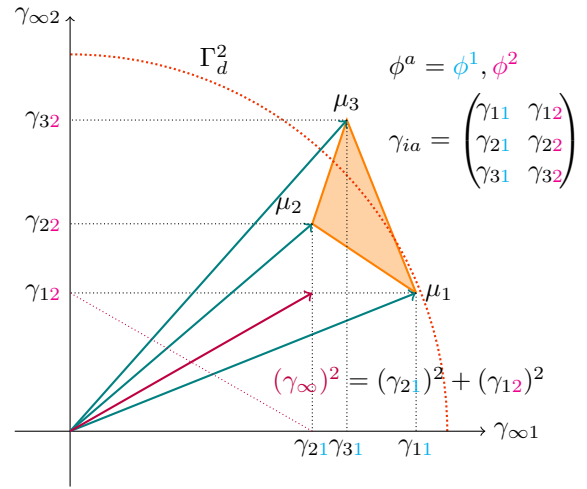
\begin{figure}[ht]
    \centering
    \begin{tikzpicture}[xscale=0.92,yscale=0.92,every node/.style={font=\normalsize}]
    
    \node[align=left] at (6,4.5){$\phi^a = {\color{cyan} \phi^1}, {\color{magenta} \phi^2}$ \\[1.0ex] $\gamma_{ia} = \left(\!\!\!\begin{array}{cc}
    \gamma_{1 \color{cyan} 1} & \gamma_{1 \color{magenta} 2} \\
    \gamma_{2 \color{cyan} 1} & \gamma_{2 \color{magenta} 2} \\
    \gamma_{3 \color{cyan} 1} & \gamma_{3 \color{magenta} 2}
    \end{array}\!\!\!\right)$};
 
    \draw[orange, thick, fill=orange!35!white] (5,2) -- (3.5,3) -- (4,4.5) -- (5,2);    

    \draw[densely dotted, thick, lime!20!red] (5.45,0) arc (0:90:5.45) node[above, black, pos=0.75]{$\Gamma_d^2$};
    
    \draw[->] (-0.8,0) -- (6,0) node[right]{$\gamma_{\infty 1}$};
    \draw[->] (0,-0.8) -- (0,6) node[left]{$\gamma_{\infty 2}$};
    
    \draw[->, thick, teal] (0,0) -- (5,2) node[right,black]{$\mu_1$};
    \draw[densely dotted] (5,2) -- (5,0) node[below]{$\gamma_{1 \color{cyan} 1}$};
    \draw[densely dotted] (5,2) -- (0,2) node[left]{$\gamma_{1 \color{magenta} 2}$};

    \draw[->, thick, teal] (0,0) -- (3.5,3) node[above left,black]{$\mu_2$};
    \draw[densely dotted] (3.5,3) -- (3.5,0) node[below]{$\gamma_{2 \color{cyan} 1}$};
    \draw[densely dotted] (3.5,3) -- (0,3) node[left]{$\gamma_{2 \color{magenta} 2}$};

    \draw[->, thick, teal] (0,0) -- (4,4.5) node[above,black]{$\mu_3$};
    \draw[densely dotted] (4,4.5) -- (4,0) node[below]{$\gamma_{3 \color{cyan} 1}$};
    \draw[densely dotted] (4,4.5) -- (0,4.5) node[left]{$\gamma_{3 \color{magenta} 2}$};

    \draw[densely dotted, purple] (3.5,0) -- (0,2) node[pos=0.2,black, above right]{${\color{purple} (\gamma_\infty)^2} = (\gamma_{2 \color{cyan} 1})^2 + (\gamma_{1 \color{magenta} 2})^2$};

    \draw[->, thick, purple] (0,0) -- (3.5,2);
    
    \end{tikzpicture}
    \caption{A representation of the late-time acceleration bound for a positive-definite scalar potential with non-negative exponential couplings. Here we also show the limiting value $\smash{\Gamma_d}$.}
    \label{fig.: acceleration bound 1}
\end{figure}

\begin{figure}[ht]
    \centering
    \begin{tikzpicture}[xscale=0.92,yscale=0.92,every node/.style={font=\normalsize}]
    
    \node[align=left] at (5.5,4.5){$\phi^a = {\color{cyan} \phi^1}, {\color{magenta} \phi^2}$ \\[1.0ex] $\gamma_{ia} = \left(\!\!\!\begin{array}{cc}
    \gamma_{1 \color{cyan} 1} & \gamma_{1 \color{magenta} 2} \\
    \gamma_{2 \color{cyan} 1} & \gamma_{2 \color{magenta} 2} \\
    \gamma_{3 \color{cyan} 1} & \gamma_{3 \color{magenta} 2}
    \end{array}\!\!\!\right)$};
 
    \draw[orange, thick, fill=orange!35!white] (5,2) -- (3.5,3) -- (-1.2,4) -- (5,2);
    
    \draw[->] (-1.5,0) -- (6,0) node[right]{$\gamma_{\infty 1}$};
    \draw[->] (0,-0.8) -- (0,5) node[left]{$\gamma_{\infty 2}$};
    
    \draw[->, thick, teal] (0,0) -- (5,2) node[right,black]{$\mu_1$};
    \draw[densely dotted] (5,2) -- (5,0) node[below]{$\gamma_{1 \color{cyan} 1}$};
    \draw[densely dotted] (5,2) -- (0,2) node[left]{$\gamma_{1 \color{magenta} 2}$};

    \draw[->, thick, teal] (0,0) -- (3.5,3) node[above left,black]{$\mu_2$};
    \draw[densely dotted] (3.5,3) -- (3.5,0) node[below]{$\gamma_{2 \color{cyan} 1}$};
    \draw[densely dotted] (3.5,3) -- (0,3) node[left]{$\gamma_{2 \color{magenta} 2}$};

    \draw[->, thick, teal] (0,0) -- (-1.2,4) node[above,black]{$\mu_3$};
    \draw[densely dotted] (-1.2,4) -- (-1.2,0) node[below]{$\gamma_{3 \color{cyan} 1}$};
    \draw[densely dotted] (-1.2,4) -- (0,4) node[above left]{$\gamma_{3 \color{magenta} 2}$};

    \draw[->, thick, purple] (0,0) -- (0,2) node[above right,black]{${\color{purple} (\gamma_\infty)^2} = (\gamma_{1 \color{magenta} 2})^2$};
    
    \end{tikzpicture}
    \caption{A representation of the late-time acceleration bound for a positive-definite scalar potential with both positive and negative exponential couplings: the bound is less restrictive than with non-negative-only couplings but still non-trivial.}
    \label{fig.: acceleration bound 2}
\end{figure}
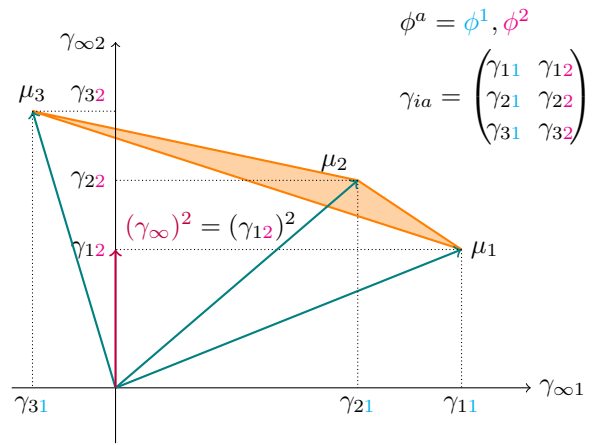

In the presence of both positive- and negative-definite terms in the potential, the interpretation of the bound is trickier, since it depends on the hierarchies and the relative sizes of all the couplings. It should not be forgotten that we are always assuming the total potential to be positive, as this is a prerequisite to even hope for acceleration. A representation of the bound with scalar-potential terms of both signs is in fig. \ref{fig.: acceleration bound 3}.

\begin{figure}[ht]
    \centering
    \begin{tikzpicture}[xscale=0.92,yscale=0.92,every node/.style={font=\normalsize}]
    
    \node[align=left] at (6.6,3.4){$\phi^a = {\color{cyan} \phi^1}, {\color{magenta} \phi^2}$ \\[0.5ex] $\gamma_{i_\pm a} = \left(\!\!\!\begin{array}{cc}
    \gamma_{1_\pm \color{cyan} 1} & \gamma_{1_\pm \color{magenta} 2} \\
    \gamma_{2_\pm \color{cyan} 1} & \gamma_{2_\pm \color{magenta} 2}
    \end{array}\!\!\!\right)$};
 
    \draw[orange, thick] (5.5,-2.5) -- (4.5,-3);
    
    \draw[->] (-0.8,0) -- (6,0) node[right]{$\gamma_{\infty 1}$};
    \draw[->] (0,-3.5) -- (0,4) node[left]{$\gamma_{\infty 2}$};
    
    \draw[->, thick, teal] (0,0) -- (5.5,-2.5) node[right,black]{};
    \draw[densely dotted] (5.5,-2.5) -- (5.5,0) node[below]{$\Gamma_{+ \color{cyan} 1}$};
    \draw[densely dotted] (5.5,-2.5) -- (0,-2.5) node[left]{$\Gamma_{+ \color{magenta} 2}$};

    \draw[->, thick, teal] (0,0) -- (4.5,-3) node[right,black]{};
    \draw[densely dotted] (4.5,-3) -- (4.5,0) node[below]{$\gamma_{+ \color{cyan} 1}$};
    \draw[densely dotted] (4.5,-3) -- (0,-3) node[left]{$\gamma_{+ \color{magenta} 2}$};

    \draw[->, thick, green] (0,0) -- (3.5,-1) node[right,black]{};
    \draw[densely dotted] (3.5,-1) -- (3.5,0) node[below]{$\Gamma_{- \color{cyan} 1}$};
    \draw[densely dotted] (3.5,-1) -- (0,-1) node[left]{$\Gamma_{- \color{magenta} 2}$};

    \draw[->, thick, green] (0,0) -- (2,-2) node[right,black]{};
    \draw[densely dotted] (2,-2) -- (2,0) node[below]{$\gamma_{- \color{cyan} 1}$};
    \draw[densely dotted] (2,-2) -- (0,-2) node[left]{$\gamma_{- \color{magenta} 2}$};

    \draw[->, densely dotted, teal] (0,0) -- (5.5,2.5) node[right,black]{};
    \draw[densely dotted] (5.5,2.5) -- (5.5,0);
    \draw[densely dotted] (5.5,2.5) -- (0,2.5) node[left]{$\gamma'{}_{+ \color{magenta} 2}$};

    \draw[->, densely dotted, teal] (0,0) -- (4.5,3) node[right,black]{};
    \draw[densely dotted] (4.5,3) -- (4.5,0);
    \draw[densely dotted] (4.5,3) -- (0,3) node[left]{$\Gamma'{}_{+ \color{magenta} 2}$};

    \draw[->, densely dotted, green] (0,0) -- (3.5,1) node[right,black]{};
    \draw[densely dotted] (3.5,1) -- (3.5,0);
    \draw[densely dotted] (3.5,1) -- (0,1) node[left]{$\gamma'{}_{- \color{magenta} 2}$};

    \draw[->, densely dotted, green] (0,0) -- (2,2) node[right,black]{};
    \draw[densely dotted] (2,2) -- (2,0);
    \draw[densely dotted] (2,2) -- (0,2) node[left]{$\Gamma'{}_{- \color{magenta} 2}$};

    \draw[densely dotted, purple, ->] (0,0) -- (4.5,2.5);

    \draw[->, thick, purple] (0,0) -- (4.5,-2.5) node[below right]{$(\gamma_\infty)^2$};
    
    \end{tikzpicture}
    \caption{A representation of the late-time acceleration bound for both positive- and negative-definite scalar-potential terms. Here, the solid teal and green lines represent the exponential couplings for the positive- and negative-definite terms in the potential, respectively; the dotted teal and green lines represent the couplings after flipping the sign of the field $\smash{\phi^2}$ (the field $\smash{\phi^1}$ does not require any flip). The definition of the $\smash{(\gamma_\infty)^2}$-term is analogous to the case of a positive-definite scalar potential.}
    \label{fig.: acceleration bound 3}
\end{figure}
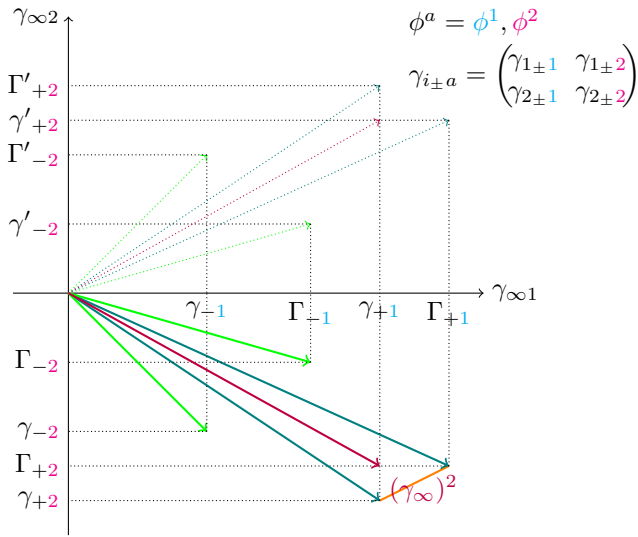

A noteworthy observation to make is that the bound in eq. (\ref{epsilon bound}) generically rules out late-time slow roll. Indeed, we can see immediately that the on-shell late-time $w$-parameter of the cosmological fluid of a multi-field multi-exponential potential is bounded from below as
\begin{equation}
    w = \dfrac{T - V}{T + V} \geq - 1 + \dfrac{1}{2} \, \dfrac{d-2}{d-1} \, (\gamma_\infty)^2.
\end{equation}
This clearly shows that, as long as the $\smash{(\gamma_\infty)^2}$-parameter is not infinitesimally small, or vanishing, the kinetic energy $T = \dot{\phi}_a \dot{\phi}^a/2$ is not parametrically suppressed with respect to the scalar potential energy $V$.

To conclude, we discuss a crucial observation on the bound in eq. (\ref{epsilon bound}). As therein stated, the bound is basis-dependent (as is also apparent e.g. in figs. \ref{fig.: acceleration bound 1} and \ref{fig.: acceleration bound 2}). Of course, it is clear by the inequality that one can optimize the bound via a choice of field basis. In fact, the action is invariant under field-space $\mathrm{O}(n)$-rotations $\smash{\hat{\phi}^a = \tensor{R}{^a_b} \phi^b}$, provided that one also redefines the exponential couplings as $\smash{\hat{\gamma}_{i a} = \gamma_{ib} \tensor{(R^{-1})}{^b_a}}$. Therefore, one can refine the bound by rotating the field-space basis so as to maximize the lowest-possible value that the $\epsilon$-parameter can take. In practice, the optimal version of the bound can be expressed as
\begin{equation} \label{optimal epsilon bound}
    \epsilon \geq \dfrac{d-2}{4} \, \max_{\mathrm{R} \in \mathrm{O}(n)} [\gamma_\infty(\mathrm{R})]^2,
\end{equation}
where $\smash{\mathrm{R} \in \mathrm{O}(n)}$ indicates all possible $\mathrm{O}(n)$-rotations in the $n$-dimensional field-space basis and $\smash{[\gamma_\infty(\mathrm{R})]^2}$ represents the $\smash{(\gamma_\infty)^2}$-coefficient computed in the $\smash{\mathrm{R}^{-1}}$-rotated field-space basis. Although this formulation of the bound is even stronger than the previous one, there can still be situations in which the bound happens to be trivial. From now on, we will express the optimal bound in eq. (\ref{optimal epsilon bound}) by referring to the quantity
\begin{equation}
    (\hat{\gamma}_\infty)^2 = \max_{\mathrm{R} \in \mathrm{O}(n)} [\gamma_\infty(\mathrm{R})]^2,
\end{equation}
which specifies the bound assuming that we have rotated the field-space basis in such a way as to reach the best bound among all the possible ones. Of course, all considerations made so far in terms of the quantity $\smash{(\gamma_\infty)^2}$ also immediately translate to the quantity $\smash{(\hat{\gamma}_\infty)^2}$.

In fact, we can explain the physical meaning of the analytic bound in eq. (\ref{optimal epsilon bound}). Let $\smash{(\hat{\phi}^a)_{a=1}^n = (\hat{\varphi}, (\hat{\phi}_{\check{a}})_{\check{a}=1}^{n-1})}$ be the field-space basis in which the $\smash{(\hat{\gamma}_\infty)^2}$-parameter is optimal and assume it is non-vanishing. Then the scalar potential in eq. (\ref{generic exponential potential}) can be written as
\begin{equation} \label{optimal-bound basis potential}
    V = V_\infty + \sum_{\iota=m_\infty+1}^m \Lambda_\iota \, \e^{- \kappa_d \hat{\gamma}_{\iota \hat{\varphi}} \hat{\varphi} - \kappa_d \hat{\gamma}_{\iota \check{a}} \hat{\phi}^{\check{a}}},
\end{equation}
where we define
\begin{equation} \label{V-infty potential}
    V_\infty = \biggl[ \sum_{\sigma=1}^{m_\infty} \Lambda_\sigma \, \e^{- \kappa_d \hat{\gamma}_{\sigma \check{a}} \hat{\phi}^{\check{a}}} \biggr] \, \e^{- \kappa_d \hat{\gamma}_\infty \hat{\varphi}}.
\end{equation}
Here, $\smash{\hat{\gamma}_\infty = \sqrt{(\hat{\gamma}_\infty)^2} > 0}$ is a universal coupling for a subset of the potential terms for the field $\smash{\hat{\varphi}}$ labelled by an index $\smash{\sigma = 1, \dots, m_\infty \leq m}$; by definition, this universal coupling is such that $\smash{\hat{\gamma}_\infty \leq \hat{\gamma}_{\iota \hat{\varphi}}}$ for the complementary subset of potential terms. All the other couplings are denoted as $\smash{\hat{\gamma}_{i \check{a}}}$. On the one hand, physical intuition suggests that the potential should drive the universally-coupled scalar towards the region $\smash{\hat{\varphi} \sim \infty}$. On the other hand, the rest of the terms should combine in a way that either stabilizes the fields (if couplings of opposite signs appear) or pushes them to also roll towards asymptotic regions (if all couplings are of equal sign). Therefore, the $\epsilon$-parameter is intuitively bounded from below by the minimal slope
\begin{align*}
    \hat{\gamma}_\infty = - \dfrac{1}{V_\infty} \dfrac{\der V_\infty}{\kappa_d \der \hat{\varphi}}.
\end{align*}
In a single-field theory with the single exponential potential $\smash{\hat{V}_\infty = \hat{\Lambda}_\infty \, \e^{- \kappa_d \hat{\gamma}_\infty \hat{\varphi}}}$, the late-time $\epsilon$-parameter would then be exactly the value $\smash{\epsilon = \hat{\epsilon}_\infty = [(d-2)/4] \, (\hat{\gamma}_\infty)^2}$, which corresponds to the lower bound in eq. (\ref{optimal epsilon bound}). If $\smash{\hat{\gamma}_\infty = 0}$, the discussion above is unchanged, but with the possibility that the leading term may give a de Sitter stationary point. A schematic interpretation of the bound of eq. (\ref{optimal epsilon bound}) is provided in figs. \ref{fig.: optimal acceleration bound} and \ref{fig.: trivial optimal acceleration bound}.

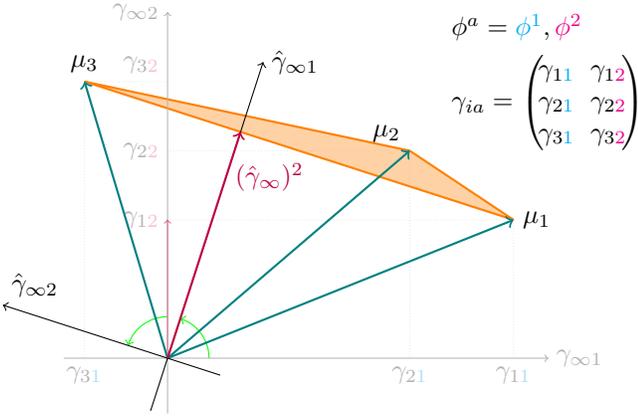
\begin{figure}[ht]
    \centering
    \begin{tikzpicture}[xscale=0.92,yscale=0.92,every node/.style={font=\normalsize}]
    
    \node[align=left] at (5.5,4){$\phi^a = {\color{cyan} \phi^1}, {\color{magenta} \phi^2}$ \\[1.0ex] $\gamma_{ia} = \left(\!\!\!\begin{array}{cc}
    \gamma_{1 \color{cyan} 1} & \gamma_{1 \color{magenta} 2} \\
    \gamma_{2 \color{cyan} 1} & \gamma_{2 \color{magenta} 2} \\
    \gamma_{3 \color{cyan} 1} & \gamma_{3 \color{magenta} 2}
    \end{array}\!\!\!\right)$};

    \draw[->, thin, green] (0.6,0) arc (0:atan(31/10):0.6);
    \draw[->, thin, green] (0,0.6) arc (90:90+atan(31/10):0.6);
 
    \draw[orange, thick, fill=orange!35!white] (5,2) -- (3.5,3) -- (-1.2,4) -- (5,2);
    
    \draw[->, ultra thin, gray!60!white] (-1.5,0) -- (5.5,0) node[right]{$\gamma_{\infty 1}$};
    \draw[->, ultra thin, gray!60!white] (0,-0.8) -- (0,5.0) node[left]{$\gamma_{\infty 2}$};
    
    \draw[->, thick, teal] (0,0) -- (5,2) node[right,black]{$\mu_1$};
    \draw[densely dotted, ultra thin, gray!60!white] (5,2) -- (5,0) node[below]{$\gamma_{1 \color{cyan!30!white} 1}$};
    \draw[densely dotted, ultra thin, gray!60!white] (5,2) -- (0,2) node[left]{$\gamma_{1 \color{magenta!30!white} 2}$};

    \draw[->, thick, teal] (0,0) -- (3.5,3) node[above left,black]{$\mu_2$};
    \draw[densely dotted, ultra thin, gray!60!white] (3.5,3) -- (3.5,0) node[below]{$\gamma_{2 \color{cyan!30!white} 1}$};
    \draw[densely dotted, ultra thin, gray!60!white] (3.5,3) -- (0,3) node[left]{$\gamma_{2 \color{magenta!30!white} 2}$};

    \draw[->, thick, teal] (0,0) -- (-1.2,4) node[above,black]{$\mu_3$};
    \draw[densely dotted, ultra thin, gray!60!white] (-1.2,4) -- (-1.2,0) node[below]{$\gamma_{3 \color{cyan!30!white} 1}$};
    \draw[densely dotted, ultra thin, gray!60!white] (-1.2,4) -- (0,4) node[above left]{$\gamma_{3 \color{magenta!30!white} 2}$};

    \draw[->, thick, ultra thin, purple!60!white] (0,0) -- (0,2) node[above right,black]{};

    \draw[->,rotate=atan(31/10)] (-0.8,0) -- (4.5,0) node[right]{$\hat{\gamma}_{\infty 1}$};
    \draw[->,rotate=atan(31/10)] (0,-0.8) -- (0,2.5) node[above right]{$\hat{\gamma}_{\infty 2}$};

    \draw[->, thick, purple] (0,0) -- (1120/1061,3472/1061) node[right,pos=0.8]{$(\hat{\gamma}_\infty)^2$};
    
    \end{tikzpicture}
    \caption{A representation of the optimal late-time acceleration bound $\smash{\epsilon \geq [(d-2)/4] \, (\hat{\gamma}_\infty)^2}$: lighter lines denote the original field basis, while darker lines denote the basis with the maximal lower bound.}
    \label{fig.: optimal acceleration bound}
\end{figure}

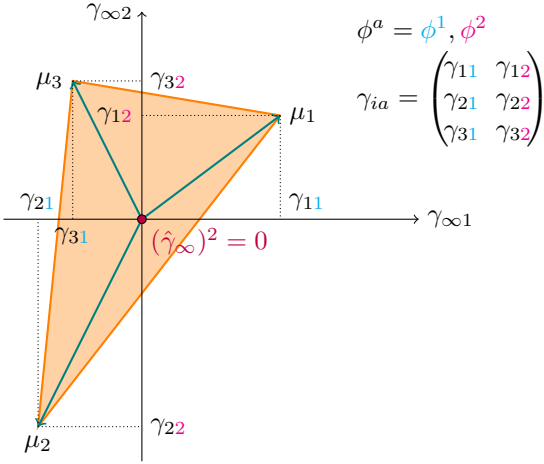
\begin{figure}[ht]
    \centering
    \begin{tikzpicture}[xscale=0.92,yscale=0.92,every node/.style={font=\normalsize}, ppoint/.style={draw,circle,minimum size=1.15mm,inner sep=0pt,outer sep=0pt,fill=purple,solid}]
    
    \node[align=left] at (4.5,2){$\phi^a = {\color{cyan} \phi^1}, {\color{magenta} \phi^2}$ \\[0.5ex] $\gamma_{ia} = \left(\!\!\!\begin{array}{cc}
    \gamma_{1 \color{cyan} 1} & \gamma_{1 \color{magenta} 2} \\
    \gamma_{2 \color{cyan} 1} & \gamma_{2 \color{magenta} 2} \\
    \gamma_{3 \color{cyan} 1} & \gamma_{3 \color{magenta} 2}
    \end{array}\!\!\!\right)$};

    \draw[orange, thick, fill=orange!35!white] (2,1.5) -- (-1,2) -- (-1.5,-3) -- (2,1.5);
    
    \draw[->] (-2,0) -- (4,0) node[right]{$\gamma_{\infty 1}$};
    \draw[->] (0,-3.5) -- (0,3) node[left]{$\gamma_{\infty 2}$};
    
    \draw[->, thick, teal] (0,0) -- (2,1.5) node[right,black]{$\mu_1$};
    \draw[densely dotted] (2,1.5) -- (2,0) node[above right]{$\gamma_{1 \color{cyan} 1}$};
    \draw[densely dotted] (2,1.5) -- (0,1.5) node[left]{$\gamma_{1 \color{magenta} 2}$};

    \draw[->, thick, teal] (0,0) -- (-1.5,-3) node[below,black]{$\mu_2$};
    \draw[densely dotted] (-1.5,-3) -- (-1.5,0) node[above]{$\gamma_{2 \color{cyan} 1}$};
    \draw[densely dotted] (-1.5,-3) -- (0,-3) node[right]{$\gamma_{2 \color{magenta} 2}$};

    \draw[->, thick, teal] (0,0) -- (-1,2) node[left,black]{$\mu_3$};
    \draw[densely dotted] (-1,2) -- (-1,0) node[below]{$\gamma_{3 \color{cyan} 1}$};
    \draw[densely dotted] (-1,2) -- (0,2) node[right]{$\gamma_{3 \color{magenta} 2}$};

    \node[ppoint] at (0,0){};
    \node[below right] at (0,0){$\color{purple} (\hat{\gamma}_\infty)^2 = 0$};
    
    \end{tikzpicture}
    \caption{A representation of a situation in which the acceleration bound is trivial. In a scenario like this, we do not see any obstruction for late-time acceleration.}
    \label{fig.: trivial optimal acceleration bound}
\end{figure}

\subsection{Alternative late-time acceleration bounds} \label{subsec: alternative late-time acceleration bounds}

As usual, we assume that the total scalar potential is positive, but make no assumption on the sign of the individual contributions. Let $\smash{{\overline{\gamma}_{\!(\sigma)}}^a}$ be the solutions to the system of equations
\begin{align*}
    \gamma_{ia} {\overline{\gamma}}^a = (\overline{\gamma})^2,
\end{align*}
for each index $i=1,\dots,m$, where the $\sigma$-subscript is a label for each of the solutions. Then, if we define $\smash{\overline{\gamma}^2 = \max_\sigma (\overline{\gamma}_{\!(\sigma)})^2}$ and $\smash{\overline{\gamma}_\infty^2 = \min \lbrace \Gamma_d^2, \overline{\gamma}^2 \rbrace}$, we can show that the late-time $\epsilon$-parameter is bounded from below as
\begin{equation} \label{alternative epsilon bound}
    \epsilon \geq \dfrac{d-2}{4} \, \overline{\gamma}^2_\infty.
\end{equation}
A mathematical proof of this is in appendix \ref{app: late-time cosmological attractors}: see corollary \ref{corollary: (X)^2-bound} and remarks \ref{remark: Xbar{X}-bound sign independence} and \ref{remark: alternative f-bound}. Such a bound is generally different from the bound in eqs. (\ref{epsilon bound}, \ref{optimal epsilon bound}). For ease of presentation, a clear physical interpretation of this bound is deferred to subsec. \ref{subsec: coupling convex hull and cosmic acceleration}.

\subsection{Comments on canonical normalization}

A comment on canonical normalization is in order. Our analysis is general, independently of any string-theoretic input. If it comes to string-theoretic realizations, however, one may worry that the assumption of canonically-normalized fields is not always justified. Here we argue that this is not the case. 

In string-theoretic realizations, the moduli space is not always flat, typically due to the presence of axions.\footnote{For instance, consider the type-IIB axio-dilaton $\tau = C_0 + \I \, \e^{-\sdil}$ and Kähler modulus $\rho = \alpha + \I \, \e^{4 \omega}$, where $C_0$ and $\alpha$ are the 0- and 4-RR-form axions. In the presence of 3-form flux $G_3 = F_3 - \tau H_3$, in a 4-dimensional Calabi-Yau orientifold compactification, their purely kinetic action can be read off the Kähler potential \cite{Giddings:2001yu, Grimm:2004uq}
\begin{align*}
    \kappa_4^2 K = - \mathrm{ln} \, [-\I (\tau - \overline{\tau})] - 3 \, \mathrm{ln} \, [-\I (\rho - \overline{\rho})] + \mathrm{ln} \, \dfrac{2}{\pi}.
\end{align*}
The cosmological equations of a set of real scalars $\varphi^a$ with field-space metric $k_{\alpha \beta} = k_{\alpha \beta}(\varphi)$ in an FLRW-background read
\begin{align*}
    & \ddot{\varphi}^\alpha + {\Gamma^\alpha}_{\beta \gamma} \, \dot{\varphi}^\beta \dot{\varphi}^\gamma + (d-1) H \dot{\varphi}^\alpha + k^{\alpha \beta} \dfrac{\der V}{\der \varphi^\beta} = 0, \\
    & \dfrac{(d-1) (d-2)}{2} \, H^2 - \kappa_d^2 \biggl[ \dfrac{1}{2} \, k_{\alpha \beta} \, \dot{\varphi}^\alpha \dot{\varphi}^\beta + V \biggr] = 0.
\end{align*}}
For instance, in type-II compactifications with $N_4=1$ supersymmetry, typically such axions $\theta$ belong to chiral supermultiplets as components of complex scalars $\xi = \theta + \I \, \e^{l \varphi}$, where $\varphi$ approaches the boundary as $l \varphi \to \infty$. Here, we can identify $\varphi$ as one of our canonically normalized scalars, provided the constant rescaling $\smash{\varphi = (\sqrt{2} \kappa_4 / \sqrt{n} l) \, \phi}$, and $l$ and $n$ are constants that depend on the details of the fields, with Kähler potentials of the form $\kappa_4^2 K = - n \, \mathrm{ln} \, [-\I (\xi - \overline{\xi})]$. In this case, the kinetic action takes the form
\begin{align*}
    S[\theta, \varphi] = \int_{\mathrm{X}_{1,3}} \de^{1,3} x \sqrt{- g_{1,3}} \; \dfrac{n}{4 \kappa_4^2} \, \bigl[ \e^{-2l \varphi} (\der \theta)^2 + l^2 (\der \varphi)^2 \bigr],
\end{align*}
and the field equations thus read
\begin{align*}
    & \ddot{\theta} - 2l \, \dot{\theta} \dot{\varphi} + 3 H \dot{\theta} + \dfrac{2 \kappa_4^2}{n} \, \e^{2l \varphi} \dfrac{\der V}{\der \theta} = 0, \\
    & \ddot{\varphi} + \dfrac{1}{l} \, \e^{-2l \varphi} \, \dot{\theta}^2 + 3 H \dot{\varphi} + \dfrac{2 \kappa_4^2}{n l^2} \, \dfrac{\der V}{\der \varphi} = 0, \\
    & 3 H^2 - \Bigl[ \dfrac{n}{4} \, \e^{-2l \varphi} \dot{\theta}^2 + \dfrac{n}{4} \, l^2 \dot{\varphi}^2 + V \Bigr] = 0.
\end{align*}
If the initial conditions are such that $\varphi_0 \gg 1$, then
the $\smash{\dot{\theta^2}}$-term in the $\varphi$-field equation and the axion kinetic term in the Friedmann equation are highly suppressed, unless the axionic time derivatives are large. On the other hand, the $\theta$-field equation can have an exponential suppression of the time-derivative terms or not, depending on the potential. If the axion is stabilized, then our analysis applies unchanged, since the axionic dynamics is asymptotically decoupled. If the axion is not stabilized, instead, the axionic kinetic energy could in principle be large and modify non-negligibly the other equations too. Remarkably, this latter scenario may accommodate late-time cosmic acceleration with less restringent bounds \cite{Sonner:2006yn, Cicoli:2020cfj, Cicoli:2020noz, Russo:2022pgo, Brinkmann:2022oxy}.

Although the discussion above has been referred to specific string models, we expect analogous conclusions to hold quite generally. Indeed, the factorized structure of the kinetic terms for pairs of fields involving an axion and a canonically-normalized scalar that has been discussed above is fairly generic in string compactifications. It does not apply just for the dilaton and the radion of a trivial compactification, but also for the radion controlling the overall compactification volume. This structure is also fairly generic for complex-structure moduli, at least asymptotically \cite{Grimm:2019ixq, Calderon-Infante:2022nxb}. However, for compactifications with multiple non-trivial intersection numbers, this may not be true for radions controlling the size of internal cycles. If such fields are rolling in the asymptotics of moduli space, they require a generalization of our results. If such fields are stabilized, our analysis still applies.

\section{Late-time scaling cosmologies} \label{sec: late-time scaling cosmologies}

Although the bounds in eqs. (\ref{epsilon bound}, \ref{optimal epsilon bound}) and eq. (\ref{alternative epsilon bound}) are strong ones, in certain conditions we can do even more and compute the late-time $\epsilon$-parameter analytically. We discuss how to do this below.

\subsection{Scaling cosmologies} \label{subsec: scaling cosmologies}

Scaling cosmologies are defined as solutions to the Friedmann equations in which the scale factor is of power-law form, meaning that it evolves over time as
\begin{equation} \label{power-law scale factor}
    a(t) = a_0 \, \Bigl( \dfrac{t}{t_0} \Bigr)^p,
\end{equation}
where the constant and positive power $p$ is related to the Hubble parameter through the identity $H = p / t$ and to the $\epsilon$-parameter as $\epsilon = 1 / p$, which is necessarily constant and positive. For a multi-field multi-exponential potential, scaling cosmologies are well-known exact solutions to the cosmological equations and, in particular, they correspond to the critical points of the cosmological autonomous system. In this subsection we consider the scaling cosmologies that generically always exist, following the classification of  ref. \cite{Collinucci:2004iw}; more details can be found in appendix \ref{app: late-time cosmological attractors}: see lemmas \ref{lemma: (X,Y)-plane critical points (rank=m)} and \ref{lemma: (X,Y)-plane critical points (x^2=0)}.

In detail, we consider the case in which the rank of the $\gamma_{ia}$-matrix matches the number of terms in the scalar potential, i.e. $\mathrm{rank} \, \gamma_{ia} = m$. This can easily be the case whenever the number of fields is not smaller than the number of scalar-potential terms, i.e. $n \geq m$. If $\mathrm{rank} \, \gamma_{ia} = m$ and also $n=m$, then the scalar potential can be regarded as the non-trivial multi-field extension of a single-field exponential potential; if instead $\mathrm{rank} \, \gamma_{ia} = m$ and $n>m$, the scalar fields outnumber the scalar-potential terms, but then we can rotate the field-space basis and obtain a theory where $n-m$ scalars are flat directions, thus reducing the problem to the previous case. If this rank-condition $\mathrm{rank} \, \gamma_{ia} = m$ is in place, then, given the matrix $\smash{M_{ij} = \gamma_{i a} {\gamma_j}^a}$, rolling-scalar solutions exist of the form
\begin{equation} \label{rank-condition scalar-field trajectories}
    \phi^a_* (t) = \phi^a_0 + \dfrac{2}{\kappa_d} \, \biggl[ \sum_{i=1}^m \sum_{j=1}^m {\gamma_i}^a (M^{-1})^{ij} \biggr] \, \mathrm{ln} \, \dfrac{t}{t_0},
\end{equation}
where the scale-factor power is
\begin{equation} \label{rank-condition scale-factor power}
    p = \dfrac{4}{d-2} \, \sum_{i=1}^m \sum_{j=1}^m (M^{-1})^{ij}.
\end{equation}
It can also be shown that in this case there are no de Sitter stationary points. Physically, this is because the shape of the multi-dimensional exponential potential is not complicated enough to constrain all the fields into a stationary point. It may also be the case that $n \geq m$, but $\mathrm{rank} \, \gamma_{ia} < m$, in which case scaling solutions may exist but are not of the form above. Scaling cosmologies in cases with $\mathrm{rank} \, \gamma_{ia} < m$ are discussed in subsec. \ref{subsec: further scaling solitions}.

Before moving on, we stress an obvious but important point: scaling cosmologies do not respect the slow-roll approximation, by which one drops the second-derivative term and the kinetic energy in eqs. (\ref{FRW-KG eq.}, \ref{Friedmann eq. 1}, \ref{Friedmann eq. 2}), and thanks to which one manages to express the $\epsilon$-parameter through the gradient of the scalar potential. This is obvious from eq. (\ref{rank-condition scalar-field trajectories}) and it will be commented on also in subsubsecs. \ref{subsubsec: scalar-potential directional derivative} and \ref{subsubsec: accidental gradient-flow trajectories}. All this means that, in a scaling cosmology, the slow-roll conditions are not fulfilled. Nonetheless, accelerated expansion is still possible if $p>1$.

\subsection{Scaling cosmologies as late-time attractors} \label{subsec: scaling cosmologies as late-time attractors}

Scaling cosmologies can be perturbatively-stable attractors of theories with multi-field multi-exponential potentials \cite{Copeland:1997et, Malik:1998gy, Coley:1999mj, Guo:2003eu, Guo:2003rs, Bergshoeff:2003vb, Kim:2005ne, Hartong:2006rt}. Moreover, at sufficiently late times, if $\smash{(\hat{\gamma}_\infty)^2 \leq \Gamma_d^2}$, in view of eq. (\ref{optimal epsilon bound}), the scale factor is bounded from below and from above by power-law evolution; if $\smash{(\hat{\gamma}_\infty)^2 > \Gamma_d^2}$, scaling solutions are inevitable, with a power $\smash{p = 1/(d-1)}$. In this paper, we extend these observations by showing the sufficient conditions under which scaling cosmologies are late-time attractors independently of the initial conditions, thus going beyond a perturbative analysis.

Given $n$ canonically-normalized scalars $\phi^a$, let a multi-exponential potential $V$ of the form in eq. (\ref{generic exponential potential}), in a $d$-dimensional FLRW-metric, as in eq. (\ref{FLRW-metric}). Under certain assumptions, we are able to prove analytically the late-time form of any cosmological solution to eqs. (\ref{FRW-KG eq.}, \ref{Friedmann eq. 1}, \ref{Friedmann eq. 2}). We focus on two distinct cases.

First, we consider scenarios that correspond to proper scaling solutions. In particular, let the following conditions be in place.
\begin{enumerate}[label=(\roman*)]
    \item \label{condition (a)} All scalar potential terms are positive-definite, i.e. $\Lambda_i > 0$, for all $i$-indices. Also, let $\smash{\mathrm{rank} \, \gamma_{ia} = m}$: this is a necessary condition for the scaling solutions in eqs. (\ref{power-law scale factor}, \ref{rank-condition scalar-field trajectories}, \ref{rank-condition scale-factor power}) to exist.
    \item \label{condition (b)} Given the matrix $\smash{M_{ij} = \gamma_{ia} {\gamma_{j}}^a}$, for all $i$-indices we have the inequalities
    \begin{align*}
        \lambda^i = \sum_{j=1}^m (M^{-1})^{ij} \geq 0,
    \end{align*}
    subject to the additional condition
    \begin{align*}
        \sum_{i=1}^m \lambda^i > 0.
    \end{align*}
    Physically, this means that the scalar-potential terms are such that there cannot be any apparent notion of subdominant terms (see e.g. the physical interpretation of the late-time convergence result below and the discussions in subsecs. \ref{subsec: partial moduli stabilization} and \ref{subsec: coupling convex hull and cosmic acceleration}).
    \item \label{condition (c)} The time-integral of $\kappa_d^2 V/H$ is divergent, i.e.
    \begin{align*}
        \Delta = \int_{t_0}^\infty \de t \; [(d-1) - \epsilon(t)] \, H(t) = \infty.
    \end{align*}
    Intuitively, this represents the fact that the $\gamma_{ia}$-couplings are not too large, since, if they were, then we would have $\Delta < \infty$ \cite{Shiu:2023nph}.
\end{enumerate}
If conditions \ref{condition (a)}, \ref{condition (b)} and \ref{condition (c)} are satisfied, we can show that the late-time solutions to the scalar-field and Friedmann equations are the scaling solutions in eqs. (\ref{power-law scale factor}, \ref{rank-condition scalar-field trajectories}, \ref{rank-condition scale-factor power}). This is proven in appendix \ref{app: late-time cosmological attractors}: see theorem \ref{theorem: critical-point convergence}. This is the first analytic proof that scaling cosmologies are attractors of the scalar-field and Friedmann equations for a multi-scalar multi-exponential scalar potential and it represents a fundamental result of this paper.

Second, we consider scenarios that correspond to degenerate non-proper scaling cosmologies: from now on, we define them to be the aforementioned scenarios in which the integral $\Delta$ is finite. If we have
\begin{align*}
    \Delta = \int_{t_0}^\infty \de t \; [(d-1) - \epsilon(t)] \, H(t) < \infty,
\end{align*}
then the $\epsilon$-parameter asymptotically approaches the value $\epsilon = d-1$ \cite{Shiu:2023nph}. Notice that this is the case whenever we have $\smash{(\hat{\gamma}_\infty)^2 > \Gamma_d^2}$, in the terminology of subsec. \ref{subsec: bounds on late-time cosmic-acceleration}; in this case, we do not even need to consider conditions \ref{condition (a)}-\ref{condition (b)}, but it is enough to have a positive-definite potential.

Before going on, we provide a simple physical interpretation of the late-time convergence to a scaling cosmology. In view of eqs. (\ref{optimal-bound basis potential}, \ref{V-infty potential}), in an appropriate field-space basis, if there are no manifestly subdominant terms in the potential, a multi-exponential potential is asymptotically equivalent to a 1-field 1-term potential. Indeed, all the field besides the minimally necessarily-rolling one -- namely a field $\smash{\hat{\varphi}}$, with a coupling $\smash{\hat{\gamma}_\infty}$ -- are asymptotically stable. This corresponds to the case where $V = V_\infty$. Then, once the problem is reduced to the study of a single-field exponential potential, it is easy to check that $\smash{\hat{\varphi}_*(t) = \hat{\varphi}_0 + 2/(\kappa_d \hat{\gamma}_\infty) \, \mathrm{ln} \, t/t_0}$ is a simple solution to the field equations. Of course, if the potential is too steep, one reaches the condition of pure kination. In case there are subdominant terms, it is intuitively clear that one should be able to follow the same arguments, after dropping such subdominant terms, but we have not analytically proven late-time convergence in the presence of subdominant terms yet; for developments, see ref. \cite{criticalpoints2}.

Now, an important point to notice is the following. Assuming conditions \ref{condition (a)}-\ref{condition (b)} to be verified by the scalar potential of interest, we are faced with a conceptual impasse: in order to compute $\Delta$, we need to know the field-space trajectories, but to determine the field-space trajectories, we need to know whether $\Delta$ is divergent. Unfortunately, the only situation in which this impasse is known to be bypassed is if we have $\smash{(\hat{\gamma}_\infty)^2 > \Gamma_d^2}$, as explained above. In all other cases, what we can conclude is only the following: we either have the proper scaling solution or the degenerate non-proper solution, namely the $\epsilon$-parameter is either $\smash{\epsilon = [(d-2)/4] / \bigl[\sum_{i=1}^m \lambda^i \bigr]}$, or $\epsilon=d-1$.\footnote{Although conceptually this impasse also appears in the bound proven in ref. \cite{Shiu:2023nph}, namely the bound reviewed here in eq. (\ref{optimal epsilon bound}), that situation is easier. Indeed, if $\smash{(\hat{\gamma}_\infty)^2 \leq \Gamma_d^2}$, then we have $\smash{[(d-2)/4] \, (\hat{\gamma}_\infty)^2 \leq \epsilon \leq d-1}$ if $\Delta = \infty$ and $\epsilon = d-1$ if $\Delta < \infty$: therefore, whichever the value of $\Delta$ is, a conservative claim is that, at least, one has $\smash{\epsilon \geq [(d-2)/4] (\hat{\gamma}_\infty)^2}$.} In what follows, if conditions \ref{condition (a)}-\ref{condition (b)} are met, we will primarily discuss the proper scaling solution, but with the caveat in mind that the degenerate non-proper scaling solution may also be the actual attractor (as a priori we do not know whether $\Delta$ is infinite or not). Since we are concerned with cosmic acceleration, therefore, even if we find an accelerating proper scaling solution, this is just a necessary check that accelerated expansion is possible, but without any sufficiency condition being known to be in place as well.

In case some of the terms $\lambda^i$ are negative, we are currently unable to prove any convergence results. However, physical intuition suggest us that it happens that $\lambda^i < 0$ in cases in which the scalar-potential term $V_i$ is asymptotically subdominant. Our intuition is based on the physical interpretation of the late-time convergence result above and on the considerations in subsecs. \ref{subsec: partial moduli stabilization} and \ref{subsec: coupling convex hull and cosmic acceleration} and on earlier perturbative and numerical analyses in refs. \cite{Guo:2003eu, Guo:2003rs, Collinucci:2004iw, Hartong:2006rt}. Therefore, we expect that one may able to prove convergence to the scaling solution that one obtains after effectively truncating the potential to the sum $\smash{V = \sum_{\sigma=1}^{m_\infty} V_\sigma}$, with $\sigma = 1, \dots, m_\infty \leq m$, where the subset of $\sigma$-indices denotes a set of potentials such that $\smash{\lambda^\sigma \geq 0}$ \cite{criticalpoints2}.

At present, we do not have a fundamental physical interpretation of the integral $\Delta$. However, we can observe a few noteworthy features. To start, given the parameter $\smash{\eta = - \dot{\epsilon} / (\epsilon H)}$, one can qualitatively see that $\Delta < \infty$ requires the asymptotic behaviors $\smash{\epsilon(t \sim \infty) \lesssim (d-1) - \ell / t}$ and $\smash{- \eta (t \sim \infty) \gtrsim \ell/t}$, for some positive constant $\ell$.\footnote{For a finite integral $\Delta < \infty$, we know the late-time behaviors $\smash{H (t \sim \infty) \simeq 1/[(d-1) \, t]}$ and $\smash{\epsilon (t \sim \infty) \simeq d-1}$. For the inverse time $s = 1/t$, in view of the asymptotic expansions
\begin{align*}
    \epsilon(s) & = [\epsilon]_{s=0} + \Bigl[\dfrac{\de \epsilon}{\de s}\Bigr]_{s=0} \cdot s + O(s)^2 \\
    & = (d-1) + \Bigl[ \dfrac{1}{s} \, \eta \Bigr]_{s=0} \cdot s + O(s)^2, \\
    H(s) & = \dfrac{s}{d-1} + O(s)^2,
\end{align*}
because the integral is finite, i.e.
\begin{align*}
    \Delta = \int_{0}^{\frac{1}{t_0}} \dfrac{\de s}{s^2} \; [(d-1) - \epsilon(s)] \, H(s) < \infty,
\end{align*}
the integration near $s=0$ must provide a finite contribution, thus bounding $\eta$ through the expansion
\begin{align*}
    \dfrac{1}{s^2} \; [(d-1) - \epsilon(s)] \, H(s) = - \dfrac{1}{d-1} \, \Bigl[ \dfrac{1}{s} \, \eta \Bigr]_{s=0} + O(s).
\end{align*}} It is harder to prove 
the corresponding
implications for $\Delta=\infty$ since in this case we only generally know bounds and not exact identities. By eqs. (\ref{Friedmann eq. 1}, \ref{Friedmann eq.}), as already anticipated, we can observe that we can express $\Delta$ as
\begin{align*}
    \Delta = \dfrac{2 \kappa_d^2}{d-2} \int_{t_0}^\infty \de t \; \dfrac{V[\phi(t)]}{H(t)}.
\end{align*}
In fact, we can see that the integral is finite if the $\epsilon$-parameter has a quick-enough time variation. Finally, we can check what the integral $\Delta$ represents in special cases. In a de Sitter spacetime with a cosmological constant $\Lambda$, we can write it as
\begin{align*}
    \Delta = \sqrt{2} \, \dfrac{\sqrt{d-2}}{\sqrt{d-1}} \, \dfrac{1}{t_\Lambda} \, \lim_{t \to \infty} (t - t_0),
\end{align*}
where $\smash{t_\Lambda = 1 / \kappa_d \sqrt{\Lambda}}$ is the timescale associated to the vacuum-energy scale, so effectively $\Delta$ scales like the duration of the de Sitter spacetime. In a proper scaling cosmology, we find instead
\begin{align*}
    \Delta = [(d-1) \, p - 1] \, \lim_{t \to \infty} \mathrm{ln} \, \dfrac{t}{t_0},
\end{align*}
which is of course divergent (as an infinite $\Delta$ is a prerequisite for a late-time non-degenerate scaling-cosmology attractor), but in a less quick way than a de Sitter space. In fact, the simple behavior of scaling solutions also allow us to check the identity
\begin{align*}
    \kappa_d \, \int_{t_0}^\infty \de t \; \sqrt{\dot{\phi}_{* a} (t) \, \dot{\phi}_*^a (t)} = \dfrac{\sqrt{p} \sqrt{d-2}}{(d-1) \, p - 1} \, \Delta = \infty,
\end{align*}
which means that the distance travelled in field space in a scaling cosmology is divergent. This is not surprising since scaling-cosmology trajectories in field space are infinite straight lines (see subsec. \ref{subsubsec: time-measuring and quintessence-like scalars}) and the field-space metric is a Kronecker-delta.

\subsection{Analytic properties of scaling cosmologies} \label{subsec: analytic properties of scaling cosmologies}

Since for scaling solutions we know the field-space trajectories exactly, we can identify a number of analytic properties. Below we report such general features, reviewing and commenting on results from ref. \cite{Shiu:2023nph}.

\subsubsection{Scalar-potential directional derivative} \label{subsubsec: scalar-potential directional derivative}
In a scaling cosmology, scalars evolve with a logarithmic dependence on time as (see eq. (\ref{rank-condition scalar-field trajectories}), but also eq. (\ref{non-rank-condition scalar-field trajectories}))
\begin{equation} \label{scaling-cosmology scalar-field trajectory}
    \phi^a_* (t) = \phi^a_\infty + \dfrac{1}{\kappa_d} \, \alpha^a \; \mathrm{ln} \, \dfrac{t}{t_\infty},
\end{equation}
where the $\alpha^a$-slopes are easily seen by eq. (\ref{Friedmann eq.}) to be subject to the constraint $\smash{\alpha^a \alpha_a = (d-2) \, p}$. In other words, the kinetic energy evolves as
\begin{equation} \label{scaling-cosmology kinetic energy}
    \dfrac{1}{2} \, \dot{\phi}_{*a} \dot{\phi}_*^a = \dfrac{1}{2 \kappa_d^2 t^2} \, (d-2) p.
\end{equation}
Let $\smash{\theta_*^a = \alpha^a / \sqrt{\alpha_b \alpha^b}}$ be the unit vector following the field-space trajectory. In view of eq. (\ref{Friedmann eq. 1}), the scalar potential takes the form
\begin{equation} \label{scaling-cosmology potential}
    V(\phi_*) = \dfrac{1}{2 \kappa_d^2 t^2} \, \bigl[ (d-1) \, p - 1 \bigr] \, (d-2) \, p.
\end{equation}
As a comparison, in a de Sitter spacetime, the scalar potential -- i.e. the cosmological constant -- is related to the Hubble scale as $\smash{\kappa_d^2 \Lambda = (d-1) (d-2) H^2 /2}$. Moreover, in view of eq. (\ref{FRW-KG eq.}), the scalar-potential derivatives can be expressed as
\begin{equation} \label{scaling-cosmology potential gradient}
    - \dfrac{\der V}{\der \phi_{* a}} = \dfrac{1}{\kappa_d t^2} \, \bigl[ (d-1) \, p - 1 \bigr] \, \alpha^a.
\end{equation}
As an extra relationship, we also highlight that the scalar-field pressure $\smash{p = \dot{\phi}^a \dot{\phi}_a / 2 - V}$ and energy density $\smash{\rho = \dot{\phi}^a \dot{\phi}_a / 2 + V}$ give the on-shell equation of state
\begin{align*}
    w = \dfrac{p_*}{\rho_*} = - 1 + \dfrac{2}{(d-1) \, p}.
\end{align*}
Furthermore, we notice that we can write the scale factor as $\smash{a(t) = a_\infty \, (t/t_\infty)^{2 / [(d-1) (w+1)]}}$.

The expressions above are sufficient to express the $\epsilon$-parameter in terms of the scalar-potential directional derivative along the scalar-field trajectory. In particular, in view of eqs. (\ref{scaling-cosmology scalar-field trajectory}, \ref{scaling-cosmology potential}, \ref{scaling-cosmology potential gradient}), we have
\begin{equation} \label{scaling-cosmology directional gamma}
    \gamma_* = - \dfrac{1}{V(\phi_*)} \, \theta_*^a \dfrac{\der V}{\kappa_d \, \der \phi^a_*} (\phi_*) = \dfrac{2 \sqrt{\epsilon}}{\sqrt{d-2}}.
\end{equation}
This is an exact identity. In particular, the $\epsilon$-parameter is expressed via a directional derivative of the scalar potential: the kinetic energy does not enter directly in the expression, but the information of the non-zero field speed is stored within the field-space trajectory; the scalar potential instead appears only through the direction that is effectively experienced by the scalars. In fact, one may rotate the scalars and identify a single scalar $\smash{\tvarphi}$ that evolves over time: to be able to do so, however, one needs the details of all the scalar potentials, as reviewed in eq. (\ref{quintessence-field time evolution}).

In view of eqs. (\ref{scaling-cosmology potential}), we can also express the norm of the scalar-potential gradient as
\begin{equation} \label{scaling-cosmology potential gradient norm}
    \dfrac{1}{\kappa_d V (\phi_*)} \sqrt{\dfrac{\der V}{\der \phi_{* a}} \dfrac{\der V}{\der \phi_*^a} (\phi_*)} = \dfrac{2 \sqrt{\epsilon}}{\sqrt{d-2}} = \gamma_*.
\end{equation}
As will be discussed thoroughly in subsec. \ref{subsec: scalar-potential derivatives and acceleration}, the fact that this numerically coincides with the directional derivative of the scalar potential $\gamma_*$ defined in eq. (\ref{scaling-cosmology directional gamma}) is an accidental feature of scaling cosmologies. One may wonder why $\epsilon$ is proportional to the squared gradient norm for scaling cosmologies even when the kinetic term is non-negligible compared with the potential. This is a peculiarity of scaling solutions. Even though $\dot{H}$ receives a contribution from a sizable kinetic term, so does $H^2$ and the ratio $\epsilon = - \dot{H}/H^2$ still comes out proportional to the squared gradient norm.\footnote{All results here in subsubsec. \ref{subsubsec: scalar-potential directional derivative} apply to both solutions in eqs. (\ref{rank-condition scalar-field trajectories}, \ref{rank-condition scale-factor power}) and eqs. (\ref{non-rank-condition scalar-field trajectories}, \ref{non-rank-condition scale-factor power}). However, we do not have a proof that the latter are inevitable late-time attractors.}

\subsubsection{Time-measuring and quintessence-like scalars} \label{subsubsec: time-measuring and quintessence-like scalars}
One can always identify a single scalar that provides a measure of time. This is because all scalar-potential terms participate in the cosmological evolution with the same parametric weight to the total energy density. Indeed, as one can see explicitly, each term in the scalar potential of a scaling cosmology evolves as
\begin{align*}
    V_i \bigl[\phi_*^a(t)\bigr] = \Lambda_i \, \e^{- \kappa_d \gamma_{i a} \phi_\infty^a} \, \Bigl(\dfrac{t_\infty}{t}\Bigr)^2.
\end{align*}
In other words, whatever combination of scalar fields appears in each $V_i$-term, this provides a $\smash{-2 \, \mathrm{ln} \, (t/t_\infty)}$-behavior. So, for an arbitrary term $\smash{V_{i_0}}$, we can define a canonically-normalized scalar field via the redefinition
\begin{equation} \label{clock-field definition}
    \tilde{\gamma} \ttau = \gamma_{i_0 a} \phi^a,
\end{equation}
where the parameter $\smash{\tilde{\gamma}}$ and the field $\smash{\ttau}$ are defined by the $\mathrm{O}(n)$-rotation for the specific $\smash{\gamma_{i_0 a}}$-coefficients, and where the field evolves over time as
\begin{equation} \label{clock-field time evolution}
    \ttau_*(t) = \ttau_\infty + \dfrac{1}{\kappa_d} \dfrac{2}{\tilde{\gamma}} \; \mathrm{ln} \, \dfrac{t}{t_\infty}.
\end{equation}
In the remaining $m-1$ scalar-potential terms, one has linear combinations of the field $\smash{\ttau}$ and further $n-1$ canonically-normalized time-dependent scalar fields.

After the field-space trajectory $\smash{\phi_*^a}$ has been identified, we can also define a scalar $\smash{\tvarphi}$ that is aligned with the field trajectory in the moduli space. This can can be done via an $\mathrm{O}(n)$-rotation where $\smash{\tvarphi}$ is parallel to the vector $\smash{\theta_*^a}$ and the remaining fields $\smash{\check{\phi}^{\check{a}}}$ are orthogonal to it, with the $\smash{\check{a}}$-index not including $\smash{\tvarphi}$. All the scalar-potential terms can then be written as
\begin{align*}
    V_i = \Lambda_i \, \e^{- \kappa_d \gamma_* \tvarphi - \kappa_d \check{\gamma}_{i \check{a}} \check{\phi}^{\check{a}}},
\end{align*}
where $\smash{\gamma_*}$ is the directional derivative of eq. (\ref{scaling-cosmology directional gamma}) and the coefficients $\smash{\check{\gamma}_{i \check{a}}}$ are instead defined by the inverse rotation. By construction, the field $\smash{\tvarphi}$ evolves as
\begin{equation} \label{quintessence-field time evolution}
    \tvarphi_*(t) = \tvarphi_\infty + \dfrac{1}{\kappa_d} \dfrac{2}{\gamma_*} \; \mathrm{ln} \, \dfrac{t}{t_\infty},
\end{equation}
while the other fields are constants $\smash{\check{\phi}_*^{\check{a}} = \check{\phi}_\infty^{\check{a}}}$. All the fields $\smash{\check{\phi}_*^{\check{a}}}$ can be absorbed into redefinitions of the constants $\Lambda_i$, so as to have a total on-shell potential
\begin{align}
    V = \Lambda \, \e^{- \kappa_d \gamma_* \tvarphi_*}.
\end{align}
Neither the constant $\smash{\Lambda}$ nor the coefficient $\smash{\gamma_*}$ can be read off simply from the dimensional reduction of a single dominating term. Of course, we can also use $\smash{\tvarphi}$ to measure time instead of $\smash{\ttau}$. A representation of the field-space trajectory for scaling solutions is in fig. \ref{fig.: field-space trajectory}.

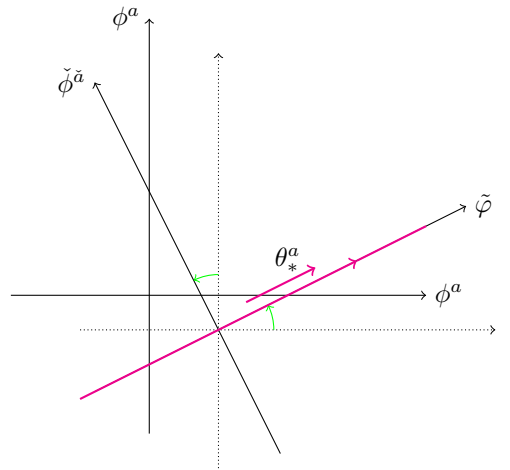
\begin{figure}[ht]
    \centering
    
    \begin{tikzpicture}[xscale=0.92,yscale=0.92,every node/.style={font=\normalsize}]

    \draw[->] (-2-1,0+0.5) -- (4-1,0+0.5) node[right]{$\phi^a$};
    \draw[->] (0-1,-2+0.5) -- (0-1,4+0.5) node[left]{$\phi^a$};

    \draw[->, densely dotted] (-2,0) -- (4,0);
    \draw[->, densely dotted] (0,-2) -- (0,4);

    \draw[->, rotate=atan(1/2)] (-2,0) -- (4,0) node[right]{$\tilde{\varphi}$};
    \draw[->, rotate=atan(1/2)] (0,-2) -- (0,4) node[left]{$\check{\phi}^{\check{a}}$};

    \draw[->, thin, green] (0.8,0) arc (0:atan(1/2):0.8);
    \draw[->, thin, green] (0,0.8) arc (90:90+atan(1/2):0.8);

    \draw[thick, magenta, ->] (-2,-1) -- (2,1);
    \draw[thick, magenta] (2,1) -- (3,1.5);

    \draw[thick, magenta, ->] (0.4,0.2+0.2) -- (1.4,0.7+0.2) node[above, black, pos=0.6] {$\theta_*^a$};
    
    \end{tikzpicture}
    \caption{A sketch of the field-space trajectory of scaling cosmologies; a rotation in field space always allows one to work with a single quintessence-like scalar.}
    \label{fig.: field-space trajectory}
\end{figure}

For a scaling solution $\smash{\phi_*^a(t)}$ of the form in eq. (\ref{scaling-cosmology scalar-field trajectory}), the initial conditions $\smash{\phi_\infty^a}$ are tied to the choice of the initial time $t_\infty$. In particular, plugging the solutions of eq. (\ref{scaling-cosmology scalar-field trajectory}) into eq. (\ref{FRW-KG eq.}), we get the compatibility condition
\begin{equation} \label{compatibility condition}
    \Biggl[\sum_{i=1}^m \Lambda_i \gamma_{ia} \, \e^{- \kappa_d \gamma_{i a} \phi_\infty^a} \Biggr] \, t_\infty^2 = \dfrac{1}{\kappa_d^2} \biggl[\dfrac{d-1}{\epsilon} - 1 \biggr] \, \alpha_a.
\end{equation}
For a fixed initial time $t_\infty$, the initial conditions $\smash{\phi_\infty^a}$ satisfy the above equation. Seen in the other way around, for a given set of initial conditions, one is choosing the initial time accordingly. For $\smash{\mathrm{rank} \, \gamma_{ia} = m}$, if $m=n$, the compatibility condition has a unique solution.

Notice that the system has an additional shift symmetry (nevertheless, the compatibility condition still has to be satisfied). For arbitrary real values $\kappa_d \xi^a \in \mathbb{R}$, the problem is invariant under the transformations \cite{Hartong:2006rt}
\begin{align*}
    \phi^a & \to \breve{\phi}^a = \phi^a + \xi^a, \\
    \Lambda_i & \to \breve{\Lambda}_i = \Lambda_i \, \e^{- \kappa_d \gamma_{i a} \xi^a}.
\end{align*}
In case $\smash{\mathrm{rank} \, \gamma_{ia} = m}$ and $m=n$, we can solve the system $\gamma_{ia} \xi^a = \omega_i$ for arbitrary values of $\omega_i$, which means that we can always rescale the $\Lambda_i$-terms to arbitrary values. Given a model, we always work in a fixed shift-symmetry gauge, i.e. for given values of $\smash{\phi_\infty^a}$ and $\Lambda_i$.

\subsubsection{Accidental gradient-flow trajectories} \label{subsubsec: accidental gradient-flow trajectories}

In this subsubsection, we observe an accidental feature of scaling cosmologies: the field-space trajectory of a scaling cosmology happens to be parallel to a gradient-flow trajectory of the scalar potential.

In essence, because the scaling solution is a straight line in field space, we can rotate the field basis in such a way as to have only one scalar that evolves over time, as described in subsubsec. \ref{subsubsec: time-measuring and quintessence-like scalars}. Because of this, along the late-time trajectory solving the field equations, a problem with a scaling cosmology as a late-time attractor is always equivalent to a 1-field 1-potential problem, for some universal scalar $\smash{\tilde{\varphi}}$. In fact, for solutions of the form in eq. (\ref{scaling-cosmology scalar-field trajectory}), we can write the second-order scalar-field derivative and the Hubble-friction terms in eq. (\ref{FRW-KG eq.}) as
\begin{align*}
    \ddot{\phi}^a_*(t) & = - \dfrac{1}{\kappa_d t^2} \, \alpha^a, \\
    (d-1) H \dot{\phi}^a_*(t) & = \dfrac{1}{\kappa_d t^2} \, \dfrac{d-1}{\epsilon} \, \alpha^a.
\end{align*}
Therefore, assuming that we have solved the field equations and that we know that they are a scaling cosmology, we may write the identities
\begin{equation} \label{scaling-solution gradient flow}
    \biggl[ 1 - \dfrac{\epsilon}{d-1} \biggr] \, (d-1) H \dot{\phi}^a_* = - \dfrac{\der V}{\der \phi_a} (\phi_*).
\end{equation}
As the overall factor in front of each $\smash{\dot{\phi}^a_*}$-term is common to all equations of the form in eq. (\ref{scaling-solution gradient flow}), the field-space trajectory is a gradient-flow trajectory of the potential. Therefore, the field-space trajectory happens to follow one of the steepest-descent directions of the potential.

A priori, the trajectory dicated by eq. (\ref{scaling-solution gradient flow}) is just one of the infinitely-many steepest-descent trajectories, as such trajectories depend critically on the initial conditions, while the scaling-cosmology trajectory is a line forever, which corresponds to a specific initial condition. It is only through the compatibility condition in eq. (\ref{compatibility condition}) that we see that the one gradient-flow trajectory that one would select is exactly the scaling solution at all times. This can be seen because the initial conditions must be such that the compatibility condition in eq. (\ref{compatibility condition}) is fulfilled. Also, as a general comment, gradient-flow trajectories can be unstable and one may incur in following the wrong one, leading to a wrong corner of the moduli space. Therefore, we consider the fact that at late times scaling cosmologies are parallel to a gradient-flow trajectory as a noteworthy feature, but we cannot consider calculating the gradient-flow trajectory as a trustworthy method to find the scaling-cosmology attractor. All these points can be seen explicitly in an example in sec. \ref{subsec: comments on gradient flow}.

One more obvious but crucial fact is the following. Without knowledge of the actual fully-fledged solution to the cosmological equations, including the scale-factor and scalar-field time dependences, knowledge of the field-space trajectory is empty: one does not even know for sure the $\epsilon$-parameter. Indeed, several solutions may correspond to the same field-space trajectory, but with a different time dependence. In particular, as is well known, the steepest-descent trajectory also characterizes the solutions to the slow-roll equations. Nonetheless, it is apparent that, even though eq. (\ref{scaling-solution gradient flow}) and eq. (\ref{slow-roll FRW-KG eq.}) happen to give the same late-time field-space trajectory, this is with a completely different physical motivation. Of course, the physical interpretation is not the only thing that changes: scaling cosmologies, which are solutions to eqs. (\ref{FRW-KG eq.}) -- i.e. eq. (\ref{scaling-solution gradient flow}) -- feature scalars evolving with a $\smash{\mathrm{ln} \, (t/t_0)}$-behavior, whereas solutions to eqs. (\ref{slow-roll FRW-KG eq.}), if they exist at all, have a different time dependence (or, if they have the same time dependence, then the slow-roll approximation is inconsistent). Proper scaling cosmologies are however the only correct universal late-time attractor, by subsec. \ref{subsec: scaling cosmologies as late-time attractors}.

Also, we stress that the results stemming from eq. (\ref{scaling-solution gradient flow}) have nothing to do with accelerated expansion since they are true independently of the value of $\epsilon$, whether it be $\epsilon < 1$ or $\epsilon \geq 1$. All is a consequence of the fact that scaling solutions are straight lines in field space, since otherwise the proportionality factor between the second- and first-derivative terms would be different for each field.

Finally, we stress that in all cases in which we cannot determine whether a scaling cosmology is the universal late-time attractor, we have shown that the slow-roll approximation is hard to justify in all models with non-zero $\smash{(\hat{\gamma}_\infty)^2}$-coefficient, thus leaving no analytical evidence for any statement on steepest-descent trajectories as late-time approximate cosmological solutions for multi-field multi-exponential potentials.

\subsubsection{Curvature-induced potentials and scale separation} \label{subsubsec: curvature-induced potentials and scale separation}

As we have seen, in a scaling cosmology we can measure time through a scalar field $\smash{\ttau}$ that we can define based on an arbitrary scalar-potential term. Because we can express time as $\smash{t = t_\infty \, \e^{ \kappa_d \tilde{\gamma} \, (\ttau_* - \ttau_\infty) / 2}}$, we can also express the Hubble scale as
\begin{equation} \label{Hubble-scale evolution}
    l_H = \dfrac{1}{H} = l_\infty \, \e^{\frac{1}{2} \, \kappa_d \tilde{\gamma} \ttau_*},
\end{equation}
where $\smash{l_\infty = (t_\infty / q) \, \e^{- \frac{1}{2} \, \kappa_d \tilde{\gamma} \ttau_\infty}}$ is the initial value. On the other hand, in an isotropic compactification, the Kaluza-Klein scale evolves as
\begin{equation} \label{KK-scale evolution}
    l_{\mathrm{KK}, d} = \Bigl( \dfrac{4 \pi}{g_s^2} \Bigr)^{\frac{1}{d-2}} \, l_{\p, d} \, \e^{- \frac{\kappa_d \tdelta}{\sqrt{d-2}} + \frac{\kappa_d \tsigma}{\sqrt{10-d}}}.
\end{equation}
In particular, the moduli dependence of the Kaluza-Klein scale is proportional to the moduli dependence of the potential induced by a non-trivial curvature $\smash{\breve{R}}$, as computed in terms of a fiducial metric with a string-size fiducial volume. In fact, we can write
\begin{align*}
    V_{R} = - \dfrac{l_s^2 \breve{R}}{2 \kappa_d^2} \, \dfrac{1}{l_{\mathrm{KK}, d}^2}.
\end{align*}
Because this potential, if present, can be used to measure the Hubble scale, we can define a scalar $\ttau$ through $\smash{2 \tdelta / \sqrt{d-2} - 2 \tsigma / \sqrt{10-d} = - \tilde{\gamma} \ttau}$ and find the ratio
\begin{equation}
    \dfrac{l_H^2}{l_{\mathrm{KK}, d}^2} = \dfrac{1}{- l_s^2 \breve{R}} \dfrac{1}{\xi} \biggl[ (d-1) - \dfrac{1}{p} \biggr] \, (d-2),
\end{equation}
where $\xi = V / V_R$ is the order-1 ratio between the on-shell total potential $V$ and the curvature-induced potential $V_R$. Since there is no parametric difference in the Hubble and Kaluza-Klein scales, the theory is by no means genuinely $d$-dimensional. Before going on, we notice that, technically, we can have scale separation with curvature-induced potentials, at the cost of acceleration: we have no parametric scale separation unless in the degenerate non-proper solution: here $1/p=d-1$, so there is scale separation. However, this is an extreme scenario since it is the case where the scaling solution would asymptotically require $V = 0$. If the internal curvature is trivial, the details of the dilaton and radion time evolution are necessary to assess the ratio.

\subsection{Scalar-potential derivatives and acceleration} \label{subsec: scalar-potential derivatives and acceleration}
Here we characterize generally the relationship between scalar-potential derivatives and the $\epsilon$-parameter. This allows us to comment on the exceptionality of scaling cosmologies, whose scalar-potential derivatives are related to the $\epsilon$-parameter through eqs. (\ref{scaling-cosmology directional gamma}, \ref{scaling-cosmology potential gradient norm}).

\subsubsection{Scalar-potential directional derivative}
Let the field-space trajectory $\smash{\phi_\star^a = \phi_\star^a(t)}$ be a solution to eqs. (\ref{FRW-KG eq.}, \ref{Friedmann eq. 1}, \ref{Friedmann eq. 2}); note that the subscript ``$_\star$'' denotes a solution to the field equations, whereas the subscript ``$_*$'' denotes specifically a scaling solution. As the field velocity $\smash{\dot{\phi}_\star^a(t)}$ is tangent to the trajectory, we define the general scalar-potential directional derivative as
\begin{equation} \label{directional derivative}
    \gamma_\star(t) = - \dfrac{1}{V (\phi)} \, \theta^a \dfrac{\der V}{\kappa_d \, \der \phi^a} [\phi_\star(t)],
\end{equation}
where $\smash{\theta^a[\phi] = \dot{\phi}^a / \sqrt{\dot{\phi}_b \dot{\phi}^b}}$ is the normalized field velocity. In view of eqs. (\ref{FRW-KG eq.}, \ref{Friedmann eq. 1}), we are able to express the directional derivative as
\begin{align*}
    \gamma_\star(t) = \dfrac{2 \sqrt{\epsilon(t)}}{\sqrt{d-2}} + \dfrac{\beta}{\kappa_d V} [\phi_\star(t)],
\end{align*}
where $\beta = \beta[\phi]$ is defined as the functional
\begin{equation*}
    \beta[\phi] = \dfrac{\de}{\de t} \sqrt{\dot{\phi}^a \dot{\phi}_{a}} + \dfrac{\sqrt{\epsilon[\phi]}}{\sqrt{d-2}} \, \dot{\phi}^a \dot{\phi}_{a}.
\end{equation*}
Therefore, in general, the directional derivative is proportional to the acceleration-parameter square root up to an additional term that depends on the total kinetic energy $\smash{T[\phi] = \dot{\phi}^a \dot{\phi}_{a}/2}$. By eq. (\ref{Friedmann eq.}), we can also write
\begin{align*}
    \beta = - \dfrac{\sqrt{d-2}}{2 \kappa_d} H^2 \eta \sqrt{\epsilon}.
\end{align*}
A suggestive way to express the directional derivative is therefore \cite{Achucarro:2018vey}
\begin{equation} \label{directional derivative and acceleration parameter}
    \gamma_\star = \dfrac{2 \sqrt{\epsilon}}{\sqrt{d-2}} \, \biggl[ 1 - \dfrac{1}{2} \, \dfrac{\eta}{(d-1) - \epsilon} \biggr].
\end{equation}
It turns out that the $\beta$-dependent term vanishes exactly for scaling solutions, for which $\eta=0$. However, unless one knows that $\eta=0$, the meaning of $\gamma_\star$ is related to the $\epsilon$-parameter in a generally complicated way.

\subsubsection{Scalar-potential gradient norm}
In the literature, a function that is often discussed is the norm of the normalized scalar-potential gradient, or potential gradient norm for short, which is defined as
\begin{equation} \label{scalar-potential gradient norm}
    \gamma = \dfrac{1}{\kappa_d V} \sqrt{\dfrac{\der V}{\der \phi_a} \dfrac{\der V}{\der \phi^a}}.
\end{equation}
Indeed, this is related to the $\epsilon$-parameter in the slow-roll approximation, where it happens to be a tiny positive number, and it appears in the formulation of the de Sitter conjecture in the Swampland Program \cite{Obied:2018sgi, Ooguri:2018wrx}, where it is conjectured to be at least an order-1 number. For a field trajectory that solves eqs. (\ref{FRW-KG eq.}, \ref{Friedmann eq. 1}, \ref{Friedmann eq. 2}), we can express $\gamma$ as
\begin{equation} \label{scalar-potential gradient norm and acceleration parameter}
    \gamma^2 = \dfrac{4 \epsilon}{d-2} + \dfrac{(d-2) \, \psi}{\kappa_d^4 V^2},
\end{equation}
where $\psi$ is the function(al)
\begin{align*}
    \psi = \dfrac{\dot{H}^3}{H^2} \!-\! \dfrac{1}{2} \dddot{H} + (d-1) \biggl[ \dot{H}^2 \!-\! \dfrac{3}{2} H \ddot{H} \biggr] + \dfrac{\dot{\phi}^a \dot{\phi}^b}{(d-2)} \dfrac{\der^2 V}{\der \phi^a \der \phi^b} & .
\end{align*}
It is easy to see that the scalar-potential directional derivative in eq. (\ref{directional derivative}) is always a lower bound for the scalar-potential gradient norm in eq. (\ref{scalar-potential gradient norm}). Indeed, for any field trajectory, we have the inequality \cite{Andriot:2022brg}
\begin{equation} \label{scalar-potential gradient inequality}
    \ab \gamma_\star[\phi] \ab \leq \gamma[\phi],
\end{equation}
which is a direct consequence of the vector inequality $\smash{\ab \theta^a \der_a V \ab \leq \sqrt{\theta_b \theta^b} \sqrt{\der_a V \der^a V} = \sqrt{\der_a V \der^a V}}$. Again, it is possible to check that in a scaling cosmology the $\psi$-dependent term vanishes. However, unless one is certain of this fact, the potential gradient norm does not represent a reliable measure of the $\epsilon$-parameter.

For a given field-space trajectory $\smash{\phi_\star^a = \phi_\star^a(t)}$ that provides a solution to eqs. (\ref{FRW-KG eq.}, \ref{Friedmann eq. 1}, \ref{Friedmann eq. 2}), besides the normalized tangent vector $\smash{\theta_\star^a = \theta^a[\phi_\star]}$ defined before, one can define a normal vector $\smash{\nu_\star^a = - \dot{\theta}_\star^a / \sqrt{\dot{\theta}_{\star b} {\dot{\theta}_\star}{}^b}}$ too, also normalized to unity. A non-geodesity factor $\Omega$ can be defined through the identity
\begin{align*}
    \dot{\theta}_\star^a = - \Omega \, \nu_\star^a.
\end{align*}
Such a factor thus represents the magnitude of the tangent-vector time derivative. In other words, $\Omega$ measures the rate of turning of the field-space trajectory. One is then able to show that the scalar-potential gradient norm can also be expressed as \cite{Achucarro:2018vey}
\begin{equation} \label{scalar-potential gradient norm and non-geodesity factor}
        \gamma^2 = \gamma_\star^2 + \dfrac{4 \epsilon}{d-2} \, \dfrac{1}{[(d-1) - \epsilon]^2} \dfrac{\Omega^2}{H^2}.
\end{equation}
This is a general expression for the potential gradient norm $\gamma$, which does not rely on any approximation. One can easily see the relationship between eq. (\ref{scalar-potential gradient norm and acceleration parameter}) and eq. (\ref{scalar-potential gradient norm and non-geodesity factor}). In particular, we observe that the $\psi$-parameter can be expressed as the sum of a term proportional to $\Omega$ and another one proportional to $\eta$. As the $\psi$-correction is easier to compute than the term $\Omega$, our eq. (\ref{scalar-potential gradient norm and acceleration parameter}) nicely complements the intuitively clear meaning of the non-geodesity factor appearing in eq. (\ref{scalar-potential gradient norm and non-geodesity factor}), as expressed in ref. \cite{Achucarro:2018vey}. Again, for the scaling cosmologies in eqs. (\ref{power-law scale factor}, \ref{rank-condition scalar-field trajectories}, \ref{rank-condition scale-factor power}) it is possible to see that the potential gradient norm provides a measure of the $\epsilon$-parameter because $\Omega=0$ and $\eta=0$. However none of these terms is generally vanishing, in principle.

\subsubsection{Comparison with slow-roll scenarios}

Under the slow-roll approximation, one studies the cosmological evolution in terms of eqs. (\ref{slow-roll FRW-KG eq.}, \ref{slow-roll Friedmann eq. 2}, \ref{slow-roll Friedmann eq. 1}), which give the solutions $\smash{\phi_{\mathrm{sr}}^a}$ and $\smash{H_{\mathrm{sr}}}$. If such approximation is self-consistent, then by eqs. (\ref{slow-roll Friedmann eq. 2}, \ref{slow-roll Friedmann eq. 1}) one can approximate the $\epsilon$-parameter as $\smash{\epsilon_{\mathrm{sr}} = - \dot{H}_{\mathrm{sr}} / H_{\mathrm{sr}}^2 \simeq \epsilon_V}$, where $\epsilon_V$ reads
\begin{align*}
    \epsilon_V = \dfrac{d-2}{4} \, \gamma_{\mathrm{sr}}^2
\end{align*}
and $\smash{\gamma_{\mathrm{sr}} = \gamma(\phi_{\mathrm{sr}})}$ is the norm of the scalar-potential gradient in eq. (\ref{scalar-potential gradient norm}) computed on a slow-roll solution. Formally, the expression of the parameter $\epsilon_V$ in terms of $\smash{\gamma_{\mathrm{sr}}}$ is the same as eqs. (\ref{scaling-cosmology directional gamma}, \ref{scalar-potential gradient norm and acceleration parameter}), but it applies to a restricted situation, i.e. that of slow-roll time evolution.

Under the slow-roll approximation, both $\epsilon_{\mathrm{sr}} \ll 1$ and $\eta_{\mathrm{sr}} \ll 1$. From eqs. (\ref{directional derivative and acceleration parameter}, \ref{scalar-potential gradient norm and non-geodesity factor}), one finds that to leading order in the slow-roll expansion $\smash{\gamma_{\star \mathrm{sr}} \simeq 2 \sqrt{\epsilon_{\mathrm{sr}}} / \sqrt{d-2}}$ and $\smash{\gamma_{\mathrm{sr}} \simeq \bigl[ 2 \sqrt{\epsilon_{\mathrm{sr}}}/ \sqrt{d-2} \bigr] \bigl[ 1 + \Omega_{\mathrm{sr}}^2 / [(d-1)^2 H_{\mathrm{sr}}^2 \bigr]}$, as observed in refs. \cite{Hetz:2016ics, Achucarro:2018vey}. However, because late-time cosmologies are not generally compatible with slow-roll evolution, one should never rely on these approximated expression for quintessence-like models, but rather on the complete eqs. (\ref{directional derivative and acceleration parameter}) and (\ref{scalar-potential gradient norm and acceleration parameter}, \ref{scalar-potential gradient norm and non-geodesity factor}).

\subsubsection{Relevance for scaling cosmologies} \label{subsubsec: relevance for scaling cosmologies}
Incidentally, if $\gamma$ is computed on a scaling cosmology, then it numerically corresponds to the parameter $\gamma_*$ in eq. (\ref{scaling-cosmology directional gamma}), namely $\gamma = \gamma_*$ as in eq. (\ref{scaling-cosmology potential gradient norm}).

Conceptually, the slow-roll approximation is however not generically correct for scaling solutions. First, eqs. (\ref{scaling-cosmology kinetic energy}, \ref{scaling-cosmology potential}) show that the kinetic and the potential energy are of the same order of magnitude, unless the power $p$ is arbitrarily large, and in particular they evolve with the same parametric time dependence, as discussed in subsubsec. \ref{subsubsec: scalar-potential directional derivative}. Second, one can check that the second-derivative term, the Hubble-friction term and the scalar-potential derivative in the scalar field equations all fall off with the same parametric dependence on time, as discussed in subsubsec. \ref{subsubsec: accidental gradient-flow trajectories}. Therefore, unless we know that we are working with a scaling cosmology -- which however is not analytically known to be a late-time attractor, in general, as discussed in subsec. \ref{subsec: scaling cosmologies as late-time attractors} -- computing $\epsilon$ through the $\gamma$-parameter is not justified. This is coherent with a more general obstruction for the slow-roll approximation that is typical of generic multi-exponential potentials \cite{Shiu:2023nph}.

To conclude, we comment on the geodesity of scaling-cosmology trajectories. As apparent from eq. (\ref{scalar-potential gradient norm and non-geodesity factor}), the fact that the trajectory is geodesic -- namely that it has $\Omega = 0$ -- is not enough to reliably compute the $\epsilon$-parameter through the potential gradient norm $\gamma$. One also needs knowledge of $\eta$.

\subsection{Further scaling solutions} \label{subsec: further scaling solitions}

Following the classification of ref. \cite{Collinucci:2004iw}, further scaling solutions exist besides those discussed in subsec. \ref{subsec: scaling cosmologies}; more details are in appendix \ref{app: late-time cosmological attractors}: see lemmas \ref{lemma: (X,Y)-plane critical points (rank<m)} and \ref{lemma: (X,Y)-plane critical points (x^2=0)}.

Here we discuss the case in which the rank of the $\gamma_{ia}$-matrix is smaller than the number of terms in the scalar potential, i.e. $\mathrm{rank} \, \gamma_{ia} < m$. One always has $\mathrm{rank} \, \gamma_{ia} < m$ if the scalar-potential terms outnumber the scalars, i.e. $n < m$, and therefore they tend to constrain their dynamics into stationary points -- which in this case can exist --, thus explaining the non-generality of rolling-scalar solutions. One can also have $\mathrm{rank} \, \gamma_{ia} < m$ if $n \geq m$, but then the same linear combinations of scalars appear in multiple scalar-potential terms.

In these cases rolling-scalar solutions are not general. If $\mathrm{rank} \, \gamma_{ia} = r < m$, let $\smash{\gamma_{\iota a}}$ denote $r$ linearly-independent vectors, for $\iota=1, \dots, r$, and let the remaining vectors be expressed as $\smash{\gamma_{\iota' a} = \sum_{\iota=1}^r \lambda_{\iota' \iota} \gamma_{\iota a}}$, for $\smash{\iota' = r+1, \dots, m}$. It is only if $\smash{\sum_{\iota=1}^m \lambda_{\iota' \iota} = 1}$ that scaling solutions exist. This can be understood intuitively as a condition in which different scalar-potential terms actually have the same exponentials, including the overall scaling, thus effectively reducing to the cases with the rank-condition in place.\footnote{To see this point with an example, one can consider the trivial 1-field 2-term scalar potential $\smash{V = \Lambda_1 \, \e^{- \kappa_d \gamma_1 \phi} + \Lambda_2 \, \e^{- \kappa_d \gamma_2 \phi}}$. This admits a scaling solution only if $\gamma_1 = \gamma_2$, which means that this is secretly a 1-field 1-term potential with $\mathrm{rank} \, \gamma_{ia} = m = 1$.}

If they exist, in these cases rolling solutions are mathematically analogous to the ones above and, given the matrix $\smash{N_{ab} = \sum_{i=1}^m \gamma_{i a} \gamma_{i b}}$, they read
\begin{equation} \label{non-rank-condition scalar-field trajectories}
    \phi^a_* (t) = \phi^a_0 + \dfrac{2}{\kappa_d} \, (N^{-1})^{ab} \sum_{i=1}^m \gamma_{ib} \, \mathrm{ln} \, \dfrac{t}{t_0},
\end{equation}
with a scale-factor power
\begin{equation} \label{non-rank-condition scale-factor power}
    p = \dfrac{4 \, \delta_{a b}}{d-2} \, (N^{-1})^{a c} (N^{-1})^{b d} \sum_{i=1}^m \!\gamma_{i c} \sum_{j=1}^m \!\gamma_{i d}.
\end{equation}

Based on the discussion above, we expect that, in the asymptotic region of moduli space, scalar potentials with $\mathrm{rank} \, \gamma_{ia} < m$ are analogous to those with the rank-condition in place. In fact, physical intuition suggests that, if $\mathrm{rank} \, \gamma_{ia} < m$, then there exists a number $\tilde{m} < m$ of asymptotically dominating scalar-potential terms, with effectively only $n = \tilde{m}$ scalars participating to the dynamics, and with $\mathrm{rank} \, \gamma_{ia} = \tilde{m}$.\footnote{To see this point with an example, one can again consider the 1-field 2-term scalar potential $\smash{V = \Lambda_1 \, \e^{- \kappa_d \gamma_1 \phi} + \Lambda_2 \, \e^{- \kappa_d \gamma_2 \phi}}$. If $\gamma_1 < \gamma_2$, then it is reasonable to expect that asymptotically the second term is subdominant, thus effectively leaving one just the first term, which admits a scaling solution.} However, at the moment we can only speculate that this is the case and leave an analytic treatment for future work \cite{criticalpoints2}.

\section{Late-time cosmologies and the Swampland Program} \label{sec: late-time cosmologies and the swampland}

Our analytic results allow us to test explicitly some of the conjectures in the Swampland Program. In a broad sense, our comments of the analytic properties of scaling cosmologies in subsec. \ref{subsec: analytic properties of scaling cosmologies} also constitute 
concrete grounds to verify Swampland claims; here however we perform more explicit comparisons.

\subsection{Dilaton obstruction to acceleration} \label{subsec: dilaton obstruction to acceleration}

As discussed in ref. \cite{Shiu:2023nph}, in string-theoretic constructions, the bound in eq. (\ref{epsilon bound}) -- and, by extension, the bound in eq. (\ref{optimal epsilon bound}) -- can acquire a very restrictive form due to the universal structure of dilaton couplings. Due to the relevance of such a constraint, we revisit it here in more detail.

In string-theoretic constructions, all $d$-dimensional exponential scalar potentials are generated by a generic string-frame action contribution of the form
\begin{align*}
    S = - \int_{\mathrm{X}_{1,9}} \bigl[ A_{r} \wedge \star_{1,9} A_r \bigr] \, K_{10,r}(\sigma) \, \e^{- \chi_{\mathrm{E}} \Phi}.
\end{align*}
Here $\smash{\Phi = \sdil + \mathrm{ln} \, g_s}$ is the shifted 10-dimensional dilaton (we define $\smash{\Phi = \sdil}$ if the dilaton is not fixed) and $\smash{\chi_{\mathrm{E}}}$ is the Euler number that weighs the string-coupling perturbative order via the string-worldsheet topology;\footnote{In the RR-sector, one has to set $\smash{\chi_{\mathrm{E}}=0}$ because the sources of RR-fields are D-branes and not fundamental strings \cite{Polchinski:1995mt}. This does not affect any of the arguments below.} furthermore, $A_r$ is an $r$-form with just internal components and $\smash{K_{10,r} (\sigma) = \Lambda_{10,r} \, \e^{- k \sigma}}$ is a function of the string-frame volume which can emerge in string-loop terms. After a dimensional reduction, the $d$-dimensional Einstein-frame action reads
\begin{align*}
    S = - \int_{\mathrm{X}_{1,d-1}} \tilde{*}_{1,d-1} \, \Lambda \, \e^{\kappa_d \gamma_{\tdelta} \tdelta - \kappa_d \gamma_{\tsigma} \tsigma},
\end{align*}
where $\smash{\tdelta}$ and $\smash{\tsigma}$ are the $d$-dimensional dilaton and string-frame radion after canonical normalization, respectively, with the couplings
\begin{align*}
     \gamma_{\tdelta} & = \frac{d}{\sqrt{d-2}} - \dfrac{1}{2} \chi_{\mathrm{E}} \sqrt{d-2}, \\
     \gamma_{\tsigma} & = \Bigl(1 - \dfrac{1}{2} \chi_{\mathrm{E}}\Bigr) \, \sqrt{10-d} - \dfrac{2r + k}{\sqrt{10-d}}.
\end{align*}
Although the radion coefficient is model-dependent and thus makes it hard to draw general conclusions, the dilaton coefficient is universal. In particular, as a consequence of the bound $\smash{\chi_{\mathrm{E}} (\mathrm{S}^2) \leq 2}$ on the string-frame coupling of the 10-dimensional dilaton, the $d$-dimensional dilaton always appears with a $\gamma$-coefficient such that
\begin{equation} \label{dilaton-gamma lower bound}
    \gamma_{\tdelta}{}^2 \geq \dfrac{4}{d-2},
\end{equation}
Because $\smash{(\gamma_\infty)^2 \geq \gamma_{\tdelta}{}^2}$, this rules out late-time accelerated expansion in all string-theoretic constructions with positive-definite scalar-potential terms in which the $d$-dimensional dilaton is one of the rolling scalars; theories with two potentials of opposite sign are also incompatible with cosmic acceleration. Of course, all these considerations also extend to models in which the late-time solution is a scaling cosmology.

To conclude, we emphasize that here we are referring to the canonically-normalized $d$-dimensional dilaton $\smash{\tdelta}$, which is different from the canonically-normalized 10-dimensional dilaton $\smash{\tsdil}$. In an isotropic compactification, such fields are related by the field-space rotation
\begin{align*}
    \tvec{\tdelta}{\tsigma} =
    \dfrac{\sqrt{2}}{4}
    \left(\! \begin{array}{cc}
        \sqrt{d-2} & - \sqrt{10-d} \\[1.0ex]
        \sqrt{10-d} & \sqrt{d-2}
    \end{array} \!\right)
    \tvec{\tsdil}{\tomega},
\end{align*}
where $\smash{\tsigma}$ and $\smash{\tomega}$ are the canonically-normalized string-frame and Einstein-frame radions, respectively. Indeed in the $\smash{(\tdelta, \tsigma)}$- and $\smash{(\tsdil, \tomega)}$-bases there are no kinetic mixings, which makes both of them convenient choices. There could be situations in which one basis is easier to deal with than the other one.\footnote{For instance, in Calabi-Yau flux compactifications with $\smash{N_4=1}$ supersymmetry, it turns out that the first basis is a more suitable choice in type-IIA theories and that the second basis is a more suitable choice in type-IIB theories (see e.g. refs. \cite{Grimm:2004uq, Grimm:2004ua}). However, the fact that the 10-dimensional dilaton $\smash{\tsdil}$ in type-IIB supergravity is stabilized  by RR- and NSNS-3-form fluxes (and that it naturally appears in a supermultiplet) does not mean that we could not equivalently work in the $\smash{(\tdelta,\tsigma)}$-basis in a type-IIB construction, after a field-basis rotation. We found that the $\epsilon$-parameter is more naturally bounded in the $\smash{(\tdelta,\tsigma)}$-basis.} However, here we remain agnostic as to any underlying higher-dimensional structure and we do not make any assumptions on which fields are being stabilized. Also, we point out a simple but important observation. If one or more fields are stabilized in a basis, then it may be that the field $\smash{\tdelta}$ is also partially stabilized: for instance, this is the case in type-IIB flux compactification in which the 10-dimensional dilaton $\smash{\tsdil}$ is stabilized by NSNS- and RR-fluxes.

\subsection{Partial moduli stabilization} \label{subsec: partial moduli stabilization}

There are situations in which scaling solutions feature the stabilization of some of the scalars. In particular, one situation where this happens is when a subset of the fields appear with the same exponents in all the scalar-potential terms, and the complementary subset of fields appears instead with different exponents: in this case, the latter fields happen to be stabilized. Below we discuss this scenario and comment on possible subtleties.

Let the scalar potential be
\begin{align*}
    V = \sum_{i = 1}^m \Lambda_i \, \e^{- \kappa_d [\gamma_{i \overline{a}} \phi^{\overline{a}} + \tilde{\gamma}_{\tilde{a}} \phi^{\tilde{a}}]},
\end{align*}
where we split the fields into the subset of scalars $\smash{\phi^{\tilde{a}}}$ with an identical coupling $\smash{\gamma_{i \tilde{a}} = \tilde{\gamma}_{\tilde{a}}}$ in all the terms (note that this may be the case after a field-space rotation for this subset of scalars) from the other scalars $\smash{\phi^{\overline{a}}}$. If the scaling solutions in eq. (\ref{rank-condition scalar-field trajectories}) are late-time attractors, then we have the field-space trajectories
\begin{align*}
    \phi_*^{\tilde{a}}(t) & = \phi_0^{\tilde{a}}(t) + \dfrac{1}{\kappa_d} \, \dfrac{d-2}{2 \epsilon} \, \tilde{\gamma}^{\tilde{a}} \, \mathrm{ln} \, \dfrac{t}{t_0}, \\
    \phi_*^{\overline{a}}(t) & = \phi_0^{\overline{a}}(t) + \dfrac{2}{\kappa_d} \, \biggl[ \sum_{i=1}^m {\gamma_i}^{\overline{a}} \lambda^i \biggr] \, \mathrm{ln} \, \dfrac{t}{t_0}.
\end{align*}
Because all scaling cosmologies feature scalar-potential terms all falling over time as $V_i(t) = V_i(t_0) \, (t_0/t)^2$, for each $i$-index we must have $\smash{\gamma_{ia} \sum_{j=1}^m {\gamma_j}^{a} \lambda^j = 1}$, i.e.
\begin{align*}
     \dfrac{d-2}{4 \epsilon} \, \tilde{\gamma}_{\tilde{a}} \tilde{\gamma}^{\tilde{a}} + \gamma_{i \overline{a}} \sum_{j=1}^m {\gamma_j}^{\overline{a}} \lambda^k = 1.
\end{align*}
This can be just seen as the linear system of equations $\smash{\gamma_{i \overline{a}} \chi^{\overline{a}} = \xi}$, where we defined $\smash{\chi^{\overline{a}} = \sum_{j=1}^m {\gamma_j}^{\overline{a}} \lambda^k}$ and $\smash{\xi = 1 - [(d-2)/4] \, \tilde{\gamma}_{\tilde{a}} \tilde{\gamma}^{\tilde{a}} / \epsilon}$ for simplicity. If $\xi \neq 0$, there cannot be a solution since not all of the $i$-rows of the system are identical. Therefore, we must have $\xi = 0$ and $\chi^{\overline{a}}=0$. In other words, this means that we have
\begin{align*}
    \epsilon = \dfrac{d-2}{4} \, (\tilde{\gamma})^2
\end{align*}
and $\smash{\sum_{j=1}^m {\gamma_j}^{\overline{a}} \lambda^k = 0}$, or equivalently
\begin{align*}
    \phi_*^{\tilde{a}}(t) & = \phi_0^{\tilde{a}}(t) + \dfrac{2}{\kappa_d} \, \dfrac{\tilde{\gamma}^{\tilde{a}}}{(\tilde{\gamma})^2} \, \mathrm{ln} \, \dfrac{t}{t_0}, \\
    \phi_*^{\overline{a}}(t) & = \phi_0^{\overline{a}}(t).
\end{align*}

This is a useful result, in that it allows us to distinguish two cases, which we discuss below.
\begin{itemize}
    \item If, for at least one field $\smash{\phi^{\overline{a}}}$, not all of the exponential couplings are zero and the non-vanishing ones are of the same sign -- i.e. if $\gamma_{ia} \geq 0$ with $\sum_{i=1}^m \gamma_{ia} > 0$ or $\gamma_{ia} \leq 0$ with $\sum_{i=1}^m \gamma_{ia} < 0$ for at least one $\overline{a}$-index -- , then our assumption of convergence to a scaling cosmology is unjustified in the first place. Indeed, in this case at least some of the $\lambda^i$-values must be negative for the identity $\smash{\sum_{i=1}^m {\gamma_i}^{\overline{a}} \lambda^i = 0}$ to be possible and therefore our proof of convergence does not apply. In short, in this case our assumption that $\lambda^i \geq 0$ for each $i$-index but with the additional requirement that $\smash{\sum_{i=1}^m \lambda^i > 0}$ is incompatible with the consequences that would follow. In physical terms, here we have a hierarchy of potentials, which is in fact known to be incompatible with a scaling cosmology: we can speculate that at late times the attractor is approximately provided by the scaling cosmology associated to the truncated potential in which the subdominant terms are removed.
    \item If all fields $\smash{\phi^{\overline{a}}}$ involve exponential couplings that are not all non-negative or non-positive, then a priori there is no incompatibility, and the stabilization discussed above is in principle possible (though one must still verify case by case that the conditions for a scaling cosmology are indeed met). In physical terms, this is reasonable since it corresponds to potentials that have exponentials of both signs, which create a competition of terms resulting in a scalar-potential valley. Scalars that see this this competition get stabilized at the minimum of such a valley. In fact, the scalar potentials that we are considering here are the ones in the basis where the bound in eq. (\ref{optimal epsilon bound}) is already optimal. As we observed in subsubsec. \ref{subsubsec: time-measuring and quintessence-like scalars}, one can always describe a scaling cosmology in terms of just one rolling scalar, with the others being stabilized. Note that this mechanism does not provide a loophole to stabilize the $d$-dimensional dilaton, since the latter has exponential couplings of a definite sign.
\end{itemize}

Although the situation considered here may sound very peculiar, there are string-theoretic arguments that show that the assumption that some fields appear with the same power in all scalar-potential terms is not so restrictive: as shown in subsec. \ref{subsec: dilaton obstruction to acceleration}, the $d$-dimensional dilaton appears with a power depending only on the string perturbative order at which the potential is generated. Therefore, a trivial conclusion is that one can in principle stabilize all moduli but the $d$-dimensional dilaton whenever all scalar-potential terms are generated at the same string-loop level. However, when the dilaton remains an unstabilized field and the potential is positive-definite, we already know that acceleration is not possible as the dilaton potential is too steep.

\subsection{Coupling convex hull and cosmic acceleration} \label{subsec: coupling convex hull and cosmic acceleration}

Here we provide a simple but rigorous interpretation of the relationship of late-time cosmic acceleration with the coupling convex hull. Because the bounds in eqs. (\ref{epsilon bound}, \ref{optimal epsilon bound}), the bound in eq. (\ref{alternative epsilon bound}) and the scaling-cosmology attractors discussed in subsec. \ref{subsec: scaling cosmologies as late-time attractors} require increasing numbers of assumptions and tie to the convex-hull discussion in different ways, we discuss each of them separately.

An important clarification is in order. As we show below, the convex hull of the exponential couplings provides a graphical interpretation of the measure of the $\epsilon$-parameter. This offers a nice visualization of our analytic results. However, this does not per se relate directly to the de Sitter conjecture, since the obstruction to flat(tish) potentials does not come from the exponential potentials themselves. It is only after one considers the bounds imposed by string theory on the couplings -- namely on the properties of the convex hull -- that one can possibly confirm the conjecture by taking advantage of our general results.

\subsubsection{Universal acceleration bound}

Here, we formulate our bound on late-time acceleration in eq. (\ref{optimal epsilon bound}) in terms of the convex hull of the exponential couplings. All results below are independent of the number $n$ of scalar fields and the number $m$ of scalar-potential terms.

For simplicity, we start by considering a scalar potential with only positive coefficients $\Lambda_i > 0$; as we will see below, this assumption is however not fundamental for our results as long as the total potential is positive. On the one hand, out of the $m \cdot n$ couplings $\gamma_{ia}$ of a given multi-field multi-exponential potential, one can define $m$ different $n$-dimensional vectors $\mu_i$ with components $\smash{(\mu_i)_a = \gamma_{ia}}$. Then, one can define the exponential-coupling convex hull, namely the parameter-space hypersurface $\mathrm{CH}(\lbrace \mu_i \rbrace_{i=1}^m) = \bigl\lbrace \nu_a = \sum_{i=1}^m \! \xi_i (\mu_i)_a: \; (\xi_i)_{i=1}^m \in (\mathbb{R}_0^+)^m, \; \sum_{i=1}^m \xi_i = 1 \bigr\rbrace$, and compute the distance of the latter from the origin as
\begin{equation} \label{convex-hull distance}
    \mu_{\mathrm{CH}} = \inf_{\nu \in \mathrm{CH}} \sqrt{\nu_a \nu^a}.
\end{equation}
The way the convex hull arises here is analogous to how it arises in the multi-field generalizations \cite{Cheung:2014vva, Rudelius:2014wla,Rudelius:2015xta,Brown:2015iha,Brown:2015lia,Palti:2017elp, Calderon-Infante:2020dhm, Etheredge:2022opl} of the Weak Gravity Conjecture \cite{Arkani-Hamed:2006emk} and the Distance Conjecture \cite{Ooguri:2006in}. A different notion of distance $\smash{\tilde{\mu}}$ will be discussed below, but for now the definition of $\smash{\mu_{\mathrm{CH}}}$ suffices. On the other hand, as an analytic result, we have our bound in eq. (\ref{optimal epsilon bound}). For any late-time solution, the $\epsilon$-parameter is bounded from below by a number proportional to the squared length of the vector constructed with the minimum non-negative exponential couplings for each of the fields, as optimized through a possible field-space rotation. By direct inspection, our optimal analytic lower bound for the $\epsilon$-parameter does in fact coincide with the convex-hull distance, i.e.\footnote{
A similar-looking relationship appears in the convex-hull formulation of the de Sitter conjecture \cite{Calderon-Infante:2022nxb}, but
with a different meaning. In eq. (\ref{convex-hull epsilon bound}), we bound the $\epsilon$-parameter, which provides the proper criterion for cosmic acceleration. Instead, the convex-hull de Sitter conjecture concerns the potential gradient norm, which is generally unrelated to the $\epsilon$-parameter, as we explained.}
\begin{equation} \label{convex-hull epsilon bound}
    \epsilon \geq \dfrac{d-2}{4} \, (\hat{\gamma}_{\infty})^2 = \dfrac{d-2}{4} \, \mu_{\mathrm{CH}}^2.
\end{equation}
This result is analytic and not conjectural. A representation of the late-time cosmic acceleration bound for positive-definite scalar potentials in terms of the coupling convex hull is depicted in figs. \ref{fig.: optimal acceleration bound}, \ref{fig.: 2by2 convex-hull picture 1}, \ref{fig.: 2by2 convex-hull picture 2} and \ref{fig.: 2by2 convex-hull picture 6}.

\begin{figure}[ht]
    \centering
    \begin{tikzpicture}[xscale=0.92,yscale=0.92,every node/.style={font=\normalsize}]
    
    \node[align=left] at (6,3.5){$\phi^a = {\color{cyan} \phi^1}, {\color{magenta} \phi^2}$ \\[1.0ex] $\gamma_{ia} = \matr{\gamma_{1 \color{cyan} 1}}{\gamma_{1 \color{magenta} 2}}{\gamma_{2 \color{cyan} 1}}{\gamma_{2 \color{magenta} 2}}$};
    
    \draw[->] (-0.8,0) -- (6,0) node[right]{$\gamma_{\infty 1}$};
    \draw[->] (0,-0.8) -- (0,4.5) node[left]{$\gamma_{\infty 2}$};
 
    \draw[orange, thick] (5,2) -- (3,4);
    
    \draw[->, thick, teal] (0,0) -- (5,2) node[right,black]{$\mu_1$};
    \draw[densely dotted] (5,2) -- (5,0) node[below]{$\gamma_{1 \color{cyan} 1}$};
    \draw[densely dotted] (5,2) -- (0,2) node[left]{$\gamma_{1 \color{magenta} 2}$};

    \draw[->, thick, teal] (0,0) -- (3,4) node[above,black]{$\mu_2$};
    \draw[densely dotted] (3,4) -- (3,0) node[below]{$\gamma_{2 \color{cyan} 1}$};
    \draw[densely dotted] (3,4) -- (0,4) node[left]{$\gamma_{2 \color{magenta} 2}$};

    \draw[->, thick, purple] (0,0) -- (3.5,3.5) node[below right,pos=0.55]{$\!\! (\hat{\gamma}_\infty)^2 = \mu_{\mathrm{CH}}^2 = (\tilde{\mu})^2$};
    
    \end{tikzpicture}
    \caption{A representation of the acceleration bound for a positive-definite potential in which the vector orthogonal to the convex-hull hyperplane intersects the convex hull itself too. In this case, it is apparent that $\smash{(\hat{\gamma}_\infty)^2} = \mu_{\mathrm{CH}}^2 = (\tilde{\mu})^2$.}
    \label{fig.: 2by2 convex-hull picture 1}
\end{figure}
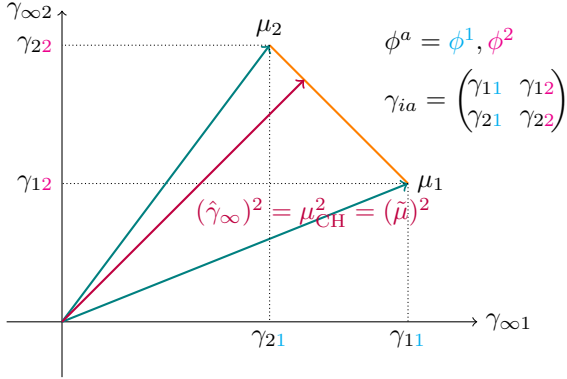

\begin{figure}[ht]
    \centering
    \begin{tikzpicture}[xscale=0.92,yscale=0.92,every node/.style={font=\normalsize}]
    
    \node[align=left] at (6,3.5){$\phi^a = {\color{cyan} \phi^1}, {\color{magenta} \phi^2}$ \\[1.0ex] $\gamma_{ia} = \matr{\gamma_{1 \color{cyan} 1}}{\gamma_{1 \color{magenta} 2}}{\gamma_{2 \color{cyan} 1}}{\gamma_{2 \color{magenta} 2}}$};
    
    \draw[->] (-0.8,0) -- (6,0) node[right]{$\gamma_{\infty 1}$};
    \draw[->] (0,-0.8) -- (0,4.5) node[left]{$\gamma_{\infty 2}$};

    \draw[densely dotted, magenta!65!orange] (6,4/3) -- (1,14/3);
    \draw[orange, thick] (5,2) -- (3.5,3);
    
    \draw[->, thick, teal] (0,0) -- (5,2) node[right,black]{$\mu_1$};
    \draw[densely dotted] (5,2) -- (5,0) node[below]{$\gamma_{1 \color{cyan} 1}$};
    \draw[densely dotted] (5,2) -- (0,2) node[left]{$\gamma_{1 \color{magenta} 2}$};

    \draw[->, thick, purple] (0,0) -- (3.5,3) node[above,black]{$\mu_2$} node[below right, pos=0.7] {$\!\! (\hat{\gamma}_\infty)^2 = \mu_{\mathrm{CH}}^2$};
    \draw[densely dotted] (3.5,3) -- (3.5,0) node[below]{$\gamma_{2 \color{cyan} 1}$};
    \draw[densely dotted] (3.5,3) -- (0,3) node[left]{$\gamma_{2 \color{magenta} 2}$};

    \draw[->, thick, magenta!65!orange] (0,0) -- (32/13,48/13) node[right,pos=0.7]{$(\tilde{\mu})^2$};
    \draw[densely dotted, magenta!65!orange] (32/13,48/13) -- (32/13,0) node[below]{$\tilde{\mu}_1$};
    \draw[densely dotted, magenta!65!orange] (32/13,48/13) -- (0,48/13) node[left]{$\tilde{\mu}_2$};
    
    \end{tikzpicture}
    \caption{A representation of the acceleration bound for a positive-definite potential in which the vector orthogonal to the convex-hull hyperplane does not intersects the convex hull itself. In this case, it is apparent that $\smash{(\hat{\gamma}_\infty)^2 = \mu_{\mathrm{CH}}^2 \geq  (\tilde{\mu})^2}$.}
    \label{fig.: 2by2 convex-hull picture 2}
\end{figure}
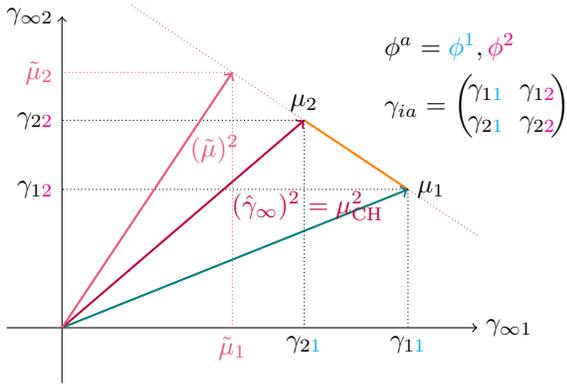

As discussed in the comments to eqs. (\ref{epsilon bound}, \ref{optimal epsilon bound}), we can also analytically bound late-time cosmic acceleration in the presence of negative terms. This involves an appropriately defined $\smash{(\hat{\gamma}_\infty)^2}$-term. As far as the coupling convex hull is concerned, if we consider an overall positive potential which however includes $\Lambda_i$-terms of both signs, we define the convex hull just through the coefficients of the positive-definite potential terms. In this case, the identity appearing in eq. (\ref{convex-hull epsilon bound}) is no longer general, as one can see by direct inspection. Rather, we have $\smash{\mu_{\mathrm{CH}}^2 \geq (\hat{\gamma}_\infty)^2}$. This means that the construction of a coupling convex hull is no longer informative per se of the bound on cosmic acceleration. One can check that, even if the convex hull were defined involving the couplings of the negative-definite terms too, then one would similarly be unable to individuate a universal convex-hull criterion for cosmic acceleration. Of course, what is physically and mathematically relevant is the maximal bound in eq. (\ref{optimal epsilon bound}), which can still be visualized easily in coupling space through the $\smash{(\hat{\gamma}_\infty)^2}$-term. For a representation of these scenarios, see figs. \ref{fig.: acceleration bound 3}, \ref{fig.: 2by2 convex-hull picture 3}, \ref{fig.: 2by2 convex-hull picture 4} and \ref{fig.: 2by2 convex-hull picture 5}.

\begin{figure}[ht]
    \centering
    \begin{tikzpicture}[xscale=0.92,yscale=0.92,every node/.style={font=\normalsize}]
    
    \node[align=left] at (5,2.5){$\phi^a = {\color{cyan} \phi^1}, {\color{magenta} \phi^2}$ \\[1.0ex] $\gamma_{i_\pm a} = (\gamma_{\pm \color{cyan} 1}, \gamma_{\pm \color{magenta} 2})$ \\
    $\hphantom{\gamma_{i_\pm a}} = (\Gamma_{\pm \color{cyan} 1}, \Gamma_{\pm \color{magenta} 2})$};
    
    \draw[->, ultra thin, gray!60!white] (-0.8,0) -- (5.5,0) node[right]{$\gamma_{\infty 1}$};
    \draw[->, ultra thin, gray!60!white] (0,-1.6) -- (0,3.5) node[left]{$\gamma_{\infty 2}$};

    \draw[->, rotate=-atan(5/14)] (-0.8,0) -- (5.5,0) node[right]{$\hat{\gamma}_{\infty 1}$};
    \draw[->, rotate=180-atan(5/14)] (0,-1.6) -- (0,2.5) node[left]{$\hat{\gamma}_{\infty 2}$};
 
    \draw[dashed, orange] (4/3,2) -- (7/2,-5/4);
    \draw[->, thick, magenta!65!orange] (0,0) -- (24/13,16/13) node[right]{$(\tilde{\mu})^2$};
    \draw[densely dotted, magenta!65!orange] (4,-2) -- (7/2,-5/4);
    \draw[densely dotted, magenta!65!orange] (2/3,3) -- (4/3,2);
    
    \draw[->, thick, green] (0,0) -- (4/3,2) node[right,black]{};
    \draw[densely dotted, gray!60!white] (4/3,2) -- (4/3,0) node[below]{$\Gamma_{- \color{cyan!30!white} 1}$};
    \draw[densely dotted, gray!60!white] (4/3,2) -- (0,2) node[left]{$\gamma_{- \color{magenta!30!white} 2}$};

    \draw[->, thick, teal] (0,0) -- (7/2,-5/4) node[above,black]{};
    \draw[densely dotted, gray!60!white] (7/2,-5/4) -- (7/2,0) node[above]{$\gamma_{+ \color{cyan!30!white} 1}$};
    \draw[densely dotted, gray!60!white] (7/2,-5/4) -- (0,-5/4) node[below right]{$\Gamma_{+ \color{magenta!30!white} 2}$};

    \node[purple, above right] at (7/2,-5/4){$(\hat{\gamma}_\infty)^2 = \mu_{\mathrm{CH}}^2$};
    
    \end{tikzpicture}
    \caption{A representation of the acceleration bound for a potential with terms of both signs. The convex hull is defined only by the positive-definite potential terms, in teal; the green line refers to the negative-definite terms. The dotted orange line represents the would-be convex hull if all potential terms were positive-definite. In this case, $\smash{(\hat{\gamma}_\infty)^2 = \mu_{\mathrm{CH}}^2 \geq (\tilde{\mu})^2}$.}
    \label{fig.: 2by2 convex-hull picture 3}
\end{figure}
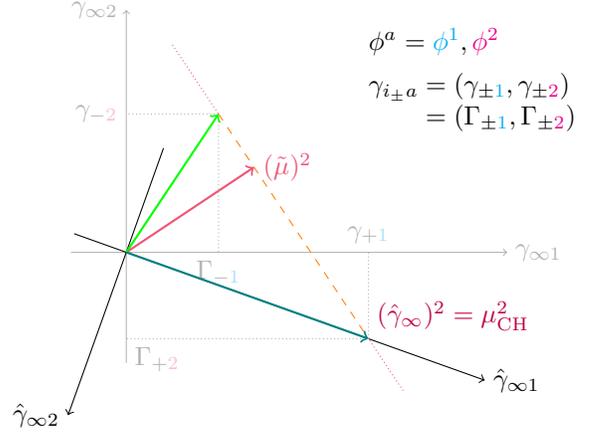

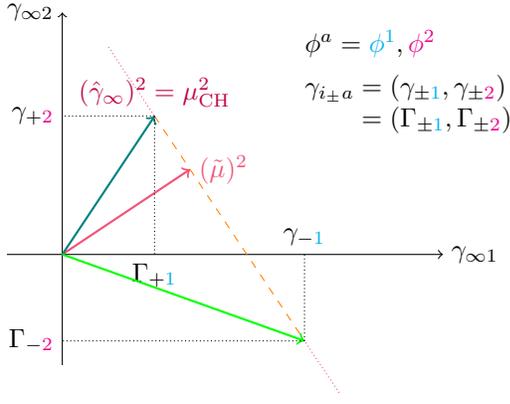
\begin{figure}[ht]
    \centering
    \begin{tikzpicture}[xscale=0.92,yscale=0.92,every node/.style={font=\normalsize}]
    
    \node[align=left] at (5,2.5){$\phi^a = {\color{cyan} \phi^1}, {\color{magenta} \phi^2}$ \\[1.0ex] $\gamma_{i_\pm a} = (\gamma_{\pm \color{cyan} 1}, \gamma_{\pm \color{magenta} 2})$ \\
    $\hphantom{\gamma_{i_\pm a}} = (\Gamma_{\pm \color{cyan} 1}, \Gamma_{\pm \color{magenta} 2})$};
    
    \draw[->] (-0.8,0) -- (5.5,0) node[right]{$\gamma_{\infty 1}$};
    \draw[->] (0,-1.6) -- (0,3.5) node[left]{$\gamma_{\infty 2}$};
 
    \draw[dashed, orange] (4/3,2) -- (7/2,-5/4);
    \draw[->, thick, magenta!65!orange] (0,0) -- (24/13,16/13) node[right]{$(\tilde{\mu})^2$};
    \draw[densely dotted, magenta!65!orange] (4,-2) -- (7/2,-5/4);
    \draw[densely dotted, magenta!65!orange] (2/3,3) -- (4/3,2);
    
    \draw[->, thick, teal] (0,0) -- (4/3,2) node[above, purple]{$(\hat{\gamma}_\infty)^2 = \mu_{\mathrm{CH}}^2$};
    \draw[densely dotted] (4/3,2) -- (4/3,0) node[below]{$\Gamma_{+ \color{cyan} 1}$};
    \draw[densely dotted] (4/3,2) -- (0,2) node[left]{$\gamma_{+ \color{magenta} 2}$};

    \draw[->, thick, green] (0,0) -- (7/2,-5/4) node[above,black]{};
    \draw[densely dotted] (7/2,-5/4) -- (7/2,0) node[above]{$\gamma_{- \color{cyan} 1}$};
    \draw[densely dotted] (7/2,-5/4) -- (0,-5/4) node[left]{$\Gamma_{- \color{magenta} 2}$};
    
    \end{tikzpicture}
    \caption{A representation of the acceleration bound, for a potential with terms of both signs, with $\smash{(\hat{\gamma}_\infty)^2 = \mu_{\mathrm{CH}}^2 \geq (\tilde{\mu})^2}$.}
    \label{fig.: 2by2 convex-hull picture 4}
\end{figure}

\begin{figure}[ht]
    \centering
    \begin{tikzpicture}[xscale=0.92,yscale=0.92,every node/.style={font=\normalsize}]
    
    \node[align=left] at (5,2.5){$\phi^a = {\color{cyan} \phi^1}, {\color{magenta} \phi^2}$ \\[1.0ex] $\gamma_{i_\pm a} = (\gamma_{\pm \color{cyan} 1}, \gamma_{\pm \color{magenta} 2})$ \\
    $\hphantom{\gamma_{i_\pm a}} = (\Gamma_{\pm \color{cyan} 1}, \Gamma_{\pm \color{magenta} 2})$};
    
    \draw[->] (-0.8,0) -- (5.5,0) node[right]{$\gamma_{\infty 1}$};
    \draw[->] (0,-2.0) -- (0,3) node[left]{$\gamma_{\infty 2}$};
 
    \draw[dashed, orange] (7/3,1/2) -- (7/2,-5/4);
    \draw[->, thick, purple] (0,0) -- (24/13,16/13) node[right]{$(\hat{\gamma}_\infty)^2 = (\tilde{\mu})^2$};
    \draw[densely dotted, magenta!65!orange] (4,-2) -- (7/2,-5/4);
    \draw[densely dotted, magenta!65!orange] (4/3,2) -- (7/3,1/2);
    
    \draw[->, thick, teal] (0,0) -- (7/3,1/2) node[right]{$\mu_{\mathrm{CH}}^2$};
    \draw[densely dotted] (7/3,1/2) -- (7/3,0) node[below]{$\gamma_{+ \color{cyan} 1}$};
    \draw[densely dotted] (7/3,1/2) -- (0,1/2) node[left]{$\gamma_{+ \color{magenta} 2}$};

    \draw[->, thick, green] (0,0) -- (7/2,-5/4) node[above,black]{};
    \draw[densely dotted] (7/2,-5/4) -- (7/2,0) node[below right]{$\Gamma_{- \color{cyan} 1}$};
    \draw[densely dotted] (7/2,-5/4) -- (0,-5/4) node[below right]{$\Gamma_{- \color{magenta} 2}$};
    
    \end{tikzpicture}
    \caption{A representation of the acceleration bound, for a potential with terms of both signs, with $\smash{\mu_{\mathrm{CH}}^2 \geq (\hat{\gamma}_\infty)^2 = (\tilde{\mu})^2}$.}
    \label{fig.: 2by2 convex-hull picture 5}
\end{figure}
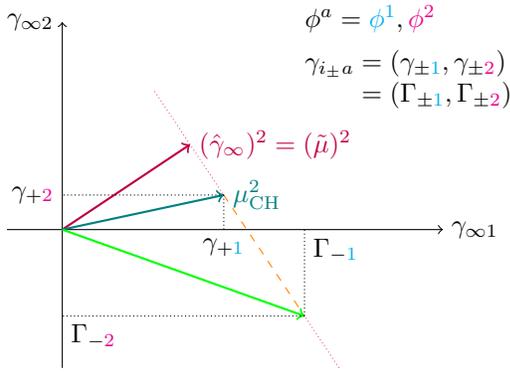

\subsubsection{Alternative acceleration bound}

Again, we begin with a scalar potential with only positive coefficients $\Lambda_i > 0$. In case the scalars outnumber the potential terms, namely if $n \geq m$, the bound in eq. (\ref{alternative epsilon bound}) allows us to further characterize the relationship of the coupling convex hull with late-time acceleration. Given the coupling matrix $\gamma_{ia}$, we have $m$ $n$-dimensional vectors $\smash{(\mu_{i})_a = \gamma_{ia}}$. If $n \geq m$, the convex hull $\mathrm{CH}(\lbrace \mu_i \rbrace_{i=1}^m) = \Pi$ of these $m$ vectors is at most $(n-1)$-dimensional, and so is the hyperplane $\smash{\tilde{\Pi} \supset \Pi}$ on which the convex hull lies. Once $\tilde{\Pi}$ is at most $(n-1)$-dimensional, we can find a unique vector $\smash{\tilde{\mu}_a}$ that is perpendicular to the hyperplane $\smash{\tilde{\Pi}}$ and that extends from the origin to the hyperplane. By definition, for any vector $\nu_a \in \tilde{\Pi}$, we can write the orthogonality condition $\smash{\tilde{\mu}_a (\tilde{\mu}^a - \nu^a) = 0}$, which can be expressed as $\smash{\nu_a \tilde{\mu}^a = (\tilde{\mu})^2}$. As by construction any of the $m$ vectors $\smash{(\mu_{i})_a}$ is one of the vectors $\smash{\nu^a}$, we conclude that the vector $\smash{\tilde{\mu}_a}$ is a solution to the defining equation $\smash{(\mu_{i})_a {\overline{\gamma}}^a = (\overline{\gamma})^2}$. Therefore, we can write the inequalities
\begin{equation}
    (\hat{\gamma}_{\infty})^2 = \mu_{\mathrm{CH}}^2 \geq (\tilde{\mu})^2.
\end{equation}
A complete saturation of the inequalities, namely the condition $\smash{(\hat{\gamma}_{\infty})^2 = \mu_{\mathrm{CH}}^2 = (\tilde{\mu})^2}$, takes place when the distance of the coupling convex hull from the origin coincides with the distance of the convex-hull hyperplane from the origin. A graphical interpretation of the bound for positive-definite potential terms is depicted in figs. \ref{fig.: 2by2 convex-hull picture 1} and \ref{fig.: 2by2 convex-hull picture 2}.

As a conclusive observation, we note that the discussion above evidences that the bound in eq. (\ref{optimal epsilon bound}) is never weaker than the alternative bound in eq. (\ref{alternative epsilon bound}), taking $\smash{(\tilde{\mu})^2 = (\overline{\gamma}_\infty)^2}$ (though there may be other solutions to eq. (\ref{alternative epsilon bound}), in principle). Nevertheless, such alternative bound may be easier to compute analytically. Furthermore, identifying scenarios in which $\smash{(\hat{\gamma}_\infty)^2 > (\tilde{\mu})^2}$ may be helpful for two reasons. First, these cases manifestly feature a subset of dominating potentials, as can be seen in the basis that optimizes the bound of eq. (\ref{optimal epsilon bound}). Second, such scenarios cannot correspond to scaling cosmologies, according to the discussion of partial moduli stabilization in subsec. \ref{subsec: partial moduli stabilization} and of scaling cosmologies in terms of the coupling convex hull in subsubsec. \ref{subsubsec: coupling convex hull for scaling cosmologies}. For these reasons, it makes sense to keep the bound of eq. (\ref{alternative epsilon bound}) into consideration even though it can be weaker than the bound in eq. (\ref{optimal epsilon bound}).

If we have an overall positive potential that nonetheless has terms of both signs, to check the bound in eq. (\ref{alternative epsilon bound}) we still use the vector of square length $\smash{(\tilde{\mu})^2}$. Such a length is indeed unchanged since it stems from a geometric condition on all the couplings, and therefore it also involves the coefficients of the negative-definite potentials (in appendix \ref{app: late-time cosmological attractors}, see corollary \ref{corollary: (X)^2-bound} and remarks \ref{remark: Xbar{X}-bound sign independence} and \ref{remark: alternative f-bound}). Again, the simple notion of convex hull by itself is not always helpful anymore, since physically the bounds of eqs. (\ref{epsilon bound}, \ref{optimal epsilon bound}) and (\ref{alternative epsilon bound}) are formulated in terms of $\smash{(\hat{\gamma}_\infty)^2}$ and $\smash{(\tilde{\mu})^2}$, respectively, but we can only see that
\begin{equation}
    \mu_{\mathrm{CH}}^2 \geq (\hat{\gamma}_\infty)^2 \geq (\tilde{\mu})^2.
\end{equation}
Rather, what is helpful is the notion of the convex-hull hyperplane -- with the convex hull being defined here with respect to all the couplings --, since it allows one to always write at least the alternative bound in eq. (\ref{alternative epsilon bound}). It also remains true that, if $\smash{(\tilde{\mu})^2 = (\overline{\gamma}_\infty)^2}$, the bound in eq. (\ref{optimal epsilon bound}) is always the optimal one. A series of examples is in figs. \ref{fig.: 2by2 convex-hull picture 3}, \ref{fig.: 2by2 convex-hull picture 4} and \ref{fig.: 2by2 convex-hull picture 5}.

The lesson we learn from this discussion is the following. For positive-definite scalar potentials, the coupling convex hull provides a general visual interpretation of the lowest bound on the $\epsilon$-parameter. Instead, for scalar potentials involving terms of both signs -- even if the total potential is positive -- the guiding principles should be the bounds in eqs. (\ref{epsilon bound}, \ref{optimal epsilon bound}) and (\ref{alternative epsilon bound}), since the distance from the origin to the convex hull may induce an inaccurate estimation of the lowest bound.

\subsubsection{Scaling cosmologies} \label{subsubsec: coupling convex hull for scaling cosmologies}

If the exponential couplings are such that the late-time solutions are scaling cosmologies, based on the discussion of subsec. \ref{subsec: scaling cosmologies as late-time attractors}, then we can do more as we have analytic knowledge of the time evolution of each operator. In particular, let us consider an $m$-field $m$-term potential -- that is: we assume $m=n$. In this case, the $\smash{(\mu_i)_a}$-components make up $m$ different $m$-dimensional real vectors, i.e. $\mu_i \in \mathbb{R}^m$ for $i=1,\dots,m$, and the corresponding convex hull is contained in the $(m-1)$-dimensional hyperplane $\smash{\tilde{\Pi}}$ that passes through each point $\mu_i$. Such hyperplane is defined by the equation
\begin{align*}
    \tilde{\Pi}: \quad \sigma^a \nu_a + \rho = 0,
\end{align*}
for some constants $\sigma^a, \rho \in \mathbb{R}$. A way to determine such constants is to impose the defining condition that all points $\mu_i$ do belong to the plane, which amounts to solving the $m$ equations $\smash{\sigma^a (\mu_i)_a + \rho = 0}$. In vector notation, we can write $\smash{\gamma \, \underline{\sigma} = - \rho \, \underline{e}}$, where we used the definition of the vectors $\smash{\smash{\underline{\mu}}_i}$ and we defined the vector $\smash{\underline{e}}$ with components $\smash{e_i = 1}$. A solution to this is $\smash{\underline{\sigma} = - \rho \, \gamma^{-1} \underline{e}}$. Such a vector $\smash{\underline{\sigma}}$ is orthogonal to the hyperplane, and it is thus easy to compute the squared distance of the latter from the origin to be $\smash{(\tilde{\mu})^2 = d^2(\underline{0}, \Pi) = \rho^2 / \underline{\sigma}^T \cdot \underline{\sigma}}$, i.e.
\begin{align*}
     (\tilde{\mu})^2 = \dfrac{1}{\displaystyle \sum_{i=1}^m \sum_{j=1}^m \delta_{ab} \, (\gamma^{-1})^{a i} (\gamma^{-1})^{b j}}.
\end{align*}
Remarkably, such a square length $\smash{(\tilde{\mu})^2}$ provides an exact measure of the scaling-solution $\epsilon$-parameter of eq. (\ref{rank-condition scale-factor power}), being $\smash{\epsilon = [(d-2)/4] \, (\tilde{\mu})^2}$, whenever the hyperplane distance vector $\smash{\tilde{\mu}}$ intersects the convex hull itself, which implies $\smash{(\hat{\gamma}_\infty)^2 = (\tilde{\mu})^2 = \mu_{\mathrm{CH}}^2}$. In other words, in such cases, the bound in eq. (\ref{optimal epsilon bound}) is saturated by the scaling-cosmology attractor, with
\begin{equation}
    \epsilon = \dfrac{d-2}{4} \, (\hat{\gamma}_{\infty})^2 = \dfrac{d-2}{4} \, \mu_{\mathrm{CH}}^2 = \dfrac{d-2}{4} \, (\tilde{\mu})^2.
\end{equation}
However, it is not necessarily the case that the hyperplane distance vector $\smash{\tilde{\mu}}$ intersects the convex hull itself. If it is not, nevertheless, we are able to prove that we do not automatically have convergence to the scaling solution. In fact, in view of eqs. (\ref{rank-condition scalar-field trajectories}, \ref{rank-condition scale-factor power}), and by rotating coordinates if needed, the identities\footnote{To check this explicitly, one can compare with the expression of the critical points in app. \ref{app: late-time cosmological attractors}: see eqs. (\ref{eq.: X-critical point}, \ref{eq.: Y-critical point}).}
\begin{align*}
    \tilde{\mu}_a \sum_{i=1}^m \lambda^i = \sum_{i=1}^m (\mu_i)_a \lambda^i
\end{align*}
show clearly that, because in these cases there is at least one $a$-direction such that $\smash{\tilde{\mu}_a < (\mu_i)_a}$ for all $i$-terms, there must be at least one $i$-index with $\lambda^i < 0$: hence, there is no contradiction, since the scaling solutions of eqs. (\ref{power-law scale factor}, \ref{rank-condition scalar-field trajectories}, \ref{rank-condition scale-factor power}) are not guaranteed to be attractors. In such cases, the lower bounds in eqs. (\ref{optimal epsilon bound}, \ref{alternative epsilon bound}) are the strongest analytic bounds that we are aware of. Physical intuition suggests that in such cases the potentials corresponding to negative $\lambda^i$s would be asymptotically subdominant. This is because there is no field rotation such that the exponents for all of the fields would not be subdominant. Therefore, effectively one can truncate the potential, getting rid of all the terms corresponding to $i$-indices with negative $\lambda^i$. This is also suggested by earlier perturbative analyses and numerical checks \cite{Guo:2003eu, Guo:2003rs, Collinucci:2004iw, Hartong:2006rt}. A graphical interpretation of the bound for positive-definite scalar potential terms is, again, depicted in figs. \ref{fig.: 2by2 convex-hull picture 1} and \ref{fig.: 2by2 convex-hull picture 2}.

Of course, knowledge of the actual late-time solution extends our information on the $\epsilon$-parameter just based on the grounds of eqs. (\ref{optimal epsilon bound}, \ref{alternative epsilon bound}). A final graphical example is in fig. \ref{fig.: 2by2 convex-hull picture 6}.

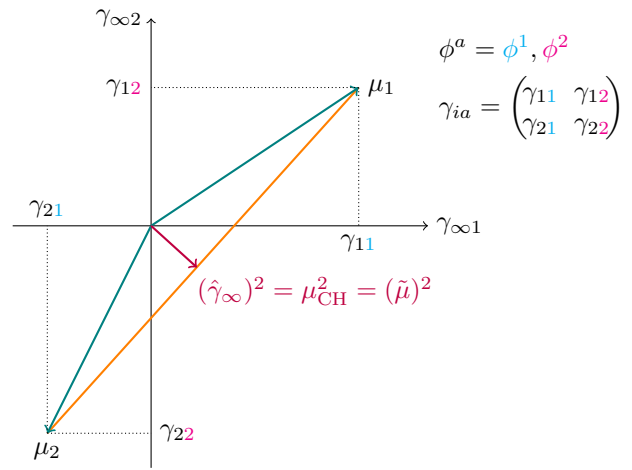
\begin{figure}[ht]
    \centering
    \begin{tikzpicture}[xscale=0.92,yscale=0.92,every node/.style={font=\normalsize}]
    
    \node[align=left] at (5.5,2){$\phi^a = {\color{cyan} \phi^1}, {\color{magenta} \phi^2}$ \\[1.0ex] $\gamma_{ia} = \matr{\gamma_{1 \color{cyan} 1}}{\gamma_{1 \color{magenta} 2}}{\gamma_{2 \color{cyan} 1}}{\gamma_{2 \color{magenta} 2}}$};
    
    \draw[->] (-2,0) -- (4,0) node[right]{$\gamma_{\infty 1}$};
    \draw[->] (0,-3.5) -- (0,3) node[left]{$\gamma_{\infty 2}$};
    
    \draw[orange, thick] (-1.5,-3) -- (3,2);
    
    \draw[->, thick, teal] (0,0) -- (3,2) node[right,black]{$\mu_1$};
    \draw[densely dotted] (3,2) -- (3,0) node[below]{$\gamma_{1 \color{cyan} 1}$};
    \draw[densely dotted] (3,2) -- (0,2) node[left]{$\gamma_{1 \color{magenta} 2}$};

    \draw[->, thick, teal] (0,0) -- (-1.5,-3) node[below,black]{$\mu_2$};
    \draw[densely dotted] (-1.5,-3) -- (-1.5,0) node[above]{$\gamma_{2 \color{cyan} 1}$};
    \draw[densely dotted] (-1.5,-3) -- (0,-3) node[right]{$\gamma_{2 \color{magenta} 2}$};

    \draw[->, thick, purple] (0,0) -- (120/181,-109/181) node[below right]{$\!\!(\hat{\gamma}_\infty)^2 = \mu_{\mathrm{CH}}^2 = (\tilde{\mu})^2$};
    
    \end{tikzpicture}
    \caption{A representation of the acceleration bound for a positive-definite multi-field multi-exponential potential in which all the coupling vectors lie in different hyperquadrants. If $\smash{\mathrm{rank} \, \gamma_{ia} = m}$, although $\smash{(\gamma_\infty})^2=0$, in the optimal basis the bound is non-trivial, with $\smash{(\hat{\gamma}_\infty})^2>0$. Moreover, due to our knowledge of a late-time convergence to the asymptotic cosmology, we can compute exactly the late-time $\epsilon$-parameter.}
    \label{fig.: 2by2 convex-hull picture 6}
\end{figure}

\subsection{de Sitter conjecture}

A core idea of the Swampland Program is the (refined) de Sitter conjecture \cite{Obied:2018sgi, Ooguri:2018wrx}, by which scalar potentials that are consistently coupled with quantum gravity are bounded as
\begin{equation} \label{de Sitter conjecture}
    \gamma \geq \gamma_{\mathrm{dS}},
\end{equation}
where $\gamma$ is the potential gradient norm defined in eq. (\ref{scalar-potential gradient norm}) and $\gamma_{\mathrm{dS}}$ is an unknown order-$1$ constant. Several arguments exist in the literature that propose the value of this constant \cite{Bedroya:2019snp, Rudelius:2021azq}. Here, we relate $\gamma_{\mathrm{dS}}$ to the constant $\smash{\hat{\gamma}_\infty}$, independently on whether the late-time solution is a scaling cosmology. As usual, we stress that, in general, there is no intrinsic obstruction to the slope of a multi-exponential potential. Rather, any Swampland bound should come from the constraints that quantum gravity would possibly impose on the constant $\smash{\hat{\gamma}_\infty}$.

\subsubsection{Preliminary observations}

For positive-definite scalar potentials, we can bound $\gamma$ exactly. In the basis that optimizes the late-time acceleration bound, the potential takes the form in eq. (\ref{optimal-bound basis potential}). Therefore, we can write the chain of inequalities
\begin{align*}
    \dfrac{1}{\kappa_d^2 V^2} \dfrac{\der V}{\der \phi_a} \dfrac{\der V}{\der \phi^a} & \geq \dfrac{1}{\kappa_d^2 V^2} \dfrac{\der V}{\der \hat{\varphi}} \dfrac{\der V}{\der \hat{\varphi}} \geq (\hat{\gamma}_\infty)^2.
\end{align*}
This implies the off-shell inequality
\begin{equation}
    \gamma \geq \gamma_{\mathrm{dS}} = \hat{\gamma}_\infty.
\end{equation}
In particular, if the $d$-dimensional dilaton is not stabilized, this implies $\gamma \geq 2 / \sqrt{d-2}$. Below, we outline a proposal to characterize the late-time de Sitter coefficient on shell, while also relaxing the assumptions on the signs of the terms in the potential.

By definition, we know that the potential gradient norm and the scalar-potential directional derivative satisfy the inequality $\smash{\gamma \geq \ab \gamma_\star \ab}$, as in eq. (\ref{scalar-potential gradient inequality}). If $\epsilon=0$, then by eqs. (\ref{directional derivative and acceleration parameter}, \ref{scalar-potential gradient norm and acceleration parameter}) we immediately get $\gamma = \gamma_\star = 0$. This is obvious because a vanishing $\epsilon$-parameter corresponds to a de Sitter stationary point. Here we are interested in characterizing scenarios with $\epsilon > 0$. In general, by the universal late-time bound $\epsilon \geq [(d-2)/4] \, (\hat{\gamma}_\infty)^2$ of eq. (\ref{optimal epsilon bound}), we know that
\begin{equation} \label{general de Sitter-coefficient bound}
    \gamma \geq \ab \gamma_\star \ab \geq \hat{\gamma}_\infty \biggl\ab 1 - \dfrac{1}{2} \, \dfrac{\eta}{(d-1) - \epsilon} \biggr\ab.
\end{equation}
This is a non-trivial inequality and it is the starting point for our discussion of the constant $\smash{\gamma_{\mathrm{dS}}}$. 
If $\smash{\hat{\gamma}_\infty = 0}$, we get the trivial inequality $\gamma \geq 0$, so here we focus on cases with $\smash{\hat{\gamma}_\infty>0}$.

To facilitate the discussion below, we recall that, by the late-time bounds reviewed in subsec. \ref{subsec: bounds on late-time cosmic-acceleration}:
\begin{itemize}
    \item the Hubble parameter is bounded from below as
    \begin{equation} \label{Hubble-parameter lower bound}
        H \overset{t \sim \infty}{\gtrsim} \dfrac{1}{d-1} \dfrac{1}{t};
    \end{equation}
    \item if $\smash{\Gamma_d^2 > (\hat{\gamma}_\infty)^2>0}$, the Hubble parameter is bounded from above as
    \begin{equation} \label{Hubble-parameter upper bound}
        H \overset{t \sim \infty}{\lesssim} \dfrac{4}{d-2} \dfrac{1}{(\hat{\gamma}_\infty)^2} \,  \dfrac{1}{t}.
    \end{equation}
\end{itemize}
Although the Hubble parameter is bounded, we do not know bounds on its first and second derivatives, therefore eqs. (\ref{Hubble-parameter lower bound}, \ref{Hubble-parameter upper bound}) do not immediately inform us further on the bound in eq. (\ref{general de Sitter-coefficient bound}). Nevertheless, these are two powerful conditions that we exploit below to verify when an inequality of the form of the de Sitter conjecture in eq. (\ref{de Sitter conjecture}) is attainable at asymptotically-late times. The case where $\smash{\hat{\gamma}_\infty \geq \Gamma_d}$ is not discussed here since it represents a singular case where the late-time attractor is formally one with a vanishing potential.

\subsubsection{Positiveness of the de Sitter coefficient}

In view of eq. (\ref{general de Sitter-coefficient bound}), if $\epsilon > 0$, we can have a vanishing potential gradient norm $\gamma = 0$ only if $(d-1) - \epsilon = \eta/2$. Note that a vanishing gradient $\smash{\der V / \der \phi^a = 0}$ does not imply a de Sitter vacuum: if $\dot{\phi}^a \neq 0$, this corresponds to pure kination. In terms of the Hubble parameter, this equation takes the simple form
\begin{align*}
    d-1 + \dfrac{1}{2} \dfrac{\ddot{H}}{H \dot{H}} = 0.
\end{align*}
What we will do is check whether solutions $H_0 = H_0(t)$ to this equation are compatible with the cosmological solutions to eqs. (\ref{FRW-KG eq.}, \ref{Friedmann eq. 1}, \ref{Friedmann eq. 2}). As $\epsilon>0$, we know that $H, \dot{H} > 0$. Then we can write the equation of interest as
\begin{align*}
    \dfrac{\de}{\de t} \bigl[ (d-1) H^2 + \dot{H} \bigr] = 0,
\end{align*}
whose integration gives a family of solutions such that
\begin{align*}
    (d-1) H_0^2 + \dot{H_0} = k_1,
\end{align*}
for constants $k_1 \in \mathbb{R}$. This implies $\smash{\epsilon_0 = (d-1) - k_1 / H_0^2}$. Because any acceptable Hubble parameter must be such that $\smash{\lim_{t \to \infty} H(t) = 0}$ by eq. (\ref{Hubble-parameter upper bound}), as the $\epsilon$-parameter is also bounded as $\smash{[(d-2)/4] (\hat{\gamma}_\infty)^2 \leq \epsilon \leq d-1}$, we must further have $\smash{\lim_{t \to \infty} \dot{H}(t) = 0}$. So, we should fix the initial condition as $k_1=0$. A further integration thus gives
\begin{align*}
    H_0 = \dfrac{1}{(d-1) \, t + k_2},
\end{align*}
which corresponds to pure kination, as expected. Therefore, the cosmological equations cannot have $H = H_0$ as a solution unless $\epsilon = d-1$, which also gives $\eta = 0$. This means that, unless $\epsilon = d-1$, non-zero values $\epsilon>0$ and $\smash{\hat{\gamma}_\infty > 0}$ imply the strict inequality $\smash{\gamma > 0}$, i.e.
\begin{equation} \label{positive de Sitter coefficient}
    \gamma_{\mathrm{dS}} > 0.
\end{equation}
This is not a trivial result. It is true that if $\smash{\hat{\gamma}_\infty>0}$, the asymptotic $\epsilon$-parameter is strictly positive and therefore there cannot be an asymptotic de Sitter stationary point. However, one could in principle have $\smash{\gamma=0}$ due to the vanishing of the factor multiplying the overall term proportional to $\sqrt{\epsilon}$, which would correspond to kination. Nonetheless, here we are seeing that this is not generally the case. Also, we consider this as a propaedeutic step towards gaining analytical control over the coefficient $\smash{\gamma_{\mathrm{dS}}}$ in the way we sketch below.

\subsubsection{Heuristic limiting behavior of the de Sitter coefficient}
In view of the inequality of eq. (\ref{general de Sitter-coefficient bound}) and of the universal bound in eq. (\ref{optimal epsilon bound}), we can attempt to relate the de Sitter coefficient $\gamma_{\mathrm{dS}}$ to the coefficient $\smash{\hat{\gamma}_\infty}$. If the de Sitter conjecture in eq. (\ref{de Sitter conjecture}) is true, then we can try and define a positive constant $\gamma_{\mathrm{dS}} > 0$ as a solution to the equation
\begin{equation} \label{de Sitter coefficient limiting behavior}
    \hat{\gamma}_\infty \biggl[ 1 - \dfrac{1}{2} \, \dfrac{\eta}{(d-1) - \epsilon} \biggr] = \gamma_{\mathrm{dS}},
\end{equation}
assuming that $\eta \leq 2 [(d-1) - \epsilon]$. Here, $\smash{\gamma_{\mathrm{dS}}}$ is a constant that we see as a way to study the limiting behavior of the possible solutions to the field equations. Intuitively, we expect $\epsilon$ to converge to a value $\epsilon_\infty$ at large times, and thus, for the limiting behavior, we use eq. (\ref{directional derivative and acceleration parameter}) to speculate that the function $\smash{g(t) = 1 - (1/2) \, \eta(t)/[(d-1) - \epsilon(t)]}$ is approximately constant. This gives us precisely eq. (\ref{de Sitter coefficient limiting behavior}). Therefore, we can try to see whether the time dependence that $\epsilon$ and $\eta$ should have can be compatible with $\gamma$ being bounded from below by a constant, and what such constant should be.

It is convenient to rewrite eq. (\ref{de Sitter coefficient limiting behavior}) explicitly in terms of the Hubble parameter, getting
\begin{equation}
    \sigma (d-1) - (1-\sigma) \dfrac{\dot{H}}{H^2} = - \dfrac{1}{2} \dfrac{\ddot{H}}{H \dot{H}},
\end{equation}
where we defined $\smash{\sigma = 1 - \gamma_{\mathrm{dS}} / \hat{\gamma}_\infty}$ for brevity. We can also write the equation as
\begin{align*}
    \dfrac{1}{2} \, \dfrac{\de}{\de t} \Bigl[ H \dot{H} + \dfrac{2 \sigma}{3} (d-1) H^3 \Bigr] = \Bigl( \dfrac{3}{2} - \sigma \Bigr) \dot{H}^2.
\end{align*}
We denote solutions to this equation as $\smash{H = H_{1 - \sigma}}$ and we look for the conditions under which they are compatible with the known late-time behavior of the actual cosmological solutions to eqs. (\ref{FRW-KG eq.}, \ref{Friedmann eq. 1}, \ref{Friedmann eq. 2}).

The cases $\sigma=0,1,3/2$ are special, since they simplify the structure of the equation. As $\sigma=1$ gives the solution $H_0$ discussed before, we are motivated to test the ansatz $\smash{H_{1-\sigma} = 1/(a_{1-\sigma} t + b_{1-\sigma})}$. If $\sigma \neq 0,1,3/2$, such a function is a solution if $\smash{a_{1-\sigma} = d-1}$, which means
\begin{align*}
    H_{1-\sigma} = \dfrac{1}{(d-1) \, t + b_{1-\sigma}},
\end{align*}
for some real constant $\smash{b_{1-\sigma} \in \mathbb{R}}$. Again, we find that this solution is incompatible with the known behavior of the actual cosmological solutions. As the case $\sigma=1$ has already been studied, and the case $\sigma=3/2$ is irrelevant as physically we are interested in $\sigma \leq 1$, we only need to check the solutions for $\sigma = 0$. In this case, the solution to the differential equation is
\begin{align*}
    H_1 = \dfrac{1}{a_1 t + b_1},
\end{align*}
for two real constants $a_1, b_1 \in \mathbb{R}$. This is compatible with both conditions in eqs. (\ref{Hubble-parameter lower bound}, \ref{Hubble-parameter upper bound}), if $\smash{[(d-2)/4] (\hat{\gamma}_\infty)^2 \leq a_1 \leq d-1}$. In view of these results, mathematical and physical intuition lead us to the suggestion that
\begin{equation} \label{conjectured bound for de Sitter coefficient}
    \gamma_{\mathrm{dS}} = \hat{\gamma}_\infty.
\end{equation}
Indeed, the value $\sigma=0$ seems to be very special because it gives exactly an admissible $(1/t)$-behavior with the known late-time constraints. Moreover, scaling solutions, which we expect to be very general late-time attractors after suitable truncations (see subsecs. \ref{subsec: scaling cosmologies as late-time attractors}, \ref{subsec: further scaling solitions}, \ref{subsec: partial moduli stabilization} and \ref{subsec: coupling convex hull and cosmic acceleration}), are exactly such as to saturate this bound. As a final argument, positive-definite potentials give exactly the same condition off shell. How this off-shell inequality changes on-shell is however not obvious (also, that off-shell condition assumes positive-definite potentials, while here we do not have this assumption).

Unlike other results in this paper, the inequality in eq. (\ref{conjectured bound for de Sitter coefficient}) is only conjectural and not analytic, since it comes as a result of a speculation, we do not know about the uniqueness of our solution. We plan to come back to this in the future.

\section{Examples} \label{sec: examples}
In this section we discuss a few examples of what our late-time cosmological characterization implies for simple multi-field multi-exponential potentials.

\subsection{String-theoretic toy models}

To start, we discuss a few simple 2-field string-theoretic toy models, to show a number of intriguing properties of scaling cosmologies.

In a $d$-dimensional type-II compactification, let the internal space host a $q$- and a $p$-cycle supporting RR-$q$- and RR-$p$-form fluxes, respectively. If we assume that the only two dynamical scalar fields are the dilaton and a radion, this means that the total scalar potential is
\begin{equation}
    V = \Lambda_1 \, \e^{\kappa_d [\frac{d}{\sqrt{d-2}} \tdelta + \frac{10-d - 2q}{\sqrt{10-d}} \tsigma]} + \Lambda_2 \, \e^{\kappa_d [\frac{d}{\sqrt{d-2}} \tdelta + \frac{10-d - 2p}{\sqrt{10-d}} \tsigma]},
\end{equation}
in terms of the canonically-normalized $d$-dimensional dilaton $\smash{\tdelta}$ and string-frame radion $\smash{\tsigma}$. As $\mathrm{rank} \, \gamma_{ia} = m = n = 2$ if $p \neq q$ and the potential is positive-definite, we can look for scaling solutions. It turns out that, if $q < (10-d)/2$ and $p>(10-d)/2$, the scaling solution can be an attractor since $\lambda^1, \lambda^2 > 0$. In all such cases, the late-time cosmologies feature the late-time field-space trajectories
\begin{align*}
    \tdelta_*(t) & = \tdelta_0 - \dfrac{1}{\kappa_d } \, \dfrac{2 \sqrt{d-2}}{d} \; \mathrm{ln} \, \dfrac{t}{t_0}, \\
    \tsigma_*(t) & = \tsigma_0,
\end{align*}
with the $\epsilon$-parameter
\begin{align*}
    \epsilon = \dfrac{d^2}{4}.
\end{align*}
Clearly, there are two main observations to make: first, there cannot be cosmic acceleration, since $\epsilon > 1$ for $d>2$, as we knew already since the dilaton is not stabilized, in accordance with subsec. \ref{subsec: dilaton obstruction to acceleration}; second, the string-frame volume modulus is stabilized, in accordance with the discussion in subsec. \ref{subsec: partial moduli stabilization}. Moreover, we observe that the $\epsilon$-parameter saturates the bound of eq. (\ref{optimal epsilon bound}): this is an example of a general feature of scaling solutions that we discuss further in subsec. \ref{subsec: coupling convex hull and cosmic acceleration}. In the $\smash{(\sdil, \omega)}$-basis (notice that these fields are not canonically-normalized yet, but this is not going to be relevant for our conclusions), the scalar potential reads
\begin{align*}
    V = \Lambda_1 \, \e^{\frac{q-5}{2} \sdil - [2 q + \frac{2 (10-d)}{d-2}] \omega} + \Lambda_2 \, \e^{\frac{p-5}{2} \sdil - [2 p + \frac{2 (10-d)}{d-2}] \omega},
\end{align*}
and we can easily express the time-evolution law of the scalars by rotating (see appendix \ref{app: string dimensional reductions} for details) the solution in the $\smash{(\tdelta, \tsigma)}$-basis, obtaining
\begin{align*}
    \sdil_*(t) & = \sdil_0 - \dfrac{d-2}{d} \; \mathrm{ln} \, \dfrac{t}{t_0}, \\
    \omega_*(t) & = \omega_0 + \dfrac{d-2}{4d} \; \mathrm{ln} \, \dfrac{t}{t_0}.
\end{align*}
As time passes, the string coupling decreases and the Einstein-frame volume increases, whilst the string-frame volume is stable. Therefore, as well as being mathematically rigorous, this toy-model solution is also physically acceptable. As a final consistency test, we check the issue of scale separation. On the one hand, the Hubble length evolves trivially as $l_H = 1/H = \epsilon t$; on the other hand, the Kaluza-Klein length is in eq. (\ref{KK-scale evolution}). As we know dilaton and the radion time dependence, we can express the time-evolution law for the Hubble and the Kaluza-Klein scale as
\begin{align*}
    l_H & = l_H (t_0)  \, \dfrac{t}{t_0}, \\
    l_{\mathrm{KK}, d} & = l_{\mathrm{KK}, d} (t_0) \, \Bigl(\dfrac{t}{t_0}\Bigr)^{\frac{2}{d}},
\end{align*}
which means that the compactification ansatz is self-consistent, since the Kaluza-Klein length grows less quickly than the Hubble length. One should not worry about the volume growing indefinitely, since, by how we defined the Kaluza-Klein scale, this means that the hierarchy between the Planck mass and the Kaluza-Klein mass grows over time, too. Asymptotically far in time, this can conceptually be compatible with the observed hierarchies appearing in nature. The dimensionally-reduced string length, corresponding to the inverse string-excitation tower mass scale (see eq. (\ref{d-dim. string scale})) grows like the Kaluza-Klein length, i.e. $\smash{l_{s, d} (t) = l_{s, d} (t_0) \, (t/t_0)^{\frac{2}{d}}}$.

As another instructive string-theoretic toy model, we consider the non-supersymmetric heterotic $\mathrm{SO}(16) \!\times\! \mathrm{SO}(16)$-theory compactified on a $(10-d)$-dimensional space with non-trivial curvature. If the internal curvature is positive, the scalar potential reads \cite{Ginsparg:1986wr, Baykara:2022cwj}
\begin{equation} \label{het-SO(16)xSO(16) potential: R+C}
    V = \Lambda_R \, \e^{\kappa_d [\frac{2}{\sqrt{d-2}} \tdelta - \frac{2}{\sqrt{10-d}} \tsigma]} + \Lambda_C \, \e^{\kappa_d [\frac{d}{\sqrt{d-2}} \, \tdelta + \sqrt{10-d} \, \tsigma]}.
\end{equation}
Here, the first term, with $\Lambda_R < 0$, is generated by the non-trivial internal curvature $\smash{\breve{R}_{(10-d)} > 0}$, while the second term, with $\Lambda_C > 0$, is the well-known Casimir energy generated by the absence of a supersymmetric matching of bosonic and fermionic degrees of freedom in the tower of string states \cite{AlvarezGaume:1986jb, Dixon:1990pc}. In this case, we cannot prove the convergence to a scaling solution by the discussion in subsec. \ref{subsec: scaling cosmologies as late-time attractors}; however, we can bound the late-time $\epsilon$-parameter by eq. (\ref{epsilon bound}) (or, of course, eq. (\ref{optimal epsilon bound})). In the $\smash{(-\tdelta, -\tsigma)}$-basis, we find that the coefficients
\begin{align*}
    \gamma_{\infty}^{-\tdelta} & = \dfrac{d}{\sqrt{d-2}}, \\
    \gamma_\infty^{-\tsigma} & = \sqrt{10-d},
\end{align*}
give $\smash{(\gamma_\infty)^2 = 4/(d-2) + 12 > \Gamma_d^2}$, which means that the asymptotic late-time solution to the field equations, if the total potential remains positive, is bound to give
\begin{align*}
    \epsilon = d-1.
\end{align*}
In such a case, we also stress that there is no evidence that the field-space evolution is aligned to the gradient-flow trajectory since the scaling solution is the degenerate one, therefore the discussion of subsubsec. \ref{subsubsec: accidental gradient-flow trajectories} does not apply. In case the internal curvature is negative, since $\Lambda_R, \Lambda_C > 0$ and since it turns out that $\lambda^1, \lambda^2 > 0$, the scaling solution can be an attractor, which corresponds to the trajectory
\begin{align*}
    \tdelta_*(t) & = \tdelta_0 - \dfrac{1}{\kappa_d} \dfrac{12-d}{10} \sqrt{d-2} \; \mathrm{ln} \, \dfrac{t}{t_0}, \\
    \tsigma_*(t) & = \tsigma_0 + \dfrac{1}{\kappa_d} \dfrac{d-2}{10} \sqrt{10-d} \; \mathrm{ln} \, \dfrac{t}{t_0},
\end{align*}
with
\begin{align*}
    \epsilon = \dfrac{1}{1 - \dfrac{3 (d-2)}{25}} > 1.
\end{align*}
As we already knew, there is no accelerated expansion since the $d$-dimensional dilaton is there as a rolling scalar. In the $\smash{(\tsdil, \tomega)}$-basis, we find
\begin{align*}
    \tsdil_*(t) & = \tsdil_0 - \dfrac{1}{\kappa_d} \dfrac{d-2}{10 \sqrt{2}} \; \mathrm{ln} \, \dfrac{t}{t_0}, \\
    \tomega_*(t) & = \tomega_0 + \dfrac{1}{\kappa_d} \dfrac{\sqrt{2}}{4} \, \sqrt{10-d} \sqrt{d-2} \; \mathrm{ln} \, \dfrac{t}{t_0},
\end{align*}
which is physically consistent, too. It can also be instructive to reconsider the problem in the rotated basis in which one field is aligned in the direction followed by the field-space trajectories, along the lines of subsec. \ref{subsubsec: time-measuring and quintessence-like scalars}. It is convenient to rotate the fields $\smash{\tdelta^- = - (\tdelta - \tdelta_0)}$ and $\smash{\tsigma^+ = \tsigma - \tsigma_0}$: the angle spanned in the $\smash{(\tdelta^-, \tsigma^+)}$-plane is $\alpha(d) = \mathrm{arctan} \, \bigl[ \sqrt{d-2} \sqrt{10-d}/(12-d) \bigr]$, and one can rotate the field basis by defining
\begin{align*}
    \tvec{\tdelta^-}{\tsigma^+} = \left(\!\! \begin{array}{cc}
        \mathrm{cos} \, \alpha(d) & - \mathrm{sin} \, \alpha(d) \\
        \mathrm{sin} \, \alpha(d) & \mathrm{cos} \, \alpha(d)
    \end{array} \!\!\right) \tvec{\tilde{\varphi}}{\tilde{\xi}},
\end{align*}
obtaining a scalar potential written as
\begin{align*}
    V = \Lambda_R \, \e^{- \kappa_d \tilde{\gamma}_* \tilde{\varphi} - \frac{2 \kappa_d \tilde{\xi}}{\sqrt{10-d} \sqrt{1 + 3(10-d)}}} \\
    + \Lambda_C \, \e^{- \kappa_d \tilde{\gamma}_* \tilde{\varphi} + \frac{6 \sqrt{10-d} \kappa_d \tilde{\xi}}{\sqrt{1+3(10-d)}}} & ,
\end{align*}
where $\gamma_*$ is the directional derivative defined in eq. (\ref{directional derivative}), which in this case reads
\begin{align*}
    \gamma_* = \dfrac{10}{\sqrt{1 + 3(10-d)} \sqrt{d-2}}.
\end{align*}
Finally, we can find the scalar-field trajectories, i.e.
\begin{align*}
    \tilde{\varphi}(t) & = \tilde{\varphi}_0 - \dfrac{1}{\kappa_d} \, \dfrac{2}{\tilde{\gamma}_*} \; \mathrm{ln} \, \dfrac{t}{t_0}, \\
    \tilde{\xi}(t) & = \tilde{\xi}_0,
\end{align*}
either by rotating the solutions in the original basis or by solving the field equations in the new basis. In the rotated basis, the bound of eq. (\ref{epsilon bound}) not only takes the form of eq. (\ref{optimal epsilon bound}), but is actually saturated, as is in fact expected from the discussion of subsec. \ref{subsec: coupling convex hull and cosmic acceleration}. This model is interesting because it shows that we can characterize any exponential potential, irrespective of its macroscopic origin, including quantum-generated effects.

\subsection{Attempts for accelerated expansion}

Based on F-theory constructions, ref. \cite{Calderon-Infante:2022nxb}\footnote{One can compare with ref. \cite{Calderon-Infante:2022nxb} through the following dictionary. By denoting non-canonically-normalized fields as $s^a$, with $a=1,\dots,n$, and by ordering the elements of the coefficient set as $\smash{\mathcal{E} = (\underline{l}_i)_{i=1}^m}$, with the coefficient components being $(l_i)_a$, their kinetic action and potential
\begin{align*}
    S = \int_{\mathrm{X}_{1,3}} \tilde{*}_{1,3} \Biggl[ \sum_{a=1}^n \dfrac{1}{4 \kappa_4^2} \dfrac{d_a}{(s^a)^2} \, \der s^a \der s^a - \sum_{i=1}^m A_i \prod_{a=1}^n (s^a)^{(l_i)_a} \Biggr]
\end{align*}
is mapped into the potential in eq. (\ref{generic exponential potential}) for canonical scalars, along with the obvious identity $\smash{A_i = \Lambda_i}$, via the identifications
\begin{align*}
    s^a & = \e^{- \sqrt{\frac{2}{d_a}} \; \kappa_4 \phi^a}, \\
    (l_i)_a & = \sqrt{\frac{d_a}{2}} \; \gamma_{ia}.
\end{align*}
For the coupling convex hull, we have $\smash{(\mu_{i})_a = \gamma_{ia}}$. A power-law evolution $\smash{s^a(\lambda) = \sigma^a \lambda^{\beta^a}}$ translates into a logarithmic evolution $\smash{\kappa_4 \phi^a (\lambda) = - \sqrt{d_a/2} \; \mathrm{ln} \, \sigma^a - \sqrt{d_a/2} \; \beta^a \, \mathrm{ln} \, \lambda}$, with $\lambda = \lambda(t)$ a reparameterized time. In a scaling cosmology with trajectories $\smash{\phi^a (t) = \phi_0^a + (\alpha^a/\kappa_4) \, \mathrm{ln} \, t/t_0}$, we have $\smash{t/t_0 = \lambda^{- \sqrt{d_a/2} \, \beta^a / \alpha^a}}$.}
advocates for realizations of late-time string-theoretic accelerated expansion. Under the assumption that the Einstein-frame volume can be stabilized and that this stabilization does not involve the presence of new scalar-potential terms, compactifications are hypothesized to exist in which two scalar fields $\smash{\tilde{\phi}^1}$ and $\smash{\tilde{\phi}^2}$ have a 4-dimensional scalar potential\footnote{This is the potential in ref. \cite[eq. (3.24)]{Calderon-Infante:2022nxb}: the fields $s$ and $u$ are canonically-normalized as $\smash{s = \e^{-\kappa_4 \sqrt{2} \tilde{\phi}^1}}$ and $\smash{u = \e^{-\kappa_4 \sqrt{2/3} \tilde{\phi}^2}}$.}
\begin{equation}
    V = \Lambda_1 \, \e^{\kappa_4 \, \sqrt{2} \, \tilde{\phi}^1 - \kappa_4 \, \sqrt{\frac{2}{3}} \, \tilde{\phi}^2} + \Lambda_2 \, \e^{- \kappa_4 \, \sqrt{2} \, \tilde{\phi}^1 + \kappa_4 \, \sqrt{6} \, \tilde{\phi}^2}.
\end{equation}
To start, we notice that in this case the bound in eq. (\ref{epsilon bound}) is trivial (i.e. $\epsilon \geq 0$), since $\smash{\gamma_\infty^{\tilde{\phi}^1} = \gamma_\infty^{\tilde{\phi}^2}=0}$. However, we can see that the equation of convex-hull hyperplane is $\gamma_2=-(2\sqrt{3}/3) \, \gamma_1-\sqrt{6}/3$: the orthogonal line is $\gamma_2=(\sqrt{3}/2) \, \gamma_1$, with an intersection at $(\gamma_1, \gamma_2) = (-2\sqrt{2}/7,-\sqrt{6}/7)$, which gives $\epsilon \geq 1/7$ in view of the optimal bound in eq. (\ref{optimal epsilon bound}). Then, we know that $\smash{\Lambda_1, \Lambda_2 > 0}$ and we can check that $\smash{(\lambda^1, \lambda^2) = (9/4,9/4)}$: therefore, we know that asymptotically the field-equation solution is the proper scaling solution. In particular, we can compute the $\epsilon$-parameter and it reads
\begin{align*}
    \epsilon = \dfrac{1}{7},
\end{align*}
corresponding to the normalized scalar-potential directional derivative $\smash{\gamma_* = \sqrt{2/7}}$. The field-space trajectories are found to be
\begin{align*}
    \tilde{\phi}^1_*(t) & = \tilde{\phi}^1_0 - \dfrac{1}{\kappa_4} \, 2 \sqrt{2} \; \mathrm{ln} \, \dfrac{t}{t_0}, \\
    \tilde{\phi}^2_*(t) & = \tilde{\phi}^2_0 - \dfrac{1}{\kappa_4} \, \sqrt{6} \; \mathrm{ln} \, \dfrac{t}{t_0}.
\end{align*}
As we know, this is not a slow-roll solution, despite having an accelerated expansion with $\epsilon = 1/7 < 1$. Through the normalized scalar-potential gradient, one would get
\begin{align*}
    \gamma(\tilde{\phi}^1, \tilde{\phi}^2) = \sqrt{\dfrac{8}{3} \dfrac{V_1^2}{V^2} + 8 \dfrac{V_2^2}{V^2} - 8 \dfrac{V_1 V_2}{V^2}},
\end{align*}
which is not a constant but rather a field-dependent quantity, in principle. However, in a scaling cosmology all potentials evolve as $V_{1,2}(t) = V_{1,2}(t_0) (t_0/t)^2$, giving
\begin{align*}
    \gamma[\tilde{\phi}^1_*(t), \tilde{\phi}^2_*(t)] = \sqrt{\dfrac{8}{3} \dfrac{V_1^2(t_0)}{V^2(t_0)} + 8 \dfrac{V_2^2(t_0)}{V^2(t_0)} - 8 \dfrac{V_1(t_0) V_2(t_0)}{V^2(t_0)}},
\end{align*}
which, in principle, depends on the initial conditions. It is possible to minimize the coefficient $\smash{\gamma(\tilde{\phi}^1, \tilde{\phi}^2)}$, which is given by the trajectory
\begin{align*}
    \zeta: \quad \tilde{\phi}^1_{\zeta}(\tilde{\phi}^2) = \dfrac{2 \sqrt{3}}{3} \, \tilde{\phi}^2 + \dfrac{\sqrt{2}}{4} \, \mathrm{ln} \, \dfrac{9 \Lambda_2}{5 \Lambda_1}.
\end{align*}
In fact, for initial conditions that are such that $\smash{\tilde{\phi}^1_0 = (2 \sqrt{3}/3) \, \tilde{\phi}^2_0 + (\sqrt{2}/4) \, \mathrm{ln} \, (9 \Lambda_2 / 5 \Lambda_1)}$, this is the same as the scaling-solution trajectory $\smash{(\tilde{\phi}^1,\tilde{\phi}^2) = (\tilde{\phi}^1_*,\tilde{\phi}^2_*)}$, and one finds a proportionality $\smash{5 V_1(t_0) = 9 V_2(t_0)}$ which gives $\smash{\gamma(\tilde{\phi}^1_\zeta(\tilde{\phi}^2), \tilde{\phi}^2) = \sqrt{2/7}}$. The compatibility condition in eq. (\ref{compatibility condition}) ensures that the appropriate choice of initial time for a given set of initial conditions is such that this is the proportionality factor for a scaling cosmology. Similarly to what has been discussed for the accidental gradient-flow trajectory, eq. (\ref{scaling-cosmology potential gradient norm}) can be written both in terms of eq. (\ref{scaling-cosmology potential gradient}) and through a direct computation, which involves the compatibility condition in eq. (\ref{compatibility condition}). This is a general result for scaling cosmologies.\footnote{Taking this example as a useful instance for a comparison, we highlight a few crucial facts, in relationship with ref. \cite{Calderon-Infante:2022nxb}.

To start, there is no slow roll (see subsec. \ref{subsec: bounds on late-time cosmic-acceleration}, subsubsecs. \ref{subsubsec: scalar-potential directional derivative} and \ref{subsubsec: accidental gradient-flow trajectories}, and also subsec. \ref{subsec: scalar-potential derivatives and acceleration}). Moreover, the trajectory of a scaling cosmology follows the gradient flow by accident (see subsubsec. \ref{subsubsec: accidental gradient-flow trajectories}), but it has a different time dependence. Finally, unless one is certain to be dealing with a cosmology where $\eta=0$ and $\Omega=0$, one cannot compute $\epsilon$ through a directional derivative or the norm of the gradient of the potential (and the potential gradient norm and directional derivative do not coincide; see the discussion of eqs. (\ref{directional derivative and acceleration parameter}, \ref{scalar-potential gradient norm and non-geodesity factor})).

To study exponential cosmologies, we also remark that the fact that scaling solutions are late-time attractors is perturbatively and numerically argued already in refs. \cite{Halliwell:1986ja, Copeland:1997et, Malik:1998gy, Coley:1999mj, Guo:2003eu, Guo:2003rs, Bergshoeff:2003vb, Collinucci:2004iw, Kim:2005ne, Hartong:2006rt}, and proven independently of initial conditions in the current paper.
}

Although so far we have seen examples where conditions \ref{condition (a)}-\ref{condition (b)} are always met, this is of course not always the case. For instance, we can consider the 4-dimensional scalar potential
\begin{equation}
    V = \Lambda_1 \, \e^{\kappa_4 \, \sqrt{2} \, \tilde{\phi}^1 + \kappa_4 \, \sqrt{\frac{2}{3}} \, \tilde{\phi}^2} + \Lambda_2 \, \e^{\kappa_4 \, \sqrt{2} \, \tilde{\phi}^1 + \kappa_4 \, \sqrt{6} \, \tilde{\phi}^2},
\end{equation}
with $\Lambda_1, \Lambda_2 > 0$. In this case, the bound in eq. (\ref{epsilon bound}) gives $\smash{\epsilon \geq 4/3}$. This is as much as we can say with certainty, since a closer look reveals $\smash{(\lambda^1,\lambda^2) = (3/4,-1/4)}$, which means that there is no proof for the proper scaling solution to be the late-time attractor. As explained above, we do not even have a proof that the late-time potential would be dominated by the one term with positive $\lambda$. However, let us assume this to be the case -- since the term with negative $\lambda$ appears to be asymptotically more suppressed -- and let us consider the potential\footnote{This potential is in ref. \cite[eq. (3.18)]{Calderon-Infante:2022nxb}, with $\smash{s = \e^{-\kappa_4 \sqrt{2} \tilde{\phi}^1}}$ and $\smash{u = \e^{-\kappa_4 \sqrt{2/3} \tilde{\phi}^2}}$.}
\begin{equation}
    V = \Lambda_1 \, \e^{\kappa_4 \, \sqrt{2} \, \tilde{\phi}^1 + \kappa_4 \, \sqrt{\frac{2}{3}} \, \tilde{\phi}^2} = \Lambda_1 \, \e^{\kappa_4 \, \frac{2 \sqrt{6}}{3} \, \tilde{\varphi}},
\end{equation}
where a field rotation\footnote{In particular, this is
\begin{align*}
    \tvec{\tilde{\varphi}}{\tilde{\xi}} = \matr{\dfrac{\sqrt{3}}{2}}{\dfrac{1}{2}}{-\dfrac{1}{2}}{\dfrac{\sqrt{3}}{2}} \tvec{\tilde{\phi}^1}{\tilde{\phi}^2},
\end{align*}} has mapped the 2-field problem into a 1-field problem. For this potential, it is trivial to check that the $\epsilon$-parameter is $\smash{\epsilon = 4/3}$. This potential has a unique late-time behavior, the scaling solution with $\smash{\tilde{\varphi} = \tilde{\varphi}_0 + (\sqrt{6}/2) \, \mathrm{ln} \, t/t_0}$.

If the potential is steep enough, we cannot have accelerated expansion. One example is the single-exponential 4-dimensional potential
\begin{equation}
    V = \Lambda \, \e^{\kappa_d \, 2 \sqrt{2} \, \varphi}.
\end{equation}
In this case, because $\smash{\gamma_\infty^2 = 8 > 6}$, the potential is sufficiently steep that the $\epsilon$-parameter approaches the asymptotic value $\epsilon = 3$ \cite{Rudelius:2022gbz, Shiu:2023nph}.\footnote{This potential is studied in ref. \cite{Calderon-Infante:2022nxb}, in the notation $V = \Lambda/{s u^3}$, upon the identifications $\smash{s = \e^{- \kappa_4 \sqrt{2} \tilde{\phi}^1}}$ and $\smash{u = \e^{- \kappa_4 \sqrt{2/3} \tilde{\phi}^2}}$, and after the field rotation
\begin{align*}
    \tvec{\tilde{\varphi}}{\tilde{\xi}} = \matr{\dfrac{1}{2}}{\dfrac{\sqrt{3}}{2}}{-\dfrac{\sqrt{3}}{2}}{\dfrac{1}{2}} \tvec{\tilde{\phi}^1}{\tilde{\phi}^2}.
\end{align*}} This is an instance with no concept of gradient flow, even in a 1-dimensional field space. Asymptotically, the field reaches the non-proper scaling solution $\smash{\varphi_*(t) = \varphi_0 - (2 \sqrt{3}/3) \; \mathrm{ln} \, t/t_0}$, corresponding to $V_*=0$, with a slope that is unrelated to the value $\smash{2 \sqrt{2} = (1 / \kappa_4 V) \, \der V / \der \varphi}$.

\subsection{Comments on gradient flow} \label{subsec: comments on gradient flow}

As an example that illustrates our earlier comments on gradient-flow trajectories, we can discuss the steepest-descent trajectory associated to the potential of the heterotic $\mathrm{SO}(16) \!\times\! \mathrm{SO}(16)$-theory compactified on a $(10-d)$-dimensional space with non-trivial negative curvature, in eq. (\ref{het-SO(16)xSO(16) potential: R+C}).

To find the steepest-descent trajectory of a surface $\smash{V = V (\tilde{\phi}^1(\lambda),\tilde{\phi}^2(\lambda))}$, one can define the parametric curve $\smash{\underline{\tau}(\lambda) = (\tilde{\phi}^1(\lambda),\tilde{\phi}^2(\lambda))}$ and require it to have a velocity that is aligned with the gradient of the surface equation. This can be done by imposing the proportionality relation
\begin{align*}
    \dfrac{\de}{\de \lambda} \, \underline{\tau}(\lambda) = - \xi(\lambda) \, \dfrac{\der V}{\der \underline{\tau}}(\underline{\tau}(\lambda)),
\end{align*}
where $\xi = \xi(\lambda)$ is an arbitrary function (one can choose this to be a constant to facilitate intuition, but this is not necessary). Because the problem is 2-dimensional, we can easily infer the differential condition
\begin{align*}
    \dfrac{\de \tilde{\phi}^2}{\de \tilde{\phi}^1} (\tilde{\phi}^1) = \dfrac{1}{\dfrac{\der V}{\der \tilde{\phi}^1}} \dfrac{\der V}{\der \tilde{\phi}^2} \, (\tilde{\phi}^1, \tilde{\phi}^2 (\tilde{\phi}^1)),
\end{align*}
which we may be able to integrate analytically. If one assumes the slow-roll conditions to hold, this is exactly the scenario that one is faced with: neglecting the field second-order derivatives, the simplified scalar-field equations, (see eq. (\ref{FRW-KG eq.})), take the form \cite{Calderon-Infante:2022nxb}
\begin{align*}
    \dot{\tilde{\phi}}^{1,2} = - \dfrac{1}{(d-1) H } \dfrac{\der V}{\der \tilde{\phi}_{1,2}},
\end{align*}
which give precisely the steepest-descent trajectory. As already explained, this also has a chance to give the scaling-cosmology trajectory in field space, but it is not equivalent to solving the same field equations coupled to gravity, as the time dependences will be different in general. If the problem is not 2-dimensional, one finds a system of coupled differential equations which is harder to solve in general, but conceptually all considerations above apply identically.

In the case at hand, it is convenient to work in the $\smash{(\tsdil, \tomega)}$-basis since one of the scalar-potential terms depends only one field and thus makes it easier to determine the gradient-flow trajectory. Indeed, the potential in eq. (\ref{het-SO(16)xSO(16) potential: R+C}) reads
\begin{align*}
    V = \Lambda_R \, \e^{- \frac{4 \sqrt{2} \, \kappa_d \tomega}{\sqrt{d-2} \sqrt{10-d}}} + \Lambda_C \, \e^{\frac{5}{\sqrt{2}} \, \kappa_d \tsdil - \frac{1}{\sqrt{2}} \frac{\sqrt{10-d}}{\sqrt{d-2}} \, \kappa_d \tomega},
\end{align*}
and it turns out that the trajectory of steepest descent is the solution to the differential equation
\begin{align*}
    \dfrac{\de \tomega}{\de \tsdil} = - \dfrac{1 + \dfrac{8}{10-d} \dfrac{\Lambda_R}{\Lambda_C} \, \e^{- \frac{5}{\sqrt{2}} \, \kappa_d \tsdil - \frac{1}{\sqrt{2}} \frac{\sqrt{d-2}}{\sqrt{10-d}} \, \kappa_d \tomega}}{5 \, \dfrac{\sqrt{d-2}}{\sqrt{10-d}}}.
\end{align*}
This equation can be integrated analytically and gives the family of functions
\begin{align*}
    \tomega_{k}^\pm(\tsdil) = \dfrac{\sqrt{2}}{\kappa_d} \dfrac{\sqrt{10-d}}{\sqrt{d-2}} \; \mathrm{ln} \, \biggl[ \dfrac{1}{3} \dfrac{\Lambda_R/\Lambda_C}{10-d} \, \e^{-\frac{12 \sqrt{2}}{5} \, \kappa_d \tsdil} \pm \e^k & \biggr] \\
    - \dfrac{1}{5} \dfrac{\sqrt{10-d}}{\sqrt{d-2}} \, \tsdil & ,
\end{align*}
which are parameterized by the arbitrary constant $k \in \mathbb{R}$. This exemplifies by itself one important aspect: the gradient-flow trajectories are not unique. This is not the case for scaling solutions. In particular, there is only one set of initial conditions in which the scaling-cosmology trajectory -- namely the correct solution to the field equations, including the right time dependence -- coincides exactly with the steepest-descent trajectory. In this example, this happens for the solutions with $k = -\infty$ and with
\begin{align*}
    \tomega_{-\infty}^\pm(\tsdil) = - 5 \dfrac{\sqrt{10-d}}{\sqrt{d-2}} \, \tsdil + \dfrac{\sqrt{2}}{\kappa_d} \dfrac{\sqrt{10-d}}{\sqrt{d-2}} \; \mathrm{ln} \, \biggl[ \dfrac{1}{3} \dfrac{\Lambda_R/\Lambda_C}{10-d} \biggr].
\end{align*}
This is also the leading-order term of all the solutions at large negative $\smash{\tsdil}$, namely $\smash{\tomega_{k}^\pm(\tsdil \sim \infty) \simeq \tomega_{-\infty}^\pm(\tsdil \sim - \infty)}$. Indeed, by the scaling solution computed above, one can write
\begin{align*}
    \tomega_*(\tsdil_*) = - 5 \dfrac{\sqrt{10-d}}{\sqrt{d-2}} \, \tsdil_* + \biggl[ \tomega_0 + 5 \dfrac{\sqrt{10-d}}{\sqrt{d-2}} \, \tsdil_0 \biggr].
\end{align*}
This trajectory, which is the solution to the field equations for arbitrary initial conditions $\smash{\tsdil_0}$ and $\smash{\tomega_0}$, evidently agrees with the steepest-descent trajectory only for one choice of the initial conditions. Such initial conditions are the ones that are consistent with the choice of the initial time -- as we have already seen in the minimization of the gradient --, due to the compatibility condition in eq. (\ref{compatibility condition}). As a final consideration, we observe that, in view of the discussion of eqs. (\ref{scaling-solution gradient flow}, \ref{slow-roll FRW-KG eq.}), there is no evidence that the solution to the proper field equations should approach a gradient-flow attractor along a gradient-flow trajectory: all we know is that the late-time attractor is a gradient-flow trajectory, but there is no evidence as to how such trajectory is approached. It is possible to visualize all the considerations discussed above in fig. \ref{fig: gradient-flow plot}.

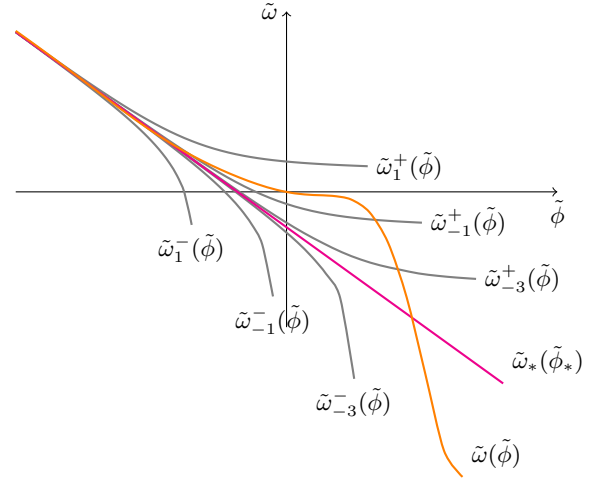
\begin{figure}[ht]
    \centering
    
    \begin{tikzpicture}[xscale=1.8,yscale=0.15]

    \draw[->] (-2,0) -- (2,0) node[below]{$\tsdil$};
    \draw[->] (0,-12) -- (0,16) node[left]{$\tomega$};

    \draw[domain=-2:0.6, smooth, thick, variable=\n, gray] plot ({\n}, {- (1/5)*sqrt(3)*\n+sqrt(6)*ln(exp(1) + (5/18)*(exp(-(12*sqrt(2)/5)*(\n)))}) node[right, black]{$\tomega_{1}^+(\tsdil)$};    
    \draw[domain=-2:1, smooth, thick, variable=\n, gray] plot ({\n}, {- (1/5)*sqrt(3)*\n+sqrt(6)*ln(exp(-1) + (5/18)*(exp(-(12*sqrt(2)/5)*(\n)))}) node[right, black]{$\tomega_{-1}^+(\tsdil)$};
    \draw[domain=-2:1.4, smooth, thick, variable=\n, gray] plot ({\n}, {- (1/5)*sqrt(3)*\n+sqrt(6)*ln(exp(-3) + (5/18)*(exp(-(12*sqrt(2)/5)*(\n)))}) node[right, black]{$\tomega_{-3}^+(\tsdil)$};

    \draw[domain=-2:-0.70, smooth, thick, variable=\n, gray] plot ({\n}, {- (1/5)*sqrt(3)*\n+sqrt(6)*ln(-exp(1) + (5/18)*(exp(-(12*sqrt(2)/5)*(\n)))}) node[below, black]{$\tomega_{1}^-(\tsdil)$};
    \draw[domain=-2:-0.10, smooth, thick, variable=\n, gray] plot ({\n}, {- (1/5)*sqrt(3)*\n+sqrt(6)*ln(-exp(-1) + (5/18)*(exp(-(12*sqrt(2)/5)*(\n)))}) node[below, black]{$\tomega_{-1}^-(\tsdil)$};
    \draw[domain=-2:0.5, smooth, thick, variable=\n, gray] plot ({\n}, {- (1/5)*sqrt(3)*\n+sqrt(6)*ln(-exp(-3) + (5/18)*(exp(-(12*sqrt(2)/5)*(\n)))}) node[below, black]{$\tomega_{-3}^-(\tsdil)$};

    \draw[domain=-2:1.6, smooth, thick, variable=\n, magenta] plot ({\n}, {- 5*sqrt(3)*\n+sqrt(6)*ln((5/18))}) node[above right, black]{$\tomega_* (\tsdil_*)$};

    \draw[domain=-2:1.3, smooth, thick, variable=\n, orange] plot ({\n}, {- 5*sqrt(3)*\n+sqrt(6)*ln((5/18))+(3*exp(2*\n))*(cos(\n*180)+(sin(\n*120))^2)+exp(3*\n)/8}) node[above right, black]{$\tomega (\tsdil)$};

    \end{tikzpicture}
    
    \caption{In this figure, gray lines correspond to different gradient-flow trajectories $\smash{\tomega_k^\pm(\tsdil)}$, while the magenta line is the trajectory corresponding to the exact scaling solution $\smash{\tomega_* (\tsdil_*)}$; in orange is depicted another trajectory that corresponds to a solution that has the scaling cosmology as the late-time attractor, but that does not approach the attractor through the steepest-descent trajectory.}
    
    \label{fig: gradient-flow plot}
    
\end{figure}

\section{Discussion} \label{sec: discussion}

In this article, we have studied cosmological late-time attractors of multi-field multi-exponential potentials. We have identified the sufficient conditions for positive-definite potentials to have scaling cosmologies as their late-time attractors, independently of initial conditions. This goes beyond previous analyses in the literature, which only discussed the linear stability of such solutions. As scaling solutions are known analytically, we can characterize any observable of interest. Several properties make such solutions very easy to handle with. To start, the field-space trajectory is a straight line. This has the accidental consequence that the trajectory is a gradient-flow solution. It should be noted that slow-roll gives the same trajectory, albeit with a different time dependence (and the slow-roll solution is not a late-time attractor). For scaling solutions, we have also shown that the norm of the potential gradient and the potential directional derivative along the field-space trajectory coincide, and give a precise measure of the $\epsilon$-parameter. This is because the scale factor is power-law in time and the non-geodesity factor is vanishing.

Parallel to the study of scaling-cosmology attractors, we have revisited the universal late-time bound on cosmic acceleration introduced in ref. \cite{Shiu:2023nph}. In particular, we have shown that, in its strongest version -- i.e. in the field-basis that maximizes it -- the lowest bound on the $\epsilon$-parameter can be computed as the distance from the origin to the convex hull of the couplings. Scaling cosmologies are special in this case, too, as they saturate the optimized bound.

For completeness, we have also revisited and expanded on several considerations appearing in ref. \cite{Shiu:2023nph}. In particular, we have detailed the argument as to why an unstabilized $d$-dimensional dilaton poses a strong obstacle to late-time acceleration in string-theoretic realizations, if the theory is at weak coupling. Also, we have elaborated on the fact that the time evolution of scaling cosmologies can be described by a single scalar field, where however the definition of the latter depends on the details of all the scalar-potential terms. Finally, we have recalled that a scaling cosmology involving a compact hyperbolic space is incompatible with scale separation.

Our results are completely general, regardless of higher-dimensional and/or higher-energy completions. Nonetheless, we have mostly studied their implications in the context of string compactifications. This is because the analytic knowledge that we have gained provides us with powerful tools to address several open problems. In relation to the Swampland Program, we have established a rigorous relationship between the coupling convex hull and late-time cosmic acceleration (as said above). Also, we have provided heuristic arguments for computing the de Sitter coefficient in terms of any given set of exponential couplings, even in the absence of a scaling-cosmology attractor. We have also noticed that a general class of potentials admits scaling-cosmology attractors where a subset of the scalars are stabilized. Being an attractor solution, this means that the dynamics of the rolling fields does not affect the stabilization.

An open problem that we plan to address soon is the discussion of the multi-field multi-exponential potentials for which we do not have yet a rigorous convergence result. As we have argued, we expect that one should be able to identify a truncated scalar potential, where only an asymptotically-dominant subset of the terms would define the attractor. Within this subset, there would instead be no dominating term and one should be able to identify the attractor of the complete potential through the attractor of the truncated potential. We plan to address this problem soon \cite{criticalpoints2}. If our intuition played out, this would allow us to test the de Sitter Swampland conjecture in a very wide class of potentials.

Another aspect that we plan to address in view of our convergence results is the Swampland Distance Conjecture \cite{Ooguri:2006in}. In loose terms, the latter implies that we should expect infinite towers of states becoming exponentially light as one approaches the moduli-space boundary. In the scenarios we consider, this is a double-edged sword but one on which we can have a clear handle on. On the one hand, states becoming light may useful to explain huge hierarchies of scales in nature. On the other hand, they are dangerous since the lower-dimensional effective-theory approach may break down. As we have seen, we have control over this since we can measure exactly how the Hubble, Kaluza-Klein and string scales evolve over time. A point that we emphasize is that a Kaluza-Klein mass becoming lighter and lighter over time does not necessarily imply an effective-field-theory breakdown. The theory cutoff may also change in time, as in the cases where it is defined in terms of the Kaluza-Klein scale.

In the context of the Swampland Program, refs. \cite{Hertzberg:2007wc, Achucarro:2018vey, Garg:2018reu, Andriot:2020lea, Andriot:2021rdy, Andriot:2022xjh} focus around linear combinations of the scalar potential and field-derivatives thereof. Such considerations may be addressed again in view of the improved understanding of the relationship between the scalar-potential directional derivative and the potential gradient norm that we have proposed here. Indeed, we now know that, if the coefficients in the linear combination of the scalar-potential partial derivatives are the components of the field-space trajectory, then such linear combination can be an analytic measure of late-time acceleration.

\onecolumngrid
\newpage

\appendix

\section{Late-time cosmological attractors} \label{app: late-time cosmological attractors}

Let $\smash{V = \sum_{i = 1}^m \Lambda_i \, \e^{- \kappa_d \gamma_{i a} \phi^a}}$ be the scalar potential for the canonically-normalized scalar fields $\phi^a$, with $a=1,\dots,n$, in a $d$-dimensional FLRW-background: one can reduce the cosmological scalar-field and Friedmann equations to a system of autonomous equations \cite{Copeland:1997et, Coley:1999mj, Guo:2003eu}. In terms of the variables
\begin{align*}
    x^a & = \dfrac{\kappa_d}{\sqrt{d-1} \sqrt{d-2}} \, \dfrac{\dot{\phi}^a}{H}, \\
    y^i & = \dfrac{\kappa_d \sqrt{2}}{\sqrt{d-1} \sqrt{d-2}} \, \dfrac{1}{H} \, \sqrt{\Lambda_i \, \e^{- \kappa_d \gamma_{i a} \phi^a}},
\end{align*}
and defining for simplicity
\begin{align*}
    f & = (d-1) H, \\
    c_{i a} & = \dfrac{1}{2} 
    \dfrac{\sqrt{d-2}}{\sqrt{d-1}} \, \gamma_{ia},
\end{align*}
the cosmological equations can indeed be expressed as
\begin{subequations}
\begin{align}
    \dot{x}^a & = \biggl[ -x^a (y)^2 + \sum_{i=1}^m {c_i}^{a} (y^i)^2 \biggr] \, f, \label{x-equation} \\
    \dot{y}^i & = \bigl[ (x)^2 - c_{ia} x^a \bigr] \, y^i f, \label{y-equation}
\end{align}
\end{subequations}
jointly with the two conditions
\begin{subequations}
\begin{align}
    & \dfrac{\dot{f}}{f^2} = - (x)^2, \label{f-equation} \\
    & (x)^2 + (y)^2 = 1. \label{(x,y)-plane sphere-condition}
\end{align}
\end{subequations}
(Here, the position of the $i$- and $a$-indices is arbitrary, since the former are just dummy labels and the latter refer to a field-space metric that is a Kronecker delta; if they appear together, such as in the coefficients $\smash{c_{ia}}$ and $\smash{{c_i}^a}$, we always write them on the left and on the right, respectively, to avoid confusion, and we place the $a$-index in the same (upper or lower) position as the other $x^a$-types variable that appear in the expression of reference. Einstein summation convention is understood for the metric-contraction on the $a$-indices and moreover we use the shorthand notations $\smash{(x)^2 = x_a x^a}$ and $\smash{(y)^2 = \sum_{i=1}^m (y^i)^2}$.)

\vspace{12pt}

Let the unknown functions be such that $x^a \in [-1, 1]$ and $y^i \in [-1,1]$ and let $t_0$ be the initial time. Let $c^a$ denote the minimum of the constant parameters for each $a$-index, i.e. $\smash{c^a = \min_i {c_i}^a}$. For each $a$-index, if $c^a > 0$, let $\smash{c_\infty^a = c^a}$; if $c^a \leq 0$, then let $\smash{c_\infty^a = 0}$. Let $\smash{\varphi(t) = \int_{t_0}^t \de s \; f(s) \bigl[ y(s) \bigr]^2}$.

\vspace{12pt}

A redefinition of time makes the autonomous system look simpler. Let $s(t)$ be the function $s: \; [t_0, \infty[ \, \to \, \mathbb{R}_0^+$ such that
\begin{equation}
    \dot{s}(t) = \dfrac{1}{f [s(t)]},
\end{equation}
with the identification $s(t_0) = t_0$.

\vspace{8pt}
\begin{lemma} \label{lemma: s - asymptotic behavior}
    One has $\smash{\lim_{t \to \infty} s(t) = \infty}$.
\end{lemma}

\begin{proof}
By eq. (\ref{f-equation}), one can write
\begin{align*}
    \ddot{s}(t) = \dfrac{\de}{\de t} \, \dot{s}(t) = - \dfrac{1}{f^2 [s(t)]} \, \dfrac{\de f}{\de s} [s(t)] \, \dot{s}(t) = \dfrac{1}{f [s(t)]} \, \bigl[x[s(t)]\bigr]^2.
\end{align*}
Because $f = f(t)>0$ at all times, both the first and the second derivative of $s=s(t)$ are positive and therefore the function diverges at infinity.
\end{proof}

\vspace{12pt}

For a function $\mathrm{g}(t) = g [s(t)]$, the time-derivative takes the form $\smash{\dot{\mathrm{g}}(t) = \dot{s}(t) \cdot \de g[s(t)]/\de s = \bigl( 1 / f[s(t)] \bigr) \cdot \de g[s(t)]/\de s}$. For the variables (notice the different notation between $\smash{(x^a,y^i)}$ and $\smash{(\x^a,\y^i)}$ from now on)
\begin{subequations}
\begin{align}
    \x^a(t) & = x^a [s(t)], \\
    \y^i(t) & = y^i [s(t)],
\end{align}
\end{subequations}
the cosmological autonomous system in eqs. (\ref{x-equation}, \ref{y-equation}) takes the form
\begin{subequations}
\begin{align}
    \dot{\x}^a & = -\x^a (\y)^2 + \sum_{i=1}^m {c_i}^{a} (\y^i)^2, \label{X-equation} \\
    \dot{\y}^i & = \bigl[(\x)^2 - c_{ia} \x^a \bigr] \, \y^i, \label{Y-equation}
\end{align}
\end{subequations}
with the constraint in eq. (\ref{(x,y)-plane sphere-condition}) being
\begin{equation} \label{(X,Y)-plane sphere-condition}
    (\x)^2 + (\y)^2 = 1.
\end{equation}
As time-integrals read $\smash{\int_{t_1}^{t_2} \de t \, \mathrm{g}(t) = \int_{s(t_1)}^{s(t_2)} \de s \, f(s) \, g(s)}$, one also has $\varphi(t) =  \int_{t_0}^t \de s \; f(s) \bigl[ y(s) \bigr]^2 = \int_{s^{-1}(t_0)}^{s^{-1}(t)} \de t' \, [\y(t')]^2$.

\vspace{12pt}

A simple but extremely general result turns out to be helpful below.

\vspace{8pt}
\begin{lemma} \label{lemma: Y>0}
    Given the initial conditions $\y^i(t_0)$, none of the $[\y^i(t)]^2$-terms ever change sign.
\end{lemma}

\begin{proof}
    To start, let $\y^i(t_0)>0$. If there is a time $t_i > t_0$ such that $\y^i(t_i)=0$, then, by standard results for ordinary differential equations \cite{Bahamonde:2017ize}, eq. (\ref{Y-equation}) imposes the identity $\y^i(t) = 0$ for all times $t \geq t_i$. Therefore, it is true that $\y^i(t) \geq 0$ for all times $t \geq t_0$. Let $\smash{C_i = \sum_{a=1}^n \ab c_{ia} \ab}$: then the chain of inequalities hold
    \begin{align*}
        \dot{\y}^i = \bigl[(\x)^2 - c_{ia} \x^a \bigr] \, \y^i \geq - c_{ia} \x^a \, \y^i \geq - \ab c_{ia} \x^a \ab \, \y^i \geq - \sum_{a=1}^n \ab c_{ia} \ab \ab \x^a \ab \, \y^i \geq - \sum_{a=1}^n \ab c_{ia} \ab \, \y^i = - C_i \y^i.
    \end{align*}
    An integration reveals the further inequality $\smash{\y^i(t) \geq \y^i(t_0) \, \e^{- C_i t}}$, which in turn implies the inequality $\y^i(t)>0$ for all times $t \geq t_0$. If instead let $\smash{-\I \, \y^{i}(t_0)>0}$, by the same proof, one finds $\smash{-\I \, \y_-^{i_-}(t)>0}$ for all times $t \geq t_0$. In other words, the $\smash{(\y^i)^2}$-terms never change sign over time.
\end{proof}

\subsection{Universal late-time convergence properties}

Let $\smash{\bc^a}$ be defined as the time-independent solutions to the equations
\begin{equation} \label{eq.: generic X-condition}
    c_{ia} \bc^a = (\bc)^2.
\end{equation}
As will be apparent below, the vector $\bc^a$ is going to be used to bound the late-time behavior of the solution $\x^a(t)$ to eq. (\ref{X-equation}), though $\bc^a$ is never assumed to be a solution to the autonomous equations.

If $\smash{\lim_{t \to \infty} \varphi(t) < \infty}$, then the late-time behavior of the problem is known to be such that $\smash{\lim_{t \to \infty} x^a(t) = \tilde{x}^a}$ and $\smash{\lim_{t \to \infty} y^i(t) = 0}$, which means that $\smash{(\tilde{x})^2 = 1}$, according to ref. \cite[lemma 3]{Shiu:2023nph}. Below is a characterization encompassing both finite and infinite values of $\smash{l = \lim_{t \to \infty} \varphi(t)}$.

\vspace{8pt}
\begin{theorem} \label{theorem: Xbar{X}-limit}
    If $\smash{\lim_{t \to \infty} \varphi(t) = \infty}$, one has
    \begin{equation} \label{eq.: Xbar{X}-limit}
        \lim_{t \to \infty} \bc_a \x^a(t) = (\bc)^2.
    \end{equation}
\end{theorem}

\begin{proof}
Starting from eq. (\ref{X-equation}), in view of eq. (\ref{eq.: generic X-condition}), one can easily write
\begin{align*}
    \dot{\x}_a \bc^a = - \x_a \bc^a (\y)^2 + \sum_{i=1}^m c_{ia} \bc^a (\y^i)^2 = \bigl[ - \x_a \bc^a + (\bc)^2 \bigr] \, (\y)^2.
\end{align*}
Let $\lambda(t) = \x_a(t) \, \bc^a$. Then one can write
\begin{align*}
    \dfrac{\de}{\de t} \bigl[ \e^{\varphi[s(t)]} \lambda(t) \bigr] & = \bigl[ [\y(t)]^2 \lambda(t) + \dot{\lambda}(t) \bigr] \, \e^{\varphi[s(t)]} = (\bc)^2 \, [\y(t)]^2 \, \e^{\varphi[s(t)]},
\end{align*}
which after an integration gives the identity
\begin{align*}
    \e^{\varphi[s(t)]} \lambda(t) - \lambda(t_0) = (\bc)^2 \, \bigl[ \e^{\varphi[s(t)]} - 1 \bigr].
\end{align*}
Because $\smash{\lim_{t \to \infty} \varphi(t) = \infty}$, one immediately finds eq. (\ref{eq.: Xbar{X}-limit}).
\end{proof}

\vspace{8pt}
\begin{corollary} \label{corollary: (X)^2-bound}
    If $\smash{\lim_{t \to \infty} \varphi(t) = \infty}$, one has
    \begin{equation} \label{eq.: (X)^2-bound}
        \liminf_{t \to \infty} [\x(t)]^2 \geq (\bc)^2.
    \end{equation}
\end{corollary}

\begin{proof}
   This follows immediately from the inequality $\smash{[\x_a - \bc_a] \, [\x^a - \bc^a] = (\x)^2 - 2 \, \x_a \bc^a + (\bc)^2 \geq 0}$.
\end{proof}

\vspace{8pt}
\begin{corollary} \label{corollary: Xbar{X}-bound}
    If $\smash{\lim_{t \to \infty} \varphi(t) = \infty}$, one has
    \begin{equation}
        \bc_a \x^a(t) - (\bc)^2 = \bigl[ \bc_a \x^a(t_0) \, - (\bc)^2 \bigr] \, \e^{- \varphi[s(t)]}.  \label{eq.: Xbar{X}-bound}
    \end{equation}
\end{corollary}

\begin{proof}
    This follows immediately from the proof of eq. (\ref{eq.: Xbar{X}-limit}).
\end{proof}

\vspace{8pt}
\begin{remark} \label{remark: Xbar{X}-bound sign independence}
    In theorem \ref{theorem: Xbar{X}-limit} and corollaries \ref{corollary: (X)^2-bound}-\ref{corollary: Xbar{X}-bound}, nothing depends on the individual signs of the $\smash{(\y^i)^2}$-terms: the only requirement is for the total sum to be $\smash{(\y)^2 = \sum_{i=1}^m (\y^i)^2>0}$ in order for it to be possible that $\smash{\lim_{t \to \infty} \varphi(t) = \infty}$.
\end{remark}

\vspace{8pt}
\begin{remark}
    As they are solutions to the same equations, since $f(t)>0$ at all times $t \geq t_0$, the solutions $\smash{\bc^a}$ to eq. (\ref{eq.: generic X-condition}) are also solutions $\smash{\overline{c}^a}$ to the analogous problem in terms of the original variables, namely $\smash{\overline{c}^a = \bc^a}$. Therefore, the same convergence results apply to the original variables $\smash{x^a(t)}$, too.
\end{remark}

\vspace{12pt}

If eq. (\ref{eq.: generic X-condition}) admits solutions, one can obtain, in addition to the bound of ref. \cite[remarks 6.2-6.3]{Shiu:2023nph}, a new bound. Let $\smash{(\bc_\sigma)^a}$ label all solutions to eq. (\ref{eq.: generic X-condition}) and let $\smash{\bc_\infty^2 = \min \lbrace 1, \max_\sigma (\bc_\sigma)^2 \rbrace}$.

\vspace{8pt}
\begin{remark} \label{remark: alternative f-bound}
    By corollary \ref{corollary: (X)^2-bound}, for some large time $t_\infty > t_0$, at all times $t > t_\infty$ one has
    \begin{equation} \label{eq.: alternative f-bound}
        f(t) \leq \dfrac{1}{\bc_\infty^2 \, (t-t_\infty) + \dfrac{1}{f(t_\infty)}}.
    \end{equation}
    This bound is not necessarily stronger or weaker than the bound in ref. \cite{Shiu:2023nph}. It requires different assumptions and its relative reach must be judged on a case-by-case basis.
\end{remark}

\vspace{8pt}
\begin{corollary}
    If $\smash{\bc_\infty^2 = 1} < \max_\sigma (\bc_\sigma)^2$, then $\smash{\lim_{t \to \infty} \varphi(t) < \infty}$.
\end{corollary}

\begin{proof}
    This follows immediately from the considerations in ref. \cite[lemma 1]{Shiu:2023nph} and in remark \ref{remark: alternative f-bound}.
\end{proof}

\subsection{Late-time convergence to critical points}

The critical points of the autonomous system in eqs. (\ref{X-equation}, \ref{Y-equation}), along with eq. (\ref{(X,Y)-plane sphere-condition}),  are the solutions $\smash{(\x^a, \y^i) = ((\x_0)^a, (\y_0)^i)}$ to the equations $\smash{(\dot{\x}^a, \dot{\y}^i) = (0,0)}$; as the equations are algebraic, the solutions can be determined and classified analytically \cite{Collinucci:2004iw}. Here, critical points are conveniently distinguished in two categories.
\begin{enumerate}
    \item Proper critical points $\smash{(\bx^a, \by^i)}$ are solutions to the equations
    \begin{align*}
        -\bx^a (\by)^2 + \sum_{i=1}^m {c_i}^{a} (\by^i)^2 & = 0,  \\
        (\bx)^2 - c_{ia} \bx^a & = 0.
    \end{align*}
    \item Non-proper critical points $\smash{(\hx^a, (\hy^\eta,\check{\y}^\zeta))}$ are solutions to the equations
    \begin{align*}
        -\hx^a (\hy)^2 + \sum_{\eta} {c_\eta}^{a} (\hy^\eta)^2 & = 0,  \\
        (\hx)^2 - c_{\eta a} \hx^a & = 0, \\[2.0ex]
        \check{\y}^{\zeta} & = 0.
    \end{align*}
\end{enumerate}
Both cases represent fixed points in the $\smash{(\x^a,\y^i)}$-plane. However, non-proper critical points effectively truncate the problem to a lower-dimensional coefficient matrix $\smash{\hat{c}_{\eta a} = c_{\eta a}}$, and one always has $\smash{(\bx)^2 \leq (\hx)^2}$ \cite{Hartong:2006rt}. Proper critical points where there exist directions such that $\smash{\by^{i_0}=0}$ are accidental non-proper critical points. Critical points where $\smash{\y^i=0}$ for all $i$-indices are degenerate non-proper critical points, denoted as $\smash{(\tilde{\x}^a,\tilde{\y}^i=0)}$.

Below is a summary of the mathematical classifications of ref. \cite{Collinucci:2004iw}.

\vspace{8pt}
\begin{lemma} \label{lemma: (X,Y)-plane critical points (rank=m)}
    If $\mathrm{rank} \, c_{ia} = m$, let $\smash{A_{ij} = c_{i a} {c_j}^a}$; the proper critical points exist and they read
    \begin{subequations}
    \begin{align}
        \bx^a & = \dfrac{1}{k} \sum_{i=1}^m \sum_{j=1}^m {c_i}^a (A^{-1})^{ij},  \label{eq.: X-critical point} \\
        (\by^i)^2 & = \dfrac{k-1}{k^2} \, \sum_{j=1}^m (A^{-1})^{ij}, \label{eq.: Y-critical point}
    \end{align}
    \end{subequations}
    with
    \begin{equation} \label{eq.: k-constant}
        k = \sum_{i=1}^m \sum_{j=1}^m (A^{-1})^{ij}.
    \end{equation}
\end{lemma}

\vspace{8pt}
\begin{lemma} \label{lemma: (X,Y)-plane critical points (rank<m)}
    If $\mathrm{rank} \, c_{ia} = r < m$, let $\smash{c_{\iota a}}$ denote $r$ linearly-independent vectors, for $\iota=1, \dots, r$, and let the other vectors be expressed as $\smash{c_{\iota' a} = \sum_{\iota=1}^r \lambda_{\iota' \iota} c_{\iota a}}$, for $\smash{\iota' = r+1, \dots, m}$. If $\smash{\sum_{\iota=1}^m \lambda_{\iota' \iota} = 1}$, then, given the matrix $\smash{B_{ab} = \sum_{i=1}^m c_{i a} c_{i b}}$, proper critical points exist and they read
    \begin{equation} \label{eq.: X-critical point 2}
        \bx^a = \dfrac{1}{k} \, (B^{-1})^{ab} \sum_{i=1}^m c_{i b},
    \end{equation}
    with
    \begin{equation} \label{eq.: k-constant 2}
        k = \delta_{a b} \biggl[ (B^{-1})^{a c} \sum_{i=1}^m c_{i c} \biggr] \biggl[ (B^{-1})^{b d} \sum_{j=1}^m c_{i d} \biggr].
    \end{equation}
    These solutions are not unique. If $\smash{\sum_{\iota=1}^m r_{\iota' \iota} \neq 1}$, there are no such solutions.
\end{lemma}

\vspace{8pt}
\begin{lemma} \label{lemma: (X,Y)-plane critical points (x^2=0)}
    If $\mathrm{rank} \, c_{ia} = m$, there are no proper critical points such that $\smash{(\bx)^2 = 0}$. If $\mathrm{rank} \, c_{ia} < m$, then proper critical points such that $\smash{(\bx)^2 = 0}$ exist and they are such that the $\smash{(\by^i)^2}$-terms belong to the $A$-matrix kernel, i.e. $\smash{(\by^i)^2 \in \ker A}$.
\end{lemma}

\vspace{8pt}
\begin{remark}
    Along with the solutions discussed in lemmas \ref{lemma: (X,Y)-plane critical points (rank=m)}, \ref{lemma: (X,Y)-plane critical points (rank<m)} and \ref{lemma: (X,Y)-plane critical points (x^2=0)}, which are proper critical points, one also has non-proper critical points: the non-vanishing $\y$-terms give a truncated system that can be solved exactly as above.
\end{remark}

\vspace{8pt}
\begin{remark}
    As they are solutions to the same equations, since $f(t)>0$ at all times $t \geq t_0$, the critical points $\smash{(\overline{x}^a, \overline{y}^i)}$ of the autonomous system in eqs. (\ref{x-equation}, \ref{y-equation}), along with eqs. (\ref{f-equation}, \ref{(x,y)-plane sphere-condition}), are identical, namely $\smash{(\overline{x}^a, \overline{y}^i) = (\bx^a, \by^i)}$.
\end{remark}

\vspace{12pt}

Here ends the overview of the main results of ref. \cite{Collinucci:2004iw} and begins an original discussion of the late-time convergence of the solutions to the critical points of the autonomous system.

\vspace{12pt}

Perturbative analyses of the late-time behavior of autonomous systems in cosmological scenarios appear in refs. \cite{Halliwell:1986ja, Malik:1998gy, Coley:1999mj, Guo:2003eu, Guo:2003rs, Kim:2005ne, Hartong:2006rt} and they show the conditions under which the critical points are perturbative late-time attractors. Below are a series of mathematical results that show the conditions under which the critical points are late-time attractors; as the current analysis is not perturbative, the convergence results are universal. A formulation in terms of physical observables and an analysis of the implications of such convergence results is in the main text.

\vspace{12pt}

From now on it is assumed that $\smash{\lim_{t \to \infty} \varphi(t) = \infty}$. Furthermore, let $\y^i(t_0)>0$ for all $i$-directions, unless differently stated. Finally, let $\mathrm{rank} \, c_{ia} = m$.

\vspace{8pt}
\begin{remark}
    As obvious from eqs. (\ref{X-equation}, \ref{Y-equation}), all proper critical points satisfy the identities
    \begin{subequations}
        \begin{align}
            c_{ia} \bx^a & = (\bx)^2, \label{eq.: cbar{X}-identity} \\
            \sum_{i=1}^m c_{ia} (\by^i)^2 & = \overline{x}_a (\by)^2. \label{eq.: cbar{Y}-identity}
        \end{align}
    \end{subequations}
    It is obvious that eqs. (\ref{eq.: cbar{X}-identity}, \ref{eq.: cbar{Y}-identity}) are more restrictive than the sole eq. (\ref{eq.: generic X-condition}).  Below are results for proper critical points $\smash{(\bx^a, \by^i)}$ satisfying both eqs. (\ref{eq.: cbar{X}-identity}, \ref{eq.: cbar{Y}-identity}). Of course, all results so far apply independently of whether critical points exist since solutions to eq. (\ref{eq.: cbar{X}-identity}, \ref{eq.: cbar{Y}-identity}) are also solutions to eq. (\ref{eq.: generic X-condition}).
\end{remark}

\vspace{8pt}
\begin{lemma} \label{lemma: Y^2-equality + X^2 integral bound}
    For all times $t \geq t_0$, the identity holds
    \begin{equation} \label{eq.: Y^2-equality}
        \prod_{i=1}^n [\y^i(t)]^{(\by^i)^2} = \biggl[ \prod_{i=1}^n [\y^i(t_0)]^{(\by^i)^2} \biggr] \, \e^{(\by)^2 \!\! \int_{t_0}^t \de t' \, [[\x(t')]^2 - \bx_a \x^a(t')]}
    \end{equation}
    In particular, if $\smash{(\by^i)^2 \geq 0}$ for all $i$-directions there exists a positive constant $\omega>0$ such that
    \begin{equation} \label{eq.: X^2 integral bound}
        \int_{t_0}^t \de t' \, \bigl[ [\x(t')]^2 - \bx_a \x^a(t') \bigr] \leq \omega.
    \end{equation}
\end{lemma}

\begin{proof}
    As $\y^i(t) > 0$ at all times $t \geq t_0$, one can write eq. (\ref{Y-equation}) as
    \begin{align*}
        \dfrac{\de}{\de t} \, \mathrm{ln} \, \y^i = \dfrac{\dot{\y}^i}{\y^i} = (\x)^2 - c_{ia} \x^a.
    \end{align*}
    In view of eq. (\ref{eq.: cbar{Y}-identity}), this in turn allows one to write
    \begin{align*}
        \sum_{i=1}^m (\by^i)^2 \dfrac{\dot{\y}^i}{\y^i} = \sum_{i=1}^m (\by^i)^2 \bigl[ (\x)^2 - c_{ia} \x^a \bigr] = \bigl[ (\x)^2 - \bx_a \x^a \bigr] \, (\by)^2.
    \end{align*}
    By manipulating this as
    \begin{align*}
        \sum_{i=1}^m (\by^i)^2 \, \dfrac{\de}{\de t} \, \mathrm{ln} \, \y^i = \dfrac{\de}{\de t} \, \mathrm{ln} \prod_{i=1}^m (\y^i)^{(\by^i)^2} = (\by)^2 \bigl[ (\x)^2 - \bx_a \x^a \bigr],
    \end{align*}
    an integration gives
    \begin{align*}
        \mathrm{ln} \prod_{i=1}^m \dfrac{[\y^i(t)]^{(\by^i)^2}}{[\y^i(t_0)]^{(\by^i)^2}} = (\by)^2 \!\! \int_{t_0}^t \de t' \, \bigl[ [\x(t')]^2 - \bx_a \x^a(t') \bigr],
    \end{align*}
    which gives eq. (\ref{eq.: Y^2-equality}). Because at all times $t \geq t_0$ the variables $\y^i(t)$ are bounded as $0 < \y^i(t) \leq 1$, if $\smash{(\by^i)^2 \geq 0}$ for all $i$-directions, the right-hand side of eq. (\ref{eq.: Y^2-equality}) is bounded as $\smash{\prod_{i=1}^n [\y^i(t)]^{(\by^i)^2} \leq 1}$, which implies the inequality
    \begin{align*}
        \int_{t_0}^t \de t' \, [[\x(t')]^2 - \bx_a \x^a(t')] \leq - \mathrm{ln} \, \prod_{i=1}^n [\y^i(t_0)]^{\frac{(\by^i)^2}{(\by)^2}} = \omega.
    \end{align*}
\end{proof}

\vspace{8pt}
\begin{lemma} \label{lemma: varphi-bound}
    If $\smash{(\by^i)^2 \geq 0}$ for all $i$-directions, for any positive constant $\delta$ such that $0 < \delta < (\by)^2$, there exists a time $t_\delta$ such that, for all times $t \geq t_\delta$, one has
    \begin{equation} \label{eq.: varphi-bound}
        \varphi[s(t)] = \int_{t_0}^{t} \de t' \, [\y(t')]^2 \geq t \delta.
    \end{equation}
\end{lemma}

\begin{proof}
    By eq. (\ref{eq.: Xbar{X}-limit}), for any number $\delta'$ such that $\delta < \delta' < (\by)^2$, one can write
    \begin{align*}
        \lim_{t \to \infty} \overline{x}_a \x^a(t) = (\bx)^2 = 1 - (\by)^2 < 1 - \delta'.
    \end{align*}
    Therefore, there exists a time $\smash{t_{\delta'} > t_0}$ such that for all times $\smash{t \geq t_{\delta'}}$, the inequality holds
    \begin{align*}
        \overline{x}_a \x^a(t) < 1 - \delta'.
    \end{align*}
    For any given time $\smash{t \geq t_{\delta'}}$, this then implies the inequality
    \begin{align*}
        \int_{t_0}^t \de t' \, \bigl[ [\x(t')]^2 - \bx_a \x^a(t') \bigr] & = \int_{t_0}^t \de t' \, \bigl[ [\x(t')]^2 - (1 - \delta') \bigr] + \int_{t_0}^t \de t' \, \bigl[ (1 - \delta') - \bx_a \x^a(t') \bigr] \\
        & = \int_{t_0}^t \de t' \, \bigl[ [\x(t')]^2 - (1 - \delta') \bigr] + \int_{t_0}^{t_{\delta'}} \de t' \, \bigl[ (1 - \delta') - \bx_a \x^a(t') \bigr] + \int_{t_{\delta'}}^t \de t' \, \bigl[ (1 - \delta') - \bx_a \x^a(t') \bigr] \\
        & > \int_{t_0}^t \de t' \, \bigl[ [\x(t')]^2 - (1 - \delta') \bigr] + \int_{t_0}^{t_{\delta'}} \de t' \, \bigl[ (1 - \delta') - \bx_a \x^a(t') \bigr],
    \end{align*}
    which in turn, in view of eq. (\ref{eq.: X^2 integral bound}), implies the chain of inequalities
    \begin{align*}
        \omega \geq \int_{t_0}^t \de t' \, \bigl[ [\x(t')]^2 - \bx_a \x^a(t') \bigr] > \int_{t_0}^t \de t' \, \bigl[ [\x(t')]^2 - (1 - \delta') \bigr] + \int_{t_0}^{t_{\delta'}} \de t' \, \bigl[ (1 - \delta') - \bx_a \x^a(t') \bigr].
    \end{align*}
    One is thus able to further write
    \begin{align*}
        - \int_{t_0}^t \de t' \, \bigl[ [\x(t')]^2 - (1 - \delta') \bigr] > - \omega + \int_{t_0}^{t_{\delta'}} \de t' \, \bigl[ (1 - \delta') - \bx_a \x^a(t') \bigr] \geq - \omega - \int_{t_0}^{t_{\delta'}} \de t' \, \bigl\ab (1 - \delta') - \bx_a \x^a(t') \bigr\ab.
    \end{align*}
    By defining the positive constant
    \begin{align*}
        \omega_{\delta'} = \omega +  \int_{t_0}^{t_{\delta'}} \de t' \, \bigl\ab (1 - \delta') - \bx_a \x^a(t') \bigr\ab,
    \end{align*}
    this eventually gives the inequality
    \begin{align*}
        \varphi[s(t)] = \int_{t_0}^{t} \de t' \, [\y(t')]^2 = \delta' (t-t_0) - \int_{t_0}^t \de t' \, \bigl[ [\x(t')]^2 - (1 - \delta') \bigr] > \delta' (t-t_0) - \omega_{\delta'}.
    \end{align*}
    For all times $\smash{t > (\delta' t_0 + \omega_{\delta'}) / (\delta' - \delta)}$, one then trivially has eq. (\ref{eq.: varphi-bound}).
\end{proof}

\vspace{8pt}
\begin{remark} \label{remark: XbarX integral finiteness}
    By eq. (\ref{eq.: varphi-bound}), the function $\varphi(t)$ is bounded from below as $\smash{\varphi(t) \geq t \delta}$, for large-enough times $t \geq t_\delta$. This means that the function $\smash{\overline{x}_a \x^a(t)}$ approaches the value $\smash{(\bx)^2}$ at an exponential rate, by eq. (\ref{eq.: Xbar{X}-bound}): in particular, at all times $t \geq t_\delta$, one can write the chain of inequalities
    \begin{align*}
        \bigl\ab\bx_a \x^a(t) - (\bx)^2\bigr\ab  \leq \bigl\ab \bx_a \x^a(t_0) \, - (\bx)^2 \bigr\ab \, \e^{- \varphi[s(t)]} \leq \bigl\ab \bx_a \x^a(t_0) \, - (\bx)^2 \bigr\ab \, \e^{- t \delta},
    \end{align*}
    and therefore the integral is finite
    \begin{equation} \label{eq.: XbarX integral finiteness}
        \theta = \int_{t_0}^\infty \de t \, \bigl\ab \bx_a \x^a(t) \, - (\bx)^2 \bigr\ab < \infty.
    \end{equation}
\end{remark}

\vspace{8pt}
\begin{remark} \label{remark: (X-barX)^2 integral bouund}
    In view of the trivial inequality
    \begin{align*}
        \int_{t_0}^\infty \de t \, [\x_a(t) - \bx_a] [\x^a(t) - \bx^a] & = \int_{t_0}^\infty \de t \, \bigl[ [\x(t)]^2 - \bx_a \x^a(t) \bigr] - \int_{t_0}^\infty \de t \, \bigl[ \bx_a \x^a(t) \, - (\bx)^2 \bigr] \\
        & \leq \int_{t_0}^\infty \de t \, \bigl[ [\x(t)]^2 - \bx_a \x^a(t) \bigr] + \int_{t_0}^\infty \de t \, \bigl\ab \bx_a \x^a(t) \, - (\bx)^2 \bigr\ab,
    \end{align*}
    by combining eqs. (\ref{eq.: X^2 integral bound}, \ref{eq.: XbarX integral finiteness}), one can show the fundamental integral inequality
    \begin{equation} \label{eq.: (X-barX)^2 integral bouund}
        \int_{t_0}^\infty \de t \, [\x_a(t) - \bx_a] [\x^a(t) - \bx^a] \leq \omega + \theta < \infty.
    \end{equation}
\end{remark}

\vspace{8pt}
\begin{theorem} \label{theorem: critical-point convergence}
    If the initial conditions are such that $\smash{\y^i(t_0)>0}$ and if the critical points are such that $\smash{(\by^i)^2 \geq 0}$ for all $i$-directions, then one has
    \begin{equation} \label{eq.: critical-point convergence}
        \lim_{t \to \infty} \x^a(t) = \bx^a.
    \end{equation}
\end{theorem}

\begin{proof}
    This result is obvious because the integral over an infinite interval of the positive-definite and Lipschitz function $\smash{\nu(t) = [\x_a(t) - \bx_a] [\x^a(t) - \bx^a]}$ is finite.
\end{proof}

\vspace{8pt}
\begin{remark}
    Because $f(t)>0$ at all times, it is obvious that the same convergence results apply equally for the original variables $\smash{x^a(t)}$.
\end{remark}

\vspace{8pt}
\begin{remark}
    Proper critical points $\smash{(\bx^a, \by^i)}$ are not necessarily late-time attractors. A simple counter-example is the following: if $[\y^i(t_0)]^2 > 0$ and $(\by^i)^2<0$, the latter can never be an attractor because, by lemma \ref{lemma: Y>0}, each $[\y^i(t)]^2$-term has the same sign as the initial condition at all times. In fact, one can see easily that the proof of lemma \ref{lemma: Y^2-equality + X^2 integral bound} immediately breaks down if some of the $\by^i$-terms are negative.
\end{remark}

\newpage

\section{String compactifications and running moduli}
In this appendix we review for completeness the scalar potentials for the dilaton and for the compactification volume of type-II and heterotic string theories in the presence of non-trivial internal curvature, NSNS-fluxes, RR-fluxes, D-branes and O-planes, and/or a non-zero Casimir-energy.

\subsection{String dimensional reductions} \label{app: string dimensional reductions}

Let the field theory describing the low-energy limit of a string-theory model be compactified over an $s_1$-dimensional space $\smash{\mathrm{K}_{s_1}}$ and an $s_2$-dimensional space $\smash{\mathrm{K}_{s_2}}$, leaving only a $d$-dimensional non-compact space $\smash{\mathrm{X}_{1,d-1}}$, with $d+s_1+s_2=10$. In particular, the coordinates $x^\mu$, with $\mu = 0, 1, \dots, d-1$ run in the non-compact spacetime, while the coordinates $y^{m_1}$ and $y^{m_2}$, with $m_1=1,\dots,s_1$ and $m_2=1,\dots,s_2$, label the compact directions. In the simpler case of an isotropic compactification, one can simply set $s_1=10-d$ and $s_2 = 0$; in the more complicated case of an even more anisotropic compactifcation, a generalization is instead trivial.

Let $l_s = 2 \pi \sqrt{\alpha'}$ be the string length, where $T=1/(2\pi\alpha')$ defines the string tension. It is possible to define the $d$-dimensional Einstein-frame metric $\smash{\tilde{g}_{\mu \nu}}$, which determines a line element expressed as $\smash{d \tilde{s}_{1,d-1}^2 = \tilde{g}_{\mu \nu} \, \de x^\mu \de x^\nu}$, by parameterizing the string-frame metric $G_{MN}$, defined through the line element $d s_{10}^2 = G_{MN} \de x^M \de x^N$, as
\begin{equation} \label{d-dim. Einstein-frame metric (phi,omega1,omega2)}
    ds_{1,9}^2 = \dfrac{\vevc \, \e^{\frac{\sdil}{2}} \, \tilde{g}_{\mu \nu} \, \de x^\mu \de x^\nu}{\e^{\frac{2 s_1}{d-2} \omega_1 + \frac{2 s_2}{d-2} \omega_2}} + \e^{\frac{\sdil}{2} + 2 \omega_1} \, \breve{g}_{m_1 n_1} \, \de y^{m_1} \de y^{n_1} + \e^{\frac{\sdil}{2} + 2 \omega_2} \, \breve{g}_{m_2 n_2} \, \de y^{m_2} \de y^{n_2}.
\end{equation}
Here, the internal metrics are assumed to be normalized in such a way that $\smash{\breve{\mathrm{vol}} \, \mathrm{K}_{s_i} = \int_{\mathrm{K}_{s_i}} \de^{s_i} y \, \sqrt{\breve{g}_{s_i}} = l_s^{s_i}}$. Along with the shifted 10-dimensional dilaton $\sdil = \Phi - \mathrm{ln} \, g_s$, where $g_s = \e^{\langle \Phi \rangle}$ is the string-coupling vacuum expectation value, there are two more moduli, the radions $\omega_1$ and $\omega_2$, which control the Einstein-frame volume scales. An arbitrary constant $\vevc$ has also been inserted: fixing it as $\vevc = \e^{2 s_1 \langle \omega_1 \rangle / (d-2)} \, \e^{2 s_2 \langle \omega_2 \rangle / (d-2)}$ makes the $d$-dimensional metric components correspond in all frames, in the vacuum. A discussion of the case where a vacuum expectation value for the dilaton and/or the radions is not defined is conveniently deferred to later comments.

It can be more convenient to work in terms of the $d$-dimensional dilaton $\delta$ and of the string-frame volume-controlling radions $\sigma_1$ and $\sigma_2$, which can be introduced by parameterizing the metric as\footnote{Here, the $(\sdil, \omega_1, \omega_2)$- and $(\delta, \sigma_1, \sigma_2)$-bases can be related via the linear transformations
\begin{align*}
    \sdil & = \delta + \dfrac{s_1}{2} \, \sigma_1 + \dfrac{s_2}{2} \, \sigma_2, \\
    \omega_1 & = - \dfrac{\delta}{4} - \dfrac{s_1-8}{8} \, \sigma_1 - \dfrac{s_2}{8} \, \sigma_2, \\
    \omega_2 & = - \dfrac{\delta}{4} - \dfrac{s_1}{8} \, \sigma_1 - \dfrac{s_2-8}{8} \, \sigma_2,
\end{align*}
or alternatively their inverse
\begin{align*}
    \delta & = \dfrac{d-2}{8} \, \sdil - \dfrac{s_1}{2} \, \omega_1 - \dfrac{s_2}{2} \, \omega_2, \\
    \sigma_1 & = \dfrac{1}{4} \, \sdil + \omega_1, \\
    \sigma_2 & = \dfrac{1}{4} \, \sdil + \omega_2.
\end{align*}
}
\begin{equation} \label{d-dim. Einstein-frame metric (delta,sigma1,sigma2)}
    \begin{split}
        ds_{1,9}^2 = \e^{\frac{4 \delta}{d-2}} \, \vevc \, \tilde{g}_{\mu \nu} \, \de x^\mu \de x^\nu + \e^{2 \sigma_1} \, \breve{g}_{m_1 n_1} \, \de y^{m_1} \de y^{n_1} + \e^{2 \sigma_2} \, \breve{g}_{m_2 n_2} \, \de y^{m_2} \de y^{n_2}.
    \end{split}
\end{equation}
Although the 10-dimensional dilaton and the Einstein-frame radions are more intuitive quantities to work with, the radions are kinetically-mixed in the $d$-dimensional Einsten frame; the $d$-dimensional dilaton and the string-frame radions are instead always diagonal.\footnote{It is possible to diagonalize the kinetic terms for the Einstein-frame radions by a field rotation in the $(\omega_1, \omega_2)$-subspace: this is standard in type-IIB theories, where one keeps working with the 10-dimensional dilaton; in type-IIA theories it is instead typically convenient to work in terms of the $d$-dimensional dilaton.} In particular, in terms of the canonically-normalized fields\footnote{As already mentioned, in the $\smash{(\sdil, \omega_1, \omega_2)}$-basis, one cannot define a canonical normalization without first rotating the $\smash{(\omega_1, \omega_2)}$-subspace basis. However, for a single radion this issue is not there and one can define canonically-normalized fields as
\begin{align*}
    \tsdil & = \dfrac{1}{\sqrt{2}} \, \dfrac{\sdil}{\kappa_d}, \\
    \tomega & = 2 \sqrt{2} \, \dfrac{\sqrt{10-d}}{\sqrt{d-2}} \, \dfrac{\omega}{\kappa_d}.
\end{align*}}
\begin{align*}
    \tdelta & = \dfrac{2}{\sqrt{d-2}} \, \dfrac{\delta}{\kappa_d}, \\
    \tsigma_1 & = \sqrt{s_1} \, \dfrac{\sigma_1}{\kappa_d}, \\
    \tsigma_2 & = \sqrt{s_2} \, \dfrac{\sigma_2}{\kappa_d},
\end{align*}
the string-frame Einstein-Hilbert action, with a gravitational coupling $\smash{2 \kappa_{10}^2 = l_s^8 / 2 \pi}$, and including the dilaton, is in the end reduced to the $d$-dimensional action
\begin{equation} \label{d-dim. Einstein-frame EH-action (delta,sigma1,sigma2)}
    \begin{split}
        S_{\mathrm{EH}} = \int_{\mathrm{X}_{1,d-1}} \de^{d} x \, \sqrt{-\tilde{g}_{1,d-1}} \; \biggl[ \dfrac{1}{2 \kappa_{d}^2} \tilde{R}_{1,d-1} - \dfrac{1}{2} \, \tilde{g}^{\mu \nu} \der_\mu \tdelta \der_\nu \tdelta - \dfrac{1}{2} \, \tilde{g}^{\mu \nu} \der_\mu \tsigma_1 \der_\nu \tsigma_1 - \dfrac{1}{2} \, \tilde{g}^{\mu \nu} \der_\mu \tsigma_2 \der_\nu \tsigma_2 \biggr],
    \end{split}
\end{equation}
where the $d$-dimensional gravitational coupling is
\begin{equation} \label{d-dim. gravitational coupling}
    2 \kappa_d^2 = \dfrac{1}{2 \pi} \dfrac{g_s^2}{\vevc^{\frac{d-2}{2}} l_s^{2-d}} = \dfrac{1}{2 \pi} \dfrac{g_s^2}{l_s^{2-d}} \, \e^{2 \langle \delta \rangle} = \dfrac{g_s^2 \, l_s^{d-2}}{2 \pi \, \e^{s_1 \langle \omega_1\rangle + s_2 \langle \omega_2 \rangle}}.
\end{equation}
One can express the $d$-dimensional Planck mass $\smash{m_{\p,d}^2 = 1 / l_{\p,d}^2 = 1 / (\kappa_d^2)^{\frac{2}{d-2}}}$ as
\begin{equation} \label{d-dim. Planck mass}
    m_{\p,d}^2 = \Bigl( \dfrac{4 \pi}{g_s^2} \Bigr)^{\frac{2}{d-2}} \, \dfrac{\vevc}{l_s^2} = \dfrac{1}{l_s^2} \, \Bigl( \dfrac{4 \pi}{g_s^2} \Bigr)^{\frac{2}{d-2}} \, \e^{- \frac{2 \, \kappa_d \langle \tdelta \rangle}{\sqrt{d-2}}}.
\end{equation}
In these conventions, the $d$-dimensional Planck mass is constant.

A dimensional reduction of a generic field equation from string- to $d$-dimensional Einstein-frame metric also provides the Kaluza-Klein mass scale, which is
\begin{equation} \label{d-dim. KK-scale}
    m_{\mathrm{KK},d}^2 = \e^{\frac{2 \, \kappa_d \langle \tdelta \rangle}{\sqrt{d-2}} - \frac{2 \, \kappa_d \langle \tsigma \rangle}{\sqrt{10-d}}} \, \Bigl( \dfrac{g_s^2}{4 \pi} \Bigr)^{\frac{2}{d-2}} m_{\p,d}^2 = \e^{\frac{- 4 \sqrt{2} \, \kappa_d \langle \tomega \rangle}{\sqrt{d-2} \sqrt{10-d}}} \, \Bigl( \dfrac{g_s^2}{4 \pi} \Bigr)^{\frac{2}{d-2}} m_{\p,d}^2.
\end{equation}
By a dimensional reduction of the string-frame mass $\smash{m_s^2 = 1 / \alpha'}$ of string excitations, one also gets the $d$-dimensional string mass
\begin{equation} \label{d-dim. string scale}
    m_{s,d}^2 = 4\pi^2 \, \e^{\frac{2 \, \kappa_d \langle \tdelta \rangle}{\sqrt{d-2}}} \, \Bigl( \dfrac{g_s^2}{4 \pi} \Bigr)^{\frac{2}{d-2}} m_{\p,d}^2 = 4 \pi^2 \, \e^{\frac{\kappa_d  \langle \tsdil \rangle}{\sqrt{2}} - \frac{\sqrt{10-d} \, \kappa_d \langle \tomega \rangle}{\sqrt{2} \sqrt{d-2}}} \, \Bigl( \dfrac{g_s^2}{4 \pi} \Bigr)^{\frac{2}{d-2}} m_{\p,d}^2.
\end{equation}
Both of these expressions are referred to an isotropic compactification, which only involves one radion, in order to avoid subtleties with different competing Kaluza-Klein scales. One should consider the scalar-field values and not the vacuum-expectation values in case the fields are evolving over time.

If vacuum expectation values for $\sdil$ and/or $\omega_1$ and $\omega_2$ are not defined, then the above definitions of $g_s$ and $\vevc$ must be reconsidered. In terms of the metrics in eqs. (\ref{d-dim. Einstein-frame metric (phi,omega1,omega2)}, \ref{d-dim. Einstein-frame metric (delta,sigma1,sigma2)}), the constants $g_s$ and $\vevc$ are now just arbitrary numbers, since there is no such thing as string- and Einstein-frame non-compact metric components being identical, here. One can check that relative ratios of physical quantities are independent of such $\chi$, while on the other hand it is convenient to set $g_s = 1$.\footnote{
To clarify these statements, one can consider the Kaluza-Klein or string mass, or the dimensional reduction of a scalar-field theory.

Let $\varphi$ be a scalar of string-frame mass $\smash{m_{10}^2 = \xi^2 / l_s^2}$, which corresponds to the string-frame action
\begin{align*}
    S & = \int_{\mathrm{X}_{1,9}} \de^{10} x \, \sqrt{-G_{10}} \; \biggl[ - \dfrac{1}{2} \, G^{M N} \der_M \varphi \der_N \varphi - \dfrac{1}{2} \, m_{10}^2 \, \varphi^2 \biggr].
\end{align*}
In $d$-dimensional Einstein frame, this action is
\begin{align*}
    S = \int_{\mathrm{X}_{1,d-1}} \!\!\!\!\!\! \de^{d} x \, \sqrt{-\tilde{g}_d} \; \e^{2 \sdil} \, \biggl[ & - \dfrac{1}{2} \, \tilde{g}^{\mu \nu} \der_\mu \tilde{\varphi} \der_\nu \tilde{\varphi} - \dfrac{1}{2} \, m_d^2(\sdil,\omega) \, \tilde{\varphi}^2 \biggr],
\end{align*}
with the field-dependent mass term
\begin{align*}
    m_d^2(\sdil,\omega) = \dfrac{\e^{\frac{\sdil}{2}} \, \xi^2}{\e^{2 \omega \, \frac{10-d}{d-2}}} \, \Bigl( \dfrac{g_s^2}{4 \pi} \Bigr)^{\frac{2}{d-2}} \, m_{\p,d}^2,
\end{align*}
and where the scalar field has been normalized as $\smash{\tilde{\varphi} = (g_s \kappa_{10} / \kappa_d) \, \varphi}$. Although a proper canonical normalization is more complicated, since one needs to absorb the $\smash{\e^{2 \sdil}}$-factor, this is enough here to check relative ratios. Such a canonical normalization however induces a non-trivial coupling with the dilaton field.

If one wants to keep the freedom to choose an arbitrary value for the constant $\vevc$, physical results must of course be independent of such a choice, even if the Einstein-frame volume is not fixed. In fact, it is apparent that the ratio of the $d$-dimensional mass with the squared $d$-dimensional Planck mass is independent of $\vevc$ and only depends on the dynamical fields $\sdil$ and $\omega$.

As far the parameter $g_s$ is concerned, if the 10-dimensional dilaton is not stabilized, one should define it as $g_s = 1$. This is a different prescription than the one for the parameter $\vevc$ because the definition of $g_s$ affects the scale of both non-compact and compact directions in the same way.
} A couple of further comments is in order. \begin{enumerate*}[label=(\roman*)]
    \item It is obvious that $g_s$ is just a label in this notation: of course the physically-meaningful string coupling, which controls the perturbative string-loop expansion, is the non-constant function $g_s(\Phi) = \mathrm{ln} \, \Phi$.
    \item One can still assume that the internal metrics are normalized such that $\smash{\breve{\mathrm{vol}} \, \mathrm{K}_{s_i} = l_s^{s_i}}$, since any constant shifts can be absorbed within the parameter $\vevc$.
\end{enumerate*}

\subsection{Scalar potentials for dilaton and volume moduli} \label{app: string-compactification scalar potentials}
If the internal spaces support non-trivial internal curvature, NSNS-fluxes, RR-fluxes, D-branes and O-planes, and/or a non-zero Casimir-energy, then a non-zero scalar potential is generated for the dilaton and the volume moduli. Here we review them one by one for type-II and heterotic strings.

\begin{itemize}
    \item A non-trivial curvature in the compactification space $\smash{\mathrm{K}_{s_i}}$ induces a scalar potential for the dilaton and radion that reads
    \begin{equation} \label{curvature-induced potential}
        V_{R_i} = \Lambda_{R_i} \, \e^{\frac{2}{\sqrt{d-2}} \, \kappa_d \tdelta - \frac{2}{\sqrt{s_i}} \, \kappa_d \tsigma_i} = \Lambda_{R_i} \, \e^{- \frac{2 (d-2+s_i)}{d-2} \, \omega_i - \frac{2 s_2}{d-2} \, \omega_{i_*}}.
    \end{equation}
   Here, the constant terms read
    \begin{equation} \label{Lambda_R}
        \Lambda_{R_i} = - \dfrac{1}{2 \kappa_d^2} \Bigl( \dfrac{g_s^2}{4 \pi} \Bigr)^{\frac{2}{d-2}} \dfrac{l_s^2 \breve{R}_{s_i}}{l_{\p, d}^2},
    \end{equation}
    where $\smash{\breve{R}_{s_i}}$ is the Ricci scalar associated with the $\smash{\breve{g}_{m_i n_i}}$-metric. As the internal curvature can be either positive or negative, so can be the sign of the scalar-potential term.
    \item Let a background NSNS-flux be quantized as $\smash{H_3 = h^{(3, i)} \, \breve{\alpha}_3^i l_s^2}$, where $\smash{h^{(3, i)} \in \mathbb{Z}}$ is an arbitrary integer and $\smash{\breve{\alpha}_3^i}$ is a harmonic 3-form in the internal space $\smash{\mathrm{K}_{s_i}}$,\footnote{Here, the harmonic 3-form is an element of an orthonormal set $\smash{\lbrace (\breve{\alpha}^{i, r})_3, (\breve{\beta}^i_s)_{s_2-3} \rbrace_{r,s}}$ of harmonic forms with respect to the $\smash{\breve{g}_{m_i n_i}}$-metric such that $\smash{\int_{\mathrm{K}_{s_i}} (\breve{\alpha}^{i, r})_3 \wedge (\breve{\beta}^i_s)_{s_i-3} = \delta^r_s}$, where $\smash{(\breve{\beta}^i_r)_{s_i-3}}$ is the harmonic $(s_i-3)$-form associated to the 3-cycle $A^{i,r}$ that is Poincaré-dual to the 3-form $(\breve{\alpha}^{i,r})_3$, and vice-versa. In this way, one finds the quantization condition $\smash{\int_{A} H_3 = \int_{\mathrm{K}_{s_i}} H_3 \wedge \breve{\beta}_{s_i-3} = h^{(3,i)} \, l_s^2}$.} with $s_i \geq 3$. This induces a scalar potential\footnote{Here and below, we do not canonically normalize the Einstein-frame radions because they are kinetically mixed. One should first rotate their basis into a diagonal one and then canonically normalize the new fields.}
    \begin{equation} \label{NSNS-flux-induced potential}
        V_{H_{3,i}} = \Lambda_{H_{3,i}} \, \e^{\frac{2}{\sqrt{d-2}} \, \kappa_d \tdelta - \frac{6}{\sqrt{s_i}} \, \kappa_d \tsigma_i} =  \Lambda_{H_{3,i}} \, \e^{- \sdil - \frac{2 (3d-6+s_i)}{d-2} \, \omega_i - \frac{2 s_2}{d-2} \, \omega_{i_*}},
    \end{equation}
    with $i_*$ denoting the subspace other than $\smash{\mathrm{K}_{s_i}}$, where the positive-definite prefactor is
    \begin{equation} \label{Lambda_H}
        \Lambda_{H_{3,i}} = \dfrac{1}{2} \, \dfrac{2 \pi}{g_s^2} \Bigl( \dfrac{g_s^2}{4 \pi \kappa_d^2} \Bigr)^{\frac{d}{d-2}} \, \dfrac{1}{3!} \, \bigl[ h^{(3, i)} \bigr]^2 \, \breve{\alpha}_{m_i n_i p_i} \breve{\alpha}_{m_i' n_i' p_i'} l_s^6 \, \breve{g}^{m_i m_i'} \breve{g}^{n_i n_i'} \breve{g}^{p_i p_i'}.
    \end{equation}
    \item For heterotic strings, let the background YM-flux be quantized as $\smash{F_{\mathrm{YM}} = f_{\mathrm{YM}}^{(i)} \, \breve{\alpha}_2^i \, t}$, where $\smash{f_{\mathrm{YM}}^{(i)} \in \mathbb{Z}}$ is an arbitrary integer, $\smash{\breve{\alpha}_2^i}$ is a harmonic 2-form in the internal space $\smash{\mathrm{K}_{s_i}}$, with $s_i \geq 2$, and $t$ is a gauge group generator. Then, the induced scalar potential is
    \begin{equation} \label{YM-flux-induced potential}
        V_{F_{\mathrm{YM},i}} = \Lambda_{F_{\mathrm{YM},i}} \, \e^{\frac{2}{\sqrt{d-2}} \, \kappa_d \tdelta - \frac{4}{\sqrt{s_i}} \, \kappa_d \tsigma_i} = \Lambda_{F_{\mathrm{YM},i}} \, \e^{- \frac{1}{2} \sdil - \frac{2 (2d-4+s_i)}{d-2} \, \omega_i - \frac{2 s_2}{d-2} \, \omega_{i_*}},
    \end{equation}
    where the positive-definite prefactor is
    \begin{equation} \label{Lambda_FYM}
        \Lambda_{F_{\mathrm{YM},i}} = \dfrac{1}{8 \pi g_s^2} \Bigl( \dfrac{g_s^2}{4 \pi \kappa_d^2} \Bigr)^{\frac{d}{d-2}} \, \dfrac{1}{2!} \, \bigl[ f_{\mathrm{YM}}^{(i)} \bigr]^2 \, \breve{\alpha}_{m_2 n_2} \breve{\alpha}_{m_2' n_2'} l_s^4 \, \breve{g}^{m_2 m_2'} \breve{g}^{n_2 n_2'} \, \mathrm{tr} \, t^2.
    \end{equation}
    \item Given the $(q-1)$-RR-form $C_{q-1}$, let the associated background RR-flux be quantized as $\smash{F_q = f^{(q, i)} \breve{\alpha}_q^i l_s^{q-1}}$, where $\smash{f^{(q, i)} \in \mathbb{Z}}$ is an arbitrary integer and $\smash{\breve{\alpha}_q^i}$ is a harmonic $q$-form in the internal space $\smash{\mathrm{K}_{s_i}}$, which means that there exists a $q$-cycle $A^{i, q}$ such that $\smash{\int_{A^{q, i}} F_q = f^{(q, i)} l_s^{q-1}}$, with $q \leq s_i$. Then, the induced scalar potential is
    \begin{equation} \label{RR-flux-induced potential}
        V_{F_q, i} = \Lambda_{F_q, i} \, \e^{\frac{d}{\sqrt{d-2}} \, \kappa_d \tdelta + \sqrt{s_{i_*}} \, \kappa_d \tsigma_{i_*} + \frac{s_i - 2q}{\sqrt{s_i}} \, \kappa_d \tsigma_i} = \Lambda_{F_q, i} \, \e^{\frac{q-5}{2} \, \sdil - \frac{2 s_{i_*}}{d-2} \, \omega_{i_*} - (2 q + \frac{2 s_i}{d-2}) \, \omega_i},
    \end{equation}
    where the positive-definite constant term is
    \begin{equation} \label{Lambda_F}
        \begin{split}
            \Lambda_{F_q, i} = \dfrac{1}{2} \, 2 \pi \, & \Bigl( \dfrac{g_s^2}{4 \pi \kappa_d^2} \Bigr)^{\frac{d}{d-2}} \, \dfrac{1}{q!} \, \bigl[f^{(q,i)}\bigr]^2 \breve{\alpha}_{m_{i,1} m_{i,2} \dots m_{i,q}} \breve{\alpha}_{n_{i,1} n_{i,2} \dots n_{i,q}} l_s^{2q} \, \breve{g}^{m_{i,1} n_{i,1}} \breve{g}^{m_{i,2} n_{i,2}} \dots \breve{g}^{m_{i,q} n_{i,q}}.
        \end{split}
    \end{equation}
    \item Let a D$p$-brane or an O$p$-plane wrap the non-compact $d$-spacetime and a $(p+1-d)$-dimensional cycle within the internal space $\smash{\mathrm{K}_{s_i}}$, with $p+1-d \leq s_i$. Then, the scalar potential induced for the dilaton and the radions reads
    \begin{equation} \label{Dp-brane/Op-plane-induced potential}
        \begin{split}
            V_{\mathrm{D}p/\mathrm{O}p, i} & = \Lambda_{\mathrm{D}p/\mathrm{O}p} \, \e^{\frac{d+2}{2 \sqrt{d-2}} \, \kappa_d \tdelta - \frac{\sqrt{s_{i_*}}}{2} \kappa_d \tsigma_{i_*} - \frac{s_i - 2 (p+1-d)}{2 \sqrt{s_i}} \kappa_d \tsigma_i} \\
            & = \Lambda_{\mathrm{D}p/\mathrm{O}p} \, \e^{\frac{p-3}{4} \, \sdil - \frac{s_{i_*} d}{d-2} \, \omega_{i_*} + \bigl[ \frac{s_i d}{d-2} - (p + 1 - d) \bigr] \, \omega_i}.
        \end{split}
    \end{equation}
    Here, the constant term is
    \begin{equation} \label{Lambda_Dp/Op}
        \Lambda_{\mathrm{D}p/\mathrm{O}p, i} = \dfrac{2 \pi}{g_s} \, \Bigl( \dfrac{g_s^2}{4 \pi \kappa_d^2} \Bigr)^{\frac{d}{d-2}} \, m(\epsilon, p).
    \end{equation}
    Here, the additional factor is $m(\epsilon, p) = 1$ for D$p$-branes, whereas for O$p$-planes it takes the form $m(\epsilon, p) = - \epsilon \, 2^{p-4}$, with $\epsilon=\pm1$ representing negative- and positive-tension O-planes, respectively. Whilst for D$p$-branes the scalar potential term is positive-definite, for O$p$-planes it can in principle of both signs.
\end{itemize}

In principle, one can add other terms in the scalar potential, sourced for instance by non-geometric fluxes or NS5-branes: all such terms are of exponential form but they have quite a different coefficient structure.

A general contribution that appears in non-supersymmetric theories is the quantum-generated vacuum energy, which is proportional to the string-theoretic $l$-loop partition function. Below, we consider an isotropic compactification for simplicity. For a number $l$ of string loops, the generic string-frame quantum effective action can be written as
\begin{align*}
    S_{C_l} = - \dfrac{1}{2 \kappa_{10}^2} \int_{\mathrm{X}_{1,9}} \de^{10} x \, \sqrt{-G_{10}} \; \dfrac{1}{\e^{(10-d) \sigma}} \, \biggl[ \dfrac{1}{l_s^2} \, g_s^{2(l-1)} \, \e^{2(l-1) \sdil} \, \e^{n \sigma} I_l \biggr].
\end{align*}
Here, $\Lambda_l = \bigl[ I_l / (2 \kappa_{10}^2 l_s^2) \bigr] \, \e^{n \sigma}$ is the $l$-loop cosmological constant, assumed to have a generic power-law string-frame radius $\rho = \e^\sigma$ dependence at large volume, and the extra $\e^{(10-d) \sigma}$-factor is a normalization due to the fact that the compactification volume is already accounted for by the KK- and winding-states in the partition function \cite{Ginsparg:1986wr}. In the end, the generic $d$-dimensional Einstein-frame closed-string quantum scalar potential takes the form
\begin{equation} \label{closed-string quantum potential}
    V_{C_l} = \Lambda_{C_l} \, \e^{[\frac{2}{\sqrt{d-2}} + l \sqrt{d-2}] \, \kappa_d \tdelta + [ (l-1) \sqrt{10-d} + \frac{n}{\sqrt{10-d}}] \, \kappa_d \tsigma},
\end{equation}
where the constant term is
\begin{equation} \label{Lambda_C_closed}
    \Lambda_{C_l} = 2 \pi I_l \, g_s^{2(l-1)} \, \Bigl( \dfrac{g_s^2}{4 \pi \kappa_d^2} \Bigr)^{\frac{d}{d-2}}.
\end{equation}
This term has a sign fixed by the constant $I_l$: at $l=1$, its sign is determined as by the difference of massless bosons and fermions as $\mathrm{sgn} \, I_1 = - \mathrm{sgn} \, (n_0^b - n_0^f)$. For an open-string sector associated to a D$p$-brane, one can write
\begin{align*}
    S_{C_{l, p}} = - T_{\mathrm{D}p} \int_{\mathrm{W}_{1,p}} \de^{p+1} \xi \, \sqrt{-(\varphi_* G)_{p+1}} \; \dfrac{1}{\e^{(p+1-d) \sigma}} \, \biggl[ g_s^{l-1} \, \e^{(l-1) \sdil} \, \e^{n \sigma} I_{l,p} \biggr].
\end{align*}
In the static gauge and neglecting open-string scalars, one thus finds a potential
\begin{equation} \label{open-string quantum potential}
    V_{C_{l,p}} = \Lambda_{C_{l,p}} \, \e^{[\frac{d+2}{2 \sqrt{d-2}} + l \frac{\sqrt{d-2}}{2}] \, \kappa_d \tdelta + [ \frac{2n - (10-d)}{2 \sqrt{10-d}} + l \frac{\sqrt{10-d}}{2}] \, \kappa_d \tsigma},
\end{equation}
with the definition
\begin{equation} \label{Lambda_C_open}
    \Lambda_{C_{l,p}} = 2 \pi I_{l,p} \, g_s^{l-1} \, \Bigl( \dfrac{g_s^2}{4 \pi \kappa_d^2} \Bigr)^{\frac{d}{d-2}}.
\end{equation}

\newpage
\begin{acknowledgments}
\subsection*{Acknowledgments}
GS and FT would like to thank B.V. Bento, M. Cicoli, S. Cremonini, M. Rajaguru, C. Vafa, V. Van Hemelryck and T. Wrase for compelling discussions and valuable comments. GS and FT are supported in part by the DOE grant DE-SC0017647. HVT is supported in part by the NSF CAREER grant DMS-1843320 and a Vilas Faculty Early-Career Investigator Award.
\end{acknowledgments}

\bibliographystyle{apsrev4-1}
\bibliography{report.bib}

\begin{thebibliography}{65}%
\makeatletter
\providecommand \@ifxundefined [1]{%
 \@ifx{#1\undefined}
}%
\providecommand \@ifnum [1]{%
 \ifnum #1\expandafter \@firstoftwo
 \else \expandafter \@secondoftwo
 \fi
}%
\providecommand \@ifx [1]{%
 \ifx #1\expandafter \@firstoftwo
 \else \expandafter \@secondoftwo
 \fi
}%
\providecommand \natexlab [1]{#1}%
\providecommand \enquote  [1]{``#1''}%
\providecommand \bibnamefont  [1]{#1}%
\providecommand \bibfnamefont [1]{#1}%
\providecommand \citenamefont [1]{#1}%
\providecommand \href@noop [0]{\@secondoftwo}%
\providecommand \href [0]{\begingroup \@sanitize@url \@href}%
\providecommand \@href[1]{\@@startlink{#1}\@@href}%
\providecommand \@@href[1]{\endgroup#1\@@endlink}%
\providecommand \@sanitize@url [0]{\catcode `\\12\catcode `\$12\catcode `\&12\catcode `\#12\catcode `\^12\catcode `\_12\catcode `\%12\relax}%
\providecommand \@@startlink[1]{}%
\providecommand \@@endlink[0]{}%
\providecommand \url  [0]{\begingroup\@sanitize@url \@url }%
\providecommand \@url [1]{\endgroup\@href {#1}{\urlprefix }}%
\providecommand \urlprefix  [0]{URL }%
\providecommand \Eprint [0]{\href }%
\providecommand \doibase [0]{http://dx.doi.org/}%
\providecommand \selectlanguage [0]{\@gobble}%
\providecommand \bibinfo  [0]{\@secondoftwo}%
\providecommand \bibfield  [0]{\@secondoftwo}%
\providecommand \translation [1]{[#1]}%
\providecommand \BibitemOpen [0]{}%
\providecommand \bibitemStop [0]{}%
\providecommand \bibitemNoStop [0]{.\EOS\space}%
\providecommand \EOS [0]{\spacefactor3000\relax}%
\providecommand \BibitemShut  [1]{\csname bibitem#1\endcsname}%
\let\auto@bib@innerbib\@empty
\bibitem [{\citenamefont {Dine}\ and\ \citenamefont {Seiberg}(1985)}]{Dine:1985he}%
  \BibitemOpen
  \bibfield  {author} {\bibinfo {author} {\bibfnamefont {M.}~\bibnamefont {Dine}}\ and\ \bibinfo {author} {\bibfnamefont {N.}~\bibnamefont {Seiberg}},\ }\href {\doibase 10.1016/0370-2693(85)90927-X} {\bibfield  {journal} {\bibinfo  {journal} {Phys. Lett. B}\ }\textbf {\bibinfo {volume} {162}},\ \bibinfo {pages} {299} (\bibinfo {year} {1985})}\BibitemShut {NoStop}%
\bibitem [{\citenamefont {Flauger}\ \emph {et~al.}(2022)\citenamefont {Flauger}, \citenamefont {Gorbenko}, \citenamefont {Joyce}, \citenamefont {McAllister}, \citenamefont {Shiu},\ and\ \citenamefont {Silverstein}}]{Flauger:2022hie}%
  \BibitemOpen
  \bibfield  {author} {\bibinfo {author} {\bibfnamefont {R.}~\bibnamefont {Flauger}}, \bibinfo {author} {\bibfnamefont {V.}~\bibnamefont {Gorbenko}}, \bibinfo {author} {\bibfnamefont {A.}~\bibnamefont {Joyce}}, \bibinfo {author} {\bibfnamefont {L.}~\bibnamefont {McAllister}}, \bibinfo {author} {\bibfnamefont {G.}~\bibnamefont {Shiu}}, \ and\ \bibinfo {author} {\bibfnamefont {E.}~\bibnamefont {Silverstein}},\ }in\ \href@noop {} {\emph {\bibinfo {booktitle} {{2022 Snowmass Summer Study}}}}\ (\bibinfo {year} {2022})\ \Eprint {http://arxiv.org/abs/2203.07629} {arXiv:2203.07629 [hep-th]} \BibitemShut {NoStop}%
\bibitem [{\citenamefont {Danielsson}\ and\ \citenamefont {Van~Riet}(2018)}]{Danielsson:2018ztv}%
  \BibitemOpen
  \bibfield  {author} {\bibinfo {author} {\bibfnamefont {U.~H.}\ \bibnamefont {Danielsson}}\ and\ \bibinfo {author} {\bibfnamefont {T.}~\bibnamefont {Van~Riet}},\ }\href {\doibase 10.1142/S0218271818300070} {\bibfield  {journal} {\bibinfo  {journal} {Int. J. Mod. Phys. D}\ }\textbf {\bibinfo {volume} {27}},\ \bibinfo {pages} {1830007} (\bibinfo {year} {2018})},\ \Eprint {http://arxiv.org/abs/1804.01120} {arXiv:1804.01120 [hep-th]} \BibitemShut {NoStop}%
\bibitem [{\citenamefont {Obied}\ \emph {et~al.}(2018)\citenamefont {Obied}, \citenamefont {Ooguri}, \citenamefont {Spodyneiko},\ and\ \citenamefont {Vafa}}]{Obied:2018sgi}%
  \BibitemOpen
  \bibfield  {author} {\bibinfo {author} {\bibfnamefont {G.}~\bibnamefont {Obied}}, \bibinfo {author} {\bibfnamefont {H.}~\bibnamefont {Ooguri}}, \bibinfo {author} {\bibfnamefont {L.}~\bibnamefont {Spodyneiko}}, \ and\ \bibinfo {author} {\bibfnamefont {C.}~\bibnamefont {Vafa}},\ }\href@noop {} {\  (\bibinfo {year} {2018})},\ \Eprint {http://arxiv.org/abs/1806.08362} {arXiv:1806.08362 [hep-th]} \BibitemShut {NoStop}%
\bibitem [{\citenamefont {Ooguri}\ \emph {et~al.}(2019)\citenamefont {Ooguri}, \citenamefont {Palti}, \citenamefont {Shiu},\ and\ \citenamefont {Vafa}}]{Ooguri:2018wrx}%
  \BibitemOpen
  \bibfield  {author} {\bibinfo {author} {\bibfnamefont {H.}~\bibnamefont {Ooguri}}, \bibinfo {author} {\bibfnamefont {E.}~\bibnamefont {Palti}}, \bibinfo {author} {\bibfnamefont {G.}~\bibnamefont {Shiu}}, \ and\ \bibinfo {author} {\bibfnamefont {C.}~\bibnamefont {Vafa}},\ }\href {\doibase 10.1016/j.physletb.2018.11.018} {\bibfield  {journal} {\bibinfo  {journal} {Phys. Lett. B}\ }\textbf {\bibinfo {volume} {788}},\ \bibinfo {pages} {180} (\bibinfo {year} {2019})},\ \Eprint {http://arxiv.org/abs/1810.05506} {arXiv:1810.05506 [hep-th]} \BibitemShut {NoStop}%
\bibitem [{\citenamefont {Bedroya}\ and\ \citenamefont {Vafa}(2020)}]{Bedroya:2019snp}%
  \BibitemOpen
  \bibfield  {author} {\bibinfo {author} {\bibfnamefont {A.}~\bibnamefont {Bedroya}}\ and\ \bibinfo {author} {\bibfnamefont {C.}~\bibnamefont {Vafa}},\ }\href {\doibase 10.1007/JHEP09(2020)123} {\bibfield  {journal} {\bibinfo  {journal} {JHEP}\ }\textbf {\bibinfo {volume} {09}},\ \bibinfo {pages} {123} (\bibinfo {year} {2020})},\ \Eprint {http://arxiv.org/abs/1909.11063} {arXiv:1909.11063 [hep-th]} \BibitemShut {NoStop}%
\bibitem [{\citenamefont {Rudelius}(2021)}]{Rudelius:2021azq}%
  \BibitemOpen
  \bibfield  {author} {\bibinfo {author} {\bibfnamefont {T.}~\bibnamefont {Rudelius}},\ }\href {\doibase 10.1103/PhysRevD.104.126023} {\bibfield  {journal} {\bibinfo  {journal} {Phys. Rev. D}\ }\textbf {\bibinfo {volume} {104}},\ \bibinfo {pages} {126023} (\bibinfo {year} {2021})},\ \Eprint {http://arxiv.org/abs/2106.09026} {arXiv:2106.09026 [hep-th]} \BibitemShut {NoStop}%
\bibitem [{\citenamefont {Andriot}(2018)}]{Andriot:2018wzk}%
  \BibitemOpen
  \bibfield  {author} {\bibinfo {author} {\bibfnamefont {D.}~\bibnamefont {Andriot}},\ }\href {\doibase 10.1016/j.physletb.2018.09.022} {\bibfield  {journal} {\bibinfo  {journal} {Phys. Lett. B}\ }\textbf {\bibinfo {volume} {785}},\ \bibinfo {pages} {570} (\bibinfo {year} {2018})},\ \Eprint {http://arxiv.org/abs/1806.10999} {arXiv:1806.10999 [hep-th]} \BibitemShut {NoStop}%
\bibitem [{\citenamefont {Garg}\ and\ \citenamefont {Krishnan}(2019)}]{Garg:2018reu}%
  \BibitemOpen
  \bibfield  {author} {\bibinfo {author} {\bibfnamefont {S.~K.}\ \bibnamefont {Garg}}\ and\ \bibinfo {author} {\bibfnamefont {C.}~\bibnamefont {Krishnan}},\ }\href {\doibase 10.1007/JHEP11(2019)075} {\bibfield  {journal} {\bibinfo  {journal} {JHEP}\ }\textbf {\bibinfo {volume} {11}},\ \bibinfo {pages} {075} (\bibinfo {year} {2019})},\ \Eprint {http://arxiv.org/abs/1807.05193} {arXiv:1807.05193 [hep-th]} \BibitemShut {NoStop}%
\bibitem [{\citenamefont {Dirac}(1974)}]{dirac1974cosmological}%
  \BibitemOpen
  \bibfield  {author} {\bibinfo {author} {\bibfnamefont {P.~A.~M.}\ \bibnamefont {Dirac}},\ }\href@noop {} {\bibfield  {journal} {\bibinfo  {journal} {Proceedings of the Royal Society of London. A. Mathematical and Physical Sciences}\ }\textbf {\bibinfo {volume} {338}},\ \bibinfo {pages} {439} (\bibinfo {year} {1974})}\BibitemShut {NoStop}%
\bibitem [{\citenamefont {Montero}\ \emph {et~al.}(2023)\citenamefont {Montero}, \citenamefont {Vafa},\ and\ \citenamefont {Valenzuela}}]{Montero:2022prj}%
  \BibitemOpen
  \bibfield  {author} {\bibinfo {author} {\bibfnamefont {M.}~\bibnamefont {Montero}}, \bibinfo {author} {\bibfnamefont {C.}~\bibnamefont {Vafa}}, \ and\ \bibinfo {author} {\bibfnamefont {I.}~\bibnamefont {Valenzuela}},\ }\href {\doibase 10.1007/JHEP02(2023)022} {\bibfield  {journal} {\bibinfo  {journal} {JHEP}\ }\textbf {\bibinfo {volume} {02}},\ \bibinfo {pages} {022} (\bibinfo {year} {2023})},\ \Eprint {http://arxiv.org/abs/2205.12293} {arXiv:2205.12293 [hep-th]} \BibitemShut {NoStop}%
\bibitem [{\citenamefont {Conlon}\ and\ \citenamefont {Revello}(2022)}]{Conlon:2022pnx}%
  \BibitemOpen
  \bibfield  {author} {\bibinfo {author} {\bibfnamefont {J.~P.}\ \bibnamefont {Conlon}}\ and\ \bibinfo {author} {\bibfnamefont {F.}~\bibnamefont {Revello}},\ }\href {\doibase 10.1007/JHEP11(2022)155} {\bibfield  {journal} {\bibinfo  {journal} {JHEP}\ }\textbf {\bibinfo {volume} {11}},\ \bibinfo {pages} {155} (\bibinfo {year} {2022})},\ \Eprint {http://arxiv.org/abs/2207.00567} {arXiv:2207.00567 [hep-th]} \BibitemShut {NoStop}%
\bibitem [{\citenamefont {Rudelius}(2022)}]{Rudelius:2022gbz}%
  \BibitemOpen
  \bibfield  {author} {\bibinfo {author} {\bibfnamefont {T.}~\bibnamefont {Rudelius}},\ }\href {\doibase 10.1007/JHEP10(2022)018} {\bibfield  {journal} {\bibinfo  {journal} {JHEP}\ }\textbf {\bibinfo {volume} {10}},\ \bibinfo {pages} {018} (\bibinfo {year} {2022})},\ \Eprint {http://arxiv.org/abs/2208.08989} {arXiv:2208.08989 [hep-th]} \BibitemShut {NoStop}%
\bibitem [{\citenamefont {Calder\'on-Infante}\ \emph {et~al.}(2022)\citenamefont {Calder\'on-Infante}, \citenamefont {Ruiz},\ and\ \citenamefont {Valenzuela}}]{Calderon-Infante:2022nxb}%
  \BibitemOpen
  \bibfield  {author} {\bibinfo {author} {\bibfnamefont {J.}~\bibnamefont {Calder\'on-Infante}}, \bibinfo {author} {\bibfnamefont {I.}~\bibnamefont {Ruiz}}, \ and\ \bibinfo {author} {\bibfnamefont {I.}~\bibnamefont {Valenzuela}},\ }\href@noop {} {\  (\bibinfo {year} {2022})},\ \Eprint {http://arxiv.org/abs/2209.11821} {arXiv:2209.11821 [hep-th]} \BibitemShut {NoStop}%
\bibitem [{\citenamefont {Marconnet}\ and\ \citenamefont {Tsimpis}(2023)}]{Marconnet:2022fmx}%
  \BibitemOpen
  \bibfield  {author} {\bibinfo {author} {\bibfnamefont {P.}~\bibnamefont {Marconnet}}\ and\ \bibinfo {author} {\bibfnamefont {D.}~\bibnamefont {Tsimpis}},\ }\href {\doibase 10.1007/JHEP01(2023)033} {\bibfield  {journal} {\bibinfo  {journal} {JHEP}\ }\textbf {\bibinfo {volume} {01}},\ \bibinfo {pages} {033} (\bibinfo {year} {2023})},\ \Eprint {http://arxiv.org/abs/2210.10813} {arXiv:2210.10813 [hep-th]} \BibitemShut {NoStop}%
\bibitem [{\citenamefont {Bedroya}(2022)}]{Bedroya:2022tbh}%
  \BibitemOpen
  \bibfield  {author} {\bibinfo {author} {\bibfnamefont {A.}~\bibnamefont {Bedroya}},\ }\href@noop {} {\  (\bibinfo {year} {2022})},\ \Eprint {http://arxiv.org/abs/2211.09128} {arXiv:2211.09128 [hep-th]} \BibitemShut {NoStop}%
\bibitem [{\citenamefont {Shiu}\ \emph {et~al.}(2023)\citenamefont {Shiu}, \citenamefont {Tonioni},\ and\ \citenamefont {Tran}}]{Shiu:2023nph}%
  \BibitemOpen
  \bibfield  {author} {\bibinfo {author} {\bibfnamefont {G.}~\bibnamefont {Shiu}}, \bibinfo {author} {\bibfnamefont {F.}~\bibnamefont {Tonioni}}, \ and\ \bibinfo {author} {\bibfnamefont {H.~V.}\ \bibnamefont {Tran}},\ }\href@noop {} {\  (\bibinfo {year} {2023})},\ \Eprint {http://arxiv.org/abs/2303.03418} {arXiv:2303.03418 [hep-th]} \BibitemShut {NoStop}%
\bibitem [{\citenamefont {Collinucci}\ \emph {et~al.}(2005)\citenamefont {Collinucci}, \citenamefont {Nielsen},\ and\ \citenamefont {Van~Riet}}]{Collinucci:2004iw}%
  \BibitemOpen
  \bibfield  {author} {\bibinfo {author} {\bibfnamefont {A.}~\bibnamefont {Collinucci}}, \bibinfo {author} {\bibfnamefont {M.}~\bibnamefont {Nielsen}}, \ and\ \bibinfo {author} {\bibfnamefont {T.}~\bibnamefont {Van~Riet}},\ }\href {\doibase 10.1088/0264-9381/22/7/005} {\bibfield  {journal} {\bibinfo  {journal} {Class. Quant. Grav.}\ }\textbf {\bibinfo {volume} {22}},\ \bibinfo {pages} {1269} (\bibinfo {year} {2005})},\ \Eprint {http://arxiv.org/abs/hep-th/0407047} {arXiv:hep-th/0407047} \BibitemShut {NoStop}%
\bibitem [{\citenamefont {Halliwell}(1987)}]{Halliwell:1986ja}%
  \BibitemOpen
  \bibfield  {author} {\bibinfo {author} {\bibfnamefont {J.~J.}\ \bibnamefont {Halliwell}},\ }\href {\doibase 10.1016/0370-2693(87)91011-2} {\bibfield  {journal} {\bibinfo  {journal} {Phys. Lett. B}\ }\textbf {\bibinfo {volume} {185}},\ \bibinfo {pages} {341} (\bibinfo {year} {1987})}\BibitemShut {NoStop}%
\bibitem [{\citenamefont {Copeland}\ \emph {et~al.}(1998)\citenamefont {Copeland}, \citenamefont {Liddle},\ and\ \citenamefont {Wands}}]{Copeland:1997et}%
  \BibitemOpen
  \bibfield  {author} {\bibinfo {author} {\bibfnamefont {E.~J.}\ \bibnamefont {Copeland}}, \bibinfo {author} {\bibfnamefont {A.~R.}\ \bibnamefont {Liddle}}, \ and\ \bibinfo {author} {\bibfnamefont {D.}~\bibnamefont {Wands}},\ }\href {\doibase 10.1103/PhysRevD.57.4686} {\bibfield  {journal} {\bibinfo  {journal} {Phys. Rev. D}\ }\textbf {\bibinfo {volume} {57}},\ \bibinfo {pages} {4686} (\bibinfo {year} {1998})},\ \Eprint {http://arxiv.org/abs/gr-qc/9711068} {arXiv:gr-qc/9711068} \BibitemShut {NoStop}%
\bibitem [{\citenamefont {Malik}\ and\ \citenamefont {Wands}(1999)}]{Malik:1998gy}%
  \BibitemOpen
  \bibfield  {author} {\bibinfo {author} {\bibfnamefont {K.~A.}\ \bibnamefont {Malik}}\ and\ \bibinfo {author} {\bibfnamefont {D.}~\bibnamefont {Wands}},\ }\href {\doibase 10.1103/PhysRevD.59.123501} {\bibfield  {journal} {\bibinfo  {journal} {Phys. Rev. D}\ }\textbf {\bibinfo {volume} {59}},\ \bibinfo {pages} {123501} (\bibinfo {year} {1999})},\ \Eprint {http://arxiv.org/abs/astro-ph/9812204} {arXiv:astro-ph/9812204} \BibitemShut {NoStop}%
\bibitem [{\citenamefont {Coley}\ and\ \citenamefont {van~den Hoogen}(2000)}]{Coley:1999mj}%
  \BibitemOpen
  \bibfield  {author} {\bibinfo {author} {\bibfnamefont {A.~A.}\ \bibnamefont {Coley}}\ and\ \bibinfo {author} {\bibfnamefont {R.~J.}\ \bibnamefont {van~den Hoogen}},\ }\href {\doibase 10.1103/PhysRevD.62.023517} {\bibfield  {journal} {\bibinfo  {journal} {Phys. Rev. D}\ }\textbf {\bibinfo {volume} {62}},\ \bibinfo {pages} {023517} (\bibinfo {year} {2000})},\ \Eprint {http://arxiv.org/abs/gr-qc/9911075} {arXiv:gr-qc/9911075} \BibitemShut {NoStop}%
\bibitem [{\citenamefont {Guo}\ \emph {et~al.}(2003{\natexlab{a}})\citenamefont {Guo}, \citenamefont {Piao},\ and\ \citenamefont {Zhang}}]{Guo:2003eu}%
  \BibitemOpen
  \bibfield  {author} {\bibinfo {author} {\bibfnamefont {Z.~K.}\ \bibnamefont {Guo}}, \bibinfo {author} {\bibfnamefont {Y.-S.}\ \bibnamefont {Piao}}, \ and\ \bibinfo {author} {\bibfnamefont {Y.-Z.}\ \bibnamefont {Zhang}},\ }\href {\doibase 10.1016/j.physletb.2003.06.004} {\bibfield  {journal} {\bibinfo  {journal} {Phys. Lett. B}\ }\textbf {\bibinfo {volume} {568}},\ \bibinfo {pages} {1} (\bibinfo {year} {2003}{\natexlab{a}})},\ \Eprint {http://arxiv.org/abs/hep-th/0304048} {arXiv:hep-th/0304048} \BibitemShut {NoStop}%
\bibitem [{\citenamefont {Guo}\ \emph {et~al.}(2003{\natexlab{b}})\citenamefont {Guo}, \citenamefont {Piao}, \citenamefont {Cai},\ and\ \citenamefont {Zhang}}]{Guo:2003rs}%
  \BibitemOpen
  \bibfield  {author} {\bibinfo {author} {\bibfnamefont {Z.-K.}\ \bibnamefont {Guo}}, \bibinfo {author} {\bibfnamefont {Y.-S.}\ \bibnamefont {Piao}}, \bibinfo {author} {\bibfnamefont {R.-G.}\ \bibnamefont {Cai}}, \ and\ \bibinfo {author} {\bibfnamefont {Y.-Z.}\ \bibnamefont {Zhang}},\ }\href {\doibase 10.1016/j.physletb.2003.09.074} {\bibfield  {journal} {\bibinfo  {journal} {Phys. Lett. B}\ }\textbf {\bibinfo {volume} {576}},\ \bibinfo {pages} {12} (\bibinfo {year} {2003}{\natexlab{b}})},\ \Eprint {http://arxiv.org/abs/hep-th/0306245} {arXiv:hep-th/0306245} \BibitemShut {NoStop}%
\bibitem [{\citenamefont {Bergshoeff}\ \emph {et~al.}(2004)\citenamefont {Bergshoeff}, \citenamefont {Collinucci}, \citenamefont {Gran}, \citenamefont {Nielsen},\ and\ \citenamefont {Roest}}]{Bergshoeff:2003vb}%
  \BibitemOpen
  \bibfield  {author} {\bibinfo {author} {\bibfnamefont {E.}~\bibnamefont {Bergshoeff}}, \bibinfo {author} {\bibfnamefont {A.}~\bibnamefont {Collinucci}}, \bibinfo {author} {\bibfnamefont {U.}~\bibnamefont {Gran}}, \bibinfo {author} {\bibfnamefont {M.}~\bibnamefont {Nielsen}}, \ and\ \bibinfo {author} {\bibfnamefont {D.}~\bibnamefont {Roest}},\ }\href {\doibase 10.1088/0264-9381/21/8/003} {\bibfield  {journal} {\bibinfo  {journal} {Class. Quant. Grav.}\ }\textbf {\bibinfo {volume} {21}},\ \bibinfo {pages} {1947} (\bibinfo {year} {2004})},\ \Eprint {http://arxiv.org/abs/hep-th/0312102} {arXiv:hep-th/0312102} \BibitemShut {NoStop}%
\bibitem [{\citenamefont {Kim}\ \emph {et~al.}(2005)\citenamefont {Kim}, \citenamefont {Liddle},\ and\ \citenamefont {Tsujikawa}}]{Kim:2005ne}%
  \BibitemOpen
  \bibfield  {author} {\bibinfo {author} {\bibfnamefont {S.~A.}\ \bibnamefont {Kim}}, \bibinfo {author} {\bibfnamefont {A.~R.}\ \bibnamefont {Liddle}}, \ and\ \bibinfo {author} {\bibfnamefont {S.}~\bibnamefont {Tsujikawa}},\ }\href {\doibase 10.1103/PhysRevD.72.043506} {\bibfield  {journal} {\bibinfo  {journal} {Phys. Rev. D}\ }\textbf {\bibinfo {volume} {72}},\ \bibinfo {pages} {043506} (\bibinfo {year} {2005})},\ \Eprint {http://arxiv.org/abs/astro-ph/0506076} {arXiv:astro-ph/0506076} \BibitemShut {NoStop}%
\bibitem [{\citenamefont {Hartong}\ \emph {et~al.}(2006)\citenamefont {Hartong}, \citenamefont {Ploegh}, \citenamefont {Van~Riet},\ and\ \citenamefont {Westra}}]{Hartong:2006rt}%
  \BibitemOpen
  \bibfield  {author} {\bibinfo {author} {\bibfnamefont {J.}~\bibnamefont {Hartong}}, \bibinfo {author} {\bibfnamefont {A.}~\bibnamefont {Ploegh}}, \bibinfo {author} {\bibfnamefont {T.}~\bibnamefont {Van~Riet}}, \ and\ \bibinfo {author} {\bibfnamefont {D.~B.}\ \bibnamefont {Westra}},\ }\href {\doibase 10.1088/0264-9381/23/14/003} {\bibfield  {journal} {\bibinfo  {journal} {Class. Quant. Grav.}\ }\textbf {\bibinfo {volume} {23}},\ \bibinfo {pages} {4593} (\bibinfo {year} {2006})},\ \Eprint {http://arxiv.org/abs/gr-qc/0602077} {arXiv:gr-qc/0602077} \BibitemShut {NoStop}%
\bibitem [{\citenamefont {Bahamonde}\ \emph {et~al.}(2018)\citenamefont {Bahamonde}, \citenamefont {B\"ohmer}, \citenamefont {Carloni}, \citenamefont {Copeland}, \citenamefont {Fang},\ and\ \citenamefont {Tamanini}}]{Bahamonde:2017ize}%
  \BibitemOpen
  \bibfield  {author} {\bibinfo {author} {\bibfnamefont {S.}~\bibnamefont {Bahamonde}}, \bibinfo {author} {\bibfnamefont {C.~G.}\ \bibnamefont {B\"ohmer}}, \bibinfo {author} {\bibfnamefont {S.}~\bibnamefont {Carloni}}, \bibinfo {author} {\bibfnamefont {E.~J.}\ \bibnamefont {Copeland}}, \bibinfo {author} {\bibfnamefont {W.}~\bibnamefont {Fang}}, \ and\ \bibinfo {author} {\bibfnamefont {N.}~\bibnamefont {Tamanini}},\ }\href {\doibase 10.1016/j.physrep.2018.09.001} {\bibfield  {journal} {\bibinfo  {journal} {Phys. Rept.}\ }\textbf {\bibinfo {volume} {775-777}},\ \bibinfo {pages} {1} (\bibinfo {year} {2018})},\ \Eprint {http://arxiv.org/abs/1712.03107} {arXiv:1712.03107 [gr-qc]} \BibitemShut {NoStop}%
\bibitem [{\citenamefont {Apers}\ \emph {et~al.}(2022)\citenamefont {Apers}, \citenamefont {Conlon}, \citenamefont {Mosny},\ and\ \citenamefont {Revello}}]{Apers:2022cyl}%
  \BibitemOpen
  \bibfield  {author} {\bibinfo {author} {\bibfnamefont {F.}~\bibnamefont {Apers}}, \bibinfo {author} {\bibfnamefont {J.~P.}\ \bibnamefont {Conlon}}, \bibinfo {author} {\bibfnamefont {M.}~\bibnamefont {Mosny}}, \ and\ \bibinfo {author} {\bibfnamefont {F.}~\bibnamefont {Revello}},\ }\href@noop {} {\  (\bibinfo {year} {2022})},\ \Eprint {http://arxiv.org/abs/2212.10293} {arXiv:2212.10293 [hep-th]} \BibitemShut {NoStop}%
\bibitem [{\citenamefont {Liddle}\ \emph {et~al.}(1998)\citenamefont {Liddle}, \citenamefont {Mazumdar},\ and\ \citenamefont {Schunck}}]{Liddle:1998jc}%
  \BibitemOpen
  \bibfield  {author} {\bibinfo {author} {\bibfnamefont {A.~R.}\ \bibnamefont {Liddle}}, \bibinfo {author} {\bibfnamefont {A.}~\bibnamefont {Mazumdar}}, \ and\ \bibinfo {author} {\bibfnamefont {F.~E.}\ \bibnamefont {Schunck}},\ }\href {\doibase 10.1103/PhysRevD.58.061301} {\bibfield  {journal} {\bibinfo  {journal} {Phys. Rev. D}\ }\textbf {\bibinfo {volume} {58}},\ \bibinfo {pages} {061301} (\bibinfo {year} {1998})},\ \Eprint {http://arxiv.org/abs/astro-ph/9804177} {arXiv:astro-ph/9804177} \BibitemShut {NoStop}%
\bibitem [{\citenamefont {Copeland}\ \emph {et~al.}(1999)\citenamefont {Copeland}, \citenamefont {Mazumdar},\ and\ \citenamefont {Nunes}}]{Copeland:1999cs}%
  \BibitemOpen
  \bibfield  {author} {\bibinfo {author} {\bibfnamefont {E.~J.}\ \bibnamefont {Copeland}}, \bibinfo {author} {\bibfnamefont {A.}~\bibnamefont {Mazumdar}}, \ and\ \bibinfo {author} {\bibfnamefont {N.~J.}\ \bibnamefont {Nunes}},\ }\href {\doibase 10.1103/PhysRevD.60.083506} {\bibfield  {journal} {\bibinfo  {journal} {Phys. Rev. D}\ }\textbf {\bibinfo {volume} {60}},\ \bibinfo {pages} {083506} (\bibinfo {year} {1999})},\ \Eprint {http://arxiv.org/abs/astro-ph/9904309} {arXiv:astro-ph/9904309} \BibitemShut {NoStop}%
\bibitem [{\citenamefont {Lyth}\ and\ \citenamefont {Riotto}(1999)}]{Lyth:1998xn}%
  \BibitemOpen
  \bibfield  {author} {\bibinfo {author} {\bibfnamefont {D.~H.}\ \bibnamefont {Lyth}}\ and\ \bibinfo {author} {\bibfnamefont {A.}~\bibnamefont {Riotto}},\ }\href {\doibase 10.1016/S0370-1573(98)00128-8} {\bibfield  {journal} {\bibinfo  {journal} {Phys. Rept.}\ }\textbf {\bibinfo {volume} {314}},\ \bibinfo {pages} {1} (\bibinfo {year} {1999})},\ \Eprint {http://arxiv.org/abs/hep-ph/9807278} {arXiv:hep-ph/9807278} \BibitemShut {NoStop}%
\bibitem [{\citenamefont {Ach\'ucarro}\ and\ \citenamefont {Palma}(2019)}]{Achucarro:2018vey}%
  \BibitemOpen
  \bibfield  {author} {\bibinfo {author} {\bibfnamefont {A.}~\bibnamefont {Ach\'ucarro}}\ and\ \bibinfo {author} {\bibfnamefont {G.~A.}\ \bibnamefont {Palma}},\ }\href {\doibase 10.1088/1475-7516/2019/02/041} {\bibfield  {journal} {\bibinfo  {journal} {JCAP}\ }\textbf {\bibinfo {volume} {02}},\ \bibinfo {pages} {041} (\bibinfo {year} {2019})},\ \Eprint {http://arxiv.org/abs/1807.04390} {arXiv:1807.04390 [hep-th]} \BibitemShut {NoStop}%
\bibitem [{\citenamefont {Aragam}\ \emph {et~al.}(2020)\citenamefont {Aragam}, \citenamefont {Paban},\ and\ \citenamefont {Rosati}}]{Aragam:2019omo}%
  \BibitemOpen
  \bibfield  {author} {\bibinfo {author} {\bibfnamefont {V.}~\bibnamefont {Aragam}}, \bibinfo {author} {\bibfnamefont {S.}~\bibnamefont {Paban}}, \ and\ \bibinfo {author} {\bibfnamefont {R.}~\bibnamefont {Rosati}},\ }\href {\doibase 10.1088/1475-7516/2020/04/022} {\bibfield  {journal} {\bibinfo  {journal} {JCAP}\ }\textbf {\bibinfo {volume} {04}},\ \bibinfo {pages} {022} (\bibinfo {year} {2020})},\ \Eprint {http://arxiv.org/abs/1905.07495} {arXiv:1905.07495 [hep-th]} \BibitemShut {NoStop}%
\bibitem [{\citenamefont {Giddings}\ \emph {et~al.}(2002)\citenamefont {Giddings}, \citenamefont {Kachru},\ and\ \citenamefont {Polchinski}}]{Giddings:2001yu}%
  \BibitemOpen
  \bibfield  {author} {\bibinfo {author} {\bibfnamefont {S.~B.}\ \bibnamefont {Giddings}}, \bibinfo {author} {\bibfnamefont {S.}~\bibnamefont {Kachru}}, \ and\ \bibinfo {author} {\bibfnamefont {J.}~\bibnamefont {Polchinski}},\ }\href {\doibase 10.1103/PhysRevD.66.106006} {\bibfield  {journal} {\bibinfo  {journal} {Phys. Rev.}\ }\textbf {\bibinfo {volume} {D66}},\ \bibinfo {pages} {106006} (\bibinfo {year} {2002})},\ \Eprint {http://arxiv.org/abs/hep-th/0105097} {arXiv:hep-th/0105097 [hep-th]} \BibitemShut {NoStop}%
\bibitem [{\citenamefont {Grimm}\ and\ \citenamefont {Louis}(2004)}]{Grimm:2004uq}%
  \BibitemOpen
  \bibfield  {author} {\bibinfo {author} {\bibfnamefont {T.~W.}\ \bibnamefont {Grimm}}\ and\ \bibinfo {author} {\bibfnamefont {J.}~\bibnamefont {Louis}},\ }\href {\doibase 10.1016/j.nuclphysb.2004.08.005} {\bibfield  {journal} {\bibinfo  {journal} {Nucl. Phys.}\ }\textbf {\bibinfo {volume} {B699}},\ \bibinfo {pages} {387} (\bibinfo {year} {2004})},\ \Eprint {http://arxiv.org/abs/hep-th/0403067} {arXiv:hep-th/0403067 [hep-th]} \BibitemShut {NoStop}%
\bibitem [{\citenamefont {Sonner}\ and\ \citenamefont {Townsend}(2006)}]{Sonner:2006yn}%
  \BibitemOpen
  \bibfield  {author} {\bibinfo {author} {\bibfnamefont {J.}~\bibnamefont {Sonner}}\ and\ \bibinfo {author} {\bibfnamefont {P.~K.}\ \bibnamefont {Townsend}},\ }\href {\doibase 10.1103/PhysRevD.74.103508} {\bibfield  {journal} {\bibinfo  {journal} {Phys. Rev. D}\ }\textbf {\bibinfo {volume} {74}},\ \bibinfo {pages} {103508} (\bibinfo {year} {2006})},\ \Eprint {http://arxiv.org/abs/hep-th/0608068} {arXiv:hep-th/0608068} \BibitemShut {NoStop}%
\bibitem [{\citenamefont {Cicoli}\ \emph {et~al.}(2020{\natexlab{a}})\citenamefont {Cicoli}, \citenamefont {Dibitetto},\ and\ \citenamefont {Pedro}}]{Cicoli:2020cfj}%
  \BibitemOpen
  \bibfield  {author} {\bibinfo {author} {\bibfnamefont {M.}~\bibnamefont {Cicoli}}, \bibinfo {author} {\bibfnamefont {G.}~\bibnamefont {Dibitetto}}, \ and\ \bibinfo {author} {\bibfnamefont {F.~G.}\ \bibnamefont {Pedro}},\ }\href {\doibase 10.1103/PhysRevD.101.103524} {\bibfield  {journal} {\bibinfo  {journal} {Phys. Rev. D}\ }\textbf {\bibinfo {volume} {101}},\ \bibinfo {pages} {103524} (\bibinfo {year} {2020}{\natexlab{a}})},\ \Eprint {http://arxiv.org/abs/2002.02695} {arXiv:2002.02695 [gr-qc]} \BibitemShut {NoStop}%
\bibitem [{\citenamefont {Cicoli}\ \emph {et~al.}(2020{\natexlab{b}})\citenamefont {Cicoli}, \citenamefont {Dibitetto},\ and\ \citenamefont {Pedro}}]{Cicoli:2020noz}%
  \BibitemOpen
  \bibfield  {author} {\bibinfo {author} {\bibfnamefont {M.}~\bibnamefont {Cicoli}}, \bibinfo {author} {\bibfnamefont {G.}~\bibnamefont {Dibitetto}}, \ and\ \bibinfo {author} {\bibfnamefont {F.~G.}\ \bibnamefont {Pedro}},\ }\href {\doibase 10.1007/JHEP10(2020)035} {\bibfield  {journal} {\bibinfo  {journal} {JHEP}\ }\textbf {\bibinfo {volume} {10}},\ \bibinfo {pages} {035} (\bibinfo {year} {2020}{\natexlab{b}})},\ \Eprint {http://arxiv.org/abs/2007.11011} {arXiv:2007.11011 [hep-th]} \BibitemShut {NoStop}%
\bibitem [{\citenamefont {Russo}\ and\ \citenamefont {Townsend}(2022)}]{Russo:2022pgo}%
  \BibitemOpen
  \bibfield  {author} {\bibinfo {author} {\bibfnamefont {J.~G.}\ \bibnamefont {Russo}}\ and\ \bibinfo {author} {\bibfnamefont {P.~K.}\ \bibnamefont {Townsend}},\ }\href {\doibase 10.1007/JHEP06(2022)001} {\bibfield  {journal} {\bibinfo  {journal} {JHEP}\ }\textbf {\bibinfo {volume} {06}},\ \bibinfo {pages} {001} (\bibinfo {year} {2022})},\ \Eprint {http://arxiv.org/abs/2203.09398} {arXiv:2203.09398 [hep-th]} \BibitemShut {NoStop}%
\bibitem [{\citenamefont {Brinkmann}\ \emph {et~al.}(2022)\citenamefont {Brinkmann}, \citenamefont {Cicoli}, \citenamefont {Dibitetto},\ and\ \citenamefont {Pedro}}]{Brinkmann:2022oxy}%
  \BibitemOpen
  \bibfield  {author} {\bibinfo {author} {\bibfnamefont {M.}~\bibnamefont {Brinkmann}}, \bibinfo {author} {\bibfnamefont {M.}~\bibnamefont {Cicoli}}, \bibinfo {author} {\bibfnamefont {G.}~\bibnamefont {Dibitetto}}, \ and\ \bibinfo {author} {\bibfnamefont {F.~G.}\ \bibnamefont {Pedro}},\ }\href {\doibase 10.1007/JHEP11(2022)044} {\bibfield  {journal} {\bibinfo  {journal} {JHEP}\ }\textbf {\bibinfo {volume} {11}},\ \bibinfo {pages} {044} (\bibinfo {year} {2022})},\ \Eprint {http://arxiv.org/abs/2206.10649} {arXiv:2206.10649 [hep-th]} \BibitemShut {NoStop}%
\bibitem [{\citenamefont {Grimm}\ \emph {et~al.}(2020)\citenamefont {Grimm}, \citenamefont {Li},\ and\ \citenamefont {Valenzuela}}]{Grimm:2019ixq}%
  \BibitemOpen
  \bibfield  {author} {\bibinfo {author} {\bibfnamefont {T.~W.}\ \bibnamefont {Grimm}}, \bibinfo {author} {\bibfnamefont {C.}~\bibnamefont {Li}}, \ and\ \bibinfo {author} {\bibfnamefont {I.}~\bibnamefont {Valenzuela}},\ }\href {\doibase 10.1007/JHEP06(2020)009} {\bibfield  {journal} {\bibinfo  {journal} {JHEP}\ }\textbf {\bibinfo {volume} {06}},\ \bibinfo {pages} {009} (\bibinfo {year} {2020})},\ \bibinfo {note} {[Erratum: JHEP 01, 007 (2021)]},\ \Eprint {http://arxiv.org/abs/1910.09549} {arXiv:1910.09549 [hep-th]} \BibitemShut {NoStop}%
\bibitem [{\citenamefont {Shiu}\ \emph {et~al.}()\citenamefont {Shiu}, \citenamefont {Tonioni},\ and\ \citenamefont {Tran}}]{criticalpoints2}%
  \BibitemOpen
  \bibfield  {author} {\bibinfo {author} {\bibfnamefont {G.}~\bibnamefont {Shiu}}, \bibinfo {author} {\bibfnamefont {F.}~\bibnamefont {Tonioni}}, \ and\ \bibinfo {author} {\bibfnamefont {H.}~\bibnamefont {Tran}},\ }\href@noop {} {\emph {\bibinfo {title} {{In progress (2023)}}}}\BibitemShut {NoStop}%
\bibitem [{\citenamefont {Andriot}\ \emph {et~al.}(2023)\citenamefont {Andriot}, \citenamefont {Horer},\ and\ \citenamefont {Tringas}}]{Andriot:2022brg}%
  \BibitemOpen
  \bibfield  {author} {\bibinfo {author} {\bibfnamefont {D.}~\bibnamefont {Andriot}}, \bibinfo {author} {\bibfnamefont {L.}~\bibnamefont {Horer}}, \ and\ \bibinfo {author} {\bibfnamefont {G.}~\bibnamefont {Tringas}},\ }\href {\doibase 10.1007/JHEP04(2023)139} {\bibfield  {journal} {\bibinfo  {journal} {JHEP}\ }\textbf {\bibinfo {volume} {04}},\ \bibinfo {pages} {139} (\bibinfo {year} {2023})},\ \Eprint {http://arxiv.org/abs/2212.04517} {arXiv:2212.04517 [hep-th]} \BibitemShut {NoStop}%
\bibitem [{\citenamefont {Hetz}\ and\ \citenamefont {Palma}(2016)}]{Hetz:2016ics}%
  \BibitemOpen
  \bibfield  {author} {\bibinfo {author} {\bibfnamefont {A.}~\bibnamefont {Hetz}}\ and\ \bibinfo {author} {\bibfnamefont {G.~A.}\ \bibnamefont {Palma}},\ }\href {\doibase 10.1103/PhysRevLett.117.101301} {\bibfield  {journal} {\bibinfo  {journal} {Phys. Rev. Lett.}\ }\textbf {\bibinfo {volume} {117}},\ \bibinfo {pages} {101301} (\bibinfo {year} {2016})},\ \Eprint {http://arxiv.org/abs/1601.05457} {arXiv:1601.05457 [hep-th]} \BibitemShut {NoStop}%
\bibitem [{\citenamefont {Polchinski}(1995)}]{Polchinski:1995mt}%
  \BibitemOpen
  \bibfield  {author} {\bibinfo {author} {\bibfnamefont {J.}~\bibnamefont {Polchinski}},\ }\href {\doibase 10.1103/PhysRevLett.75.4724} {\bibfield  {journal} {\bibinfo  {journal} {Phys. Rev. Lett.}\ }\textbf {\bibinfo {volume} {75}},\ \bibinfo {pages} {4724} (\bibinfo {year} {1995})},\ \Eprint {http://arxiv.org/abs/hep-th/9510017} {arXiv:hep-th/9510017} \BibitemShut {NoStop}%
\bibitem [{\citenamefont {Grimm}\ and\ \citenamefont {Louis}(2005)}]{Grimm:2004ua}%
  \BibitemOpen
  \bibfield  {author} {\bibinfo {author} {\bibfnamefont {T.~W.}\ \bibnamefont {Grimm}}\ and\ \bibinfo {author} {\bibfnamefont {J.}~\bibnamefont {Louis}},\ }\href {\doibase 10.1016/j.nuclphysb.2005.04.007} {\bibfield  {journal} {\bibinfo  {journal} {Nucl. Phys. B}\ }\textbf {\bibinfo {volume} {718}},\ \bibinfo {pages} {153} (\bibinfo {year} {2005})},\ \Eprint {http://arxiv.org/abs/hep-th/0412277} {arXiv:hep-th/0412277} \BibitemShut {NoStop}%
\bibitem [{\citenamefont {Cheung}\ and\ \citenamefont {Remmen}(2014)}]{Cheung:2014vva}%
  \BibitemOpen
  \bibfield  {author} {\bibinfo {author} {\bibfnamefont {C.}~\bibnamefont {Cheung}}\ and\ \bibinfo {author} {\bibfnamefont {G.~N.}\ \bibnamefont {Remmen}},\ }\href {\doibase 10.1103/PhysRevLett.113.051601} {\bibfield  {journal} {\bibinfo  {journal} {Phys. Rev. Lett.}\ }\textbf {\bibinfo {volume} {113}},\ \bibinfo {pages} {051601} (\bibinfo {year} {2014})},\ \Eprint {http://arxiv.org/abs/1402.2287} {arXiv:1402.2287 [hep-ph]} \BibitemShut {NoStop}%
\bibitem [{\citenamefont {Rudelius}(2015{\natexlab{a}})}]{Rudelius:2014wla}%
  \BibitemOpen
  \bibfield  {author} {\bibinfo {author} {\bibfnamefont {T.}~\bibnamefont {Rudelius}},\ }\href {\doibase 10.1088/1475-7516/2015/04/049} {\bibfield  {journal} {\bibinfo  {journal} {JCAP}\ }\textbf {\bibinfo {volume} {04}},\ \bibinfo {pages} {049} (\bibinfo {year} {2015}{\natexlab{a}})},\ \Eprint {http://arxiv.org/abs/1409.5793} {arXiv:1409.5793 [hep-th]} \BibitemShut {NoStop}%
\bibitem [{\citenamefont {Rudelius}(2015{\natexlab{b}})}]{Rudelius:2015xta}%
  \BibitemOpen
  \bibfield  {author} {\bibinfo {author} {\bibfnamefont {T.}~\bibnamefont {Rudelius}},\ }\href {\doibase 10.1088/1475-7516/2015/9/020} {\bibfield  {journal} {\bibinfo  {journal} {JCAP}\ }\textbf {\bibinfo {volume} {09}},\ \bibinfo {pages} {020} (\bibinfo {year} {2015}{\natexlab{b}})},\ \Eprint {http://arxiv.org/abs/1503.00795} {arXiv:1503.00795 [hep-th]} \BibitemShut {NoStop}%
\bibitem [{\citenamefont {Brown}\ \emph {et~al.}(2015)\citenamefont {Brown}, \citenamefont {Cottrell}, \citenamefont {Shiu},\ and\ \citenamefont {Soler}}]{Brown:2015iha}%
  \BibitemOpen
  \bibfield  {author} {\bibinfo {author} {\bibfnamefont {J.}~\bibnamefont {Brown}}, \bibinfo {author} {\bibfnamefont {W.}~\bibnamefont {Cottrell}}, \bibinfo {author} {\bibfnamefont {G.}~\bibnamefont {Shiu}}, \ and\ \bibinfo {author} {\bibfnamefont {P.}~\bibnamefont {Soler}},\ }\href {\doibase 10.1007/JHEP10(2015)023} {\bibfield  {journal} {\bibinfo  {journal} {JHEP}\ }\textbf {\bibinfo {volume} {10}},\ \bibinfo {pages} {023} (\bibinfo {year} {2015})},\ \Eprint {http://arxiv.org/abs/1503.04783} {arXiv:1503.04783 [hep-th]} \BibitemShut {NoStop}%
\bibitem [{\citenamefont {Brown}\ \emph {et~al.}(2016)\citenamefont {Brown}, \citenamefont {Cottrell}, \citenamefont {Shiu},\ and\ \citenamefont {Soler}}]{Brown:2015lia}%
  \BibitemOpen
  \bibfield  {author} {\bibinfo {author} {\bibfnamefont {J.}~\bibnamefont {Brown}}, \bibinfo {author} {\bibfnamefont {W.}~\bibnamefont {Cottrell}}, \bibinfo {author} {\bibfnamefont {G.}~\bibnamefont {Shiu}}, \ and\ \bibinfo {author} {\bibfnamefont {P.}~\bibnamefont {Soler}},\ }\href {\doibase 10.1007/JHEP04(2016)017} {\bibfield  {journal} {\bibinfo  {journal} {JHEP}\ }\textbf {\bibinfo {volume} {04}},\ \bibinfo {pages} {017} (\bibinfo {year} {2016})},\ \Eprint {http://arxiv.org/abs/1504.00659} {arXiv:1504.00659 [hep-th]} \BibitemShut {NoStop}%
\bibitem [{\citenamefont {Palti}(2017)}]{Palti:2017elp}%
  \BibitemOpen
  \bibfield  {author} {\bibinfo {author} {\bibfnamefont {E.}~\bibnamefont {Palti}},\ }\href {\doibase 10.1007/JHEP08(2017)034} {\bibfield  {journal} {\bibinfo  {journal} {JHEP}\ }\textbf {\bibinfo {volume} {08}},\ \bibinfo {pages} {034} (\bibinfo {year} {2017})},\ \Eprint {http://arxiv.org/abs/1705.04328} {arXiv:1705.04328 [hep-th]} \BibitemShut {NoStop}%
\bibitem [{\citenamefont {Calder\'on-Infante}\ \emph {et~al.}(2021)\citenamefont {Calder\'on-Infante}, \citenamefont {Uranga},\ and\ \citenamefont {Valenzuela}}]{Calderon-Infante:2020dhm}%
  \BibitemOpen
  \bibfield  {author} {\bibinfo {author} {\bibfnamefont {J.}~\bibnamefont {Calder\'on-Infante}}, \bibinfo {author} {\bibfnamefont {A.~M.}\ \bibnamefont {Uranga}}, \ and\ \bibinfo {author} {\bibfnamefont {I.}~\bibnamefont {Valenzuela}},\ }\href {\doibase 10.1007/JHEP03(2021)299} {\bibfield  {journal} {\bibinfo  {journal} {JHEP}\ }\textbf {\bibinfo {volume} {03}},\ \bibinfo {pages} {299} (\bibinfo {year} {2021})},\ \Eprint {http://arxiv.org/abs/2012.00034} {arXiv:2012.00034 [hep-th]} \BibitemShut {NoStop}%
\bibitem [{\citenamefont {Etheredge}\ \emph {et~al.}(2022)\citenamefont {Etheredge}, \citenamefont {Heidenreich}, \citenamefont {Kaya}, \citenamefont {Qiu},\ and\ \citenamefont {Rudelius}}]{Etheredge:2022opl}%
  \BibitemOpen
  \bibfield  {author} {\bibinfo {author} {\bibfnamefont {M.}~\bibnamefont {Etheredge}}, \bibinfo {author} {\bibfnamefont {B.}~\bibnamefont {Heidenreich}}, \bibinfo {author} {\bibfnamefont {S.}~\bibnamefont {Kaya}}, \bibinfo {author} {\bibfnamefont {Y.}~\bibnamefont {Qiu}}, \ and\ \bibinfo {author} {\bibfnamefont {T.}~\bibnamefont {Rudelius}},\ }\href {\doibase 10.1007/JHEP12(2022)114} {\bibfield  {journal} {\bibinfo  {journal} {JHEP}\ }\textbf {\bibinfo {volume} {12}},\ \bibinfo {pages} {114} (\bibinfo {year} {2022})},\ \Eprint {http://arxiv.org/abs/2206.04063} {arXiv:2206.04063 [hep-th]} \BibitemShut {NoStop}%
\bibitem [{\citenamefont {Arkani-Hamed}\ \emph {et~al.}(2007)\citenamefont {Arkani-Hamed}, \citenamefont {Motl}, \citenamefont {Nicolis},\ and\ \citenamefont {Vafa}}]{Arkani-Hamed:2006emk}%
  \BibitemOpen
  \bibfield  {author} {\bibinfo {author} {\bibfnamefont {N.}~\bibnamefont {Arkani-Hamed}}, \bibinfo {author} {\bibfnamefont {L.}~\bibnamefont {Motl}}, \bibinfo {author} {\bibfnamefont {A.}~\bibnamefont {Nicolis}}, \ and\ \bibinfo {author} {\bibfnamefont {C.}~\bibnamefont {Vafa}},\ }\href {\doibase 10.1088/1126-6708/2007/06/060} {\bibfield  {journal} {\bibinfo  {journal} {JHEP}\ }\textbf {\bibinfo {volume} {06}},\ \bibinfo {pages} {060} (\bibinfo {year} {2007})},\ \Eprint {http://arxiv.org/abs/hep-th/0601001} {arXiv:hep-th/0601001} \BibitemShut {NoStop}%
\bibitem [{\citenamefont {Ooguri}\ and\ \citenamefont {Vafa}(2007)}]{Ooguri:2006in}%
  \BibitemOpen
  \bibfield  {author} {\bibinfo {author} {\bibfnamefont {H.}~\bibnamefont {Ooguri}}\ and\ \bibinfo {author} {\bibfnamefont {C.}~\bibnamefont {Vafa}},\ }\href {\doibase 10.1016/j.nuclphysb.2006.10.033} {\bibfield  {journal} {\bibinfo  {journal} {Nucl. Phys. B}\ }\textbf {\bibinfo {volume} {766}},\ \bibinfo {pages} {21} (\bibinfo {year} {2007})},\ \Eprint {http://arxiv.org/abs/hep-th/0605264} {arXiv:hep-th/0605264} \BibitemShut {NoStop}%
\bibitem [{\citenamefont {Ginsparg}\ and\ \citenamefont {Vafa}(1987)}]{Ginsparg:1986wr}%
  \BibitemOpen
  \bibfield  {author} {\bibinfo {author} {\bibfnamefont {P.~H.}\ \bibnamefont {Ginsparg}}\ and\ \bibinfo {author} {\bibfnamefont {C.}~\bibnamefont {Vafa}},\ }\href {\doibase 10.1016/0550-3213(87)90387-7} {\bibfield  {journal} {\bibinfo  {journal} {Nucl. Phys. B}\ }\textbf {\bibinfo {volume} {289}},\ \bibinfo {pages} {414} (\bibinfo {year} {1987})}\BibitemShut {NoStop}%
\bibitem [{\citenamefont {Baykara}\ \emph {et~al.}(2022)\citenamefont {Baykara}, \citenamefont {Robbins},\ and\ \citenamefont {Sethi}}]{Baykara:2022cwj}%
  \BibitemOpen
  \bibfield  {author} {\bibinfo {author} {\bibfnamefont {Z.~K.}\ \bibnamefont {Baykara}}, \bibinfo {author} {\bibfnamefont {D.}~\bibnamefont {Robbins}}, \ and\ \bibinfo {author} {\bibfnamefont {S.}~\bibnamefont {Sethi}},\ }\href@noop {} {\  (\bibinfo {year} {2022})},\ \Eprint {http://arxiv.org/abs/2212.02557} {arXiv:2212.02557 [hep-th]} \BibitemShut {NoStop}%
\bibitem [{\citenamefont {Alvarez-Gaume}\ \emph {et~al.}(1986)\citenamefont {Alvarez-Gaume}, \citenamefont {Ginsparg}, \citenamefont {Moore},\ and\ \citenamefont {Vafa}}]{AlvarezGaume:1986jb}%
  \BibitemOpen
  \bibfield  {author} {\bibinfo {author} {\bibfnamefont {L.}~\bibnamefont {Alvarez-Gaume}}, \bibinfo {author} {\bibfnamefont {P.~H.}\ \bibnamefont {Ginsparg}}, \bibinfo {author} {\bibfnamefont {G.~W.}\ \bibnamefont {Moore}}, \ and\ \bibinfo {author} {\bibfnamefont {C.}~\bibnamefont {Vafa}},\ }\href {\doibase 10.1016/0370-2693(86)91524-8} {\bibfield  {journal} {\bibinfo  {journal} {Phys. Lett.}\ }\textbf {\bibinfo {volume} {B171}},\ \bibinfo {pages} {155} (\bibinfo {year} {1986})}\BibitemShut {NoStop}%
\bibitem [{\citenamefont {Dixon}\ \emph {et~al.}(1991)\citenamefont {Dixon}, \citenamefont {Kaplunovsky},\ and\ \citenamefont {Louis}}]{Dixon:1990pc}%
  \BibitemOpen
  \bibfield  {author} {\bibinfo {author} {\bibfnamefont {L.~J.}\ \bibnamefont {Dixon}}, \bibinfo {author} {\bibfnamefont {V.}~\bibnamefont {Kaplunovsky}}, \ and\ \bibinfo {author} {\bibfnamefont {J.}~\bibnamefont {Louis}},\ }\href {\doibase 10.1016/0550-3213(91)90490-O} {\bibfield  {journal} {\bibinfo  {journal} {Nucl. Phys. B}\ }\textbf {\bibinfo {volume} {355}},\ \bibinfo {pages} {649} (\bibinfo {year} {1991})}\BibitemShut {NoStop}%
\bibitem [{\citenamefont {Hertzberg}\ \emph {et~al.}(2007)\citenamefont {Hertzberg}, \citenamefont {Kachru}, \citenamefont {Taylor},\ and\ \citenamefont {Tegmark}}]{Hertzberg:2007wc}%
  \BibitemOpen
  \bibfield  {author} {\bibinfo {author} {\bibfnamefont {M.~P.}\ \bibnamefont {Hertzberg}}, \bibinfo {author} {\bibfnamefont {S.}~\bibnamefont {Kachru}}, \bibinfo {author} {\bibfnamefont {W.}~\bibnamefont {Taylor}}, \ and\ \bibinfo {author} {\bibfnamefont {M.}~\bibnamefont {Tegmark}},\ }\href {\doibase 10.1088/1126-6708/2007/12/095} {\bibfield  {journal} {\bibinfo  {journal} {JHEP}\ }\textbf {\bibinfo {volume} {12}},\ \bibinfo {pages} {095} (\bibinfo {year} {2007})},\ \Eprint {http://arxiv.org/abs/0711.2512} {arXiv:0711.2512 [hep-th]} \BibitemShut {NoStop}%
\bibitem [{\citenamefont {Andriot}\ \emph {et~al.}(2020)\citenamefont {Andriot}, \citenamefont {Cribiori},\ and\ \citenamefont {Erkinger}}]{Andriot:2020lea}%
  \BibitemOpen
  \bibfield  {author} {\bibinfo {author} {\bibfnamefont {D.}~\bibnamefont {Andriot}}, \bibinfo {author} {\bibfnamefont {N.}~\bibnamefont {Cribiori}}, \ and\ \bibinfo {author} {\bibfnamefont {D.}~\bibnamefont {Erkinger}},\ }\href {\doibase 10.1007/JHEP07(2020)162} {\bibfield  {journal} {\bibinfo  {journal} {JHEP}\ }\textbf {\bibinfo {volume} {07}},\ \bibinfo {pages} {162} (\bibinfo {year} {2020})},\ \Eprint {http://arxiv.org/abs/2004.00030} {arXiv:2004.00030 [hep-th]} \BibitemShut {NoStop}%
\bibitem [{\citenamefont {Andriot}(2021)}]{Andriot:2021rdy}%
  \BibitemOpen
  \bibfield  {author} {\bibinfo {author} {\bibfnamefont {D.}~\bibnamefont {Andriot}},\ }\href {\doibase 10.1002/prop.202100063} {\bibfield  {journal} {\bibinfo  {journal} {Fortsch. Phys.}\ }\textbf {\bibinfo {volume} {69}},\ \bibinfo {pages} {2100063} (\bibinfo {year} {2021})},\ \Eprint {http://arxiv.org/abs/2101.06251} {arXiv:2101.06251 [hep-th]} \BibitemShut {NoStop}%
\bibitem [{\citenamefont {Andriot}\ and\ \citenamefont {Horer}(2023)}]{Andriot:2022xjh}%
  \BibitemOpen
  \bibfield  {author} {\bibinfo {author} {\bibfnamefont {D.}~\bibnamefont {Andriot}}\ and\ \bibinfo {author} {\bibfnamefont {L.}~\bibnamefont {Horer}},\ }\href {\doibase 10.1007/JHEP01(2023)020} {\bibfield  {journal} {\bibinfo  {journal} {JHEP}\ }\textbf {\bibinfo {volume} {01}},\ \bibinfo {pages} {020} (\bibinfo {year} {2023})},\ \Eprint {http://arxiv.org/abs/2208.14462} {arXiv:2208.14462 [hep-th]} \BibitemShut {NoStop}%
\end{thebibliography}%

\end{document}